\documentclass[a4paper,UKenglish,cleveref, autoref, thm-restate]{lipics-v2021}
\pdfoutput=1 
\hideLIPIcs  

\title{Spegion: Implicit and Non-Lexical Regions with Sized Allocations}

\author{Jack Hughes}{School of Computing, University of Kent, United Kingdom}{jh2093@kent.ac.uk}{https://orcid.org/0000-0003-3174-4689}{}

\author{Michael Vollmer}{School of Computing, University of Kent, United Kingdom}{m.vollmer@kent.ac.uk}{https://orcid.org/0000-0002-0496-8268}{}

\author{Mark Batty}{School of Computing, Univeristy of Kent, United Kingdom}{m.j.batty@kent.ac.uk}{https://orcid.org/0000-0001-7053-4364}{}

\authorrunning{J. Hughes, M. Vollmer, and M. Batty} 

\Copyright{Jack Hughes, Micheal Vollmer, and Mark Batty} 

\ccsdesc{Software and its engineering~Allocation / deallocation strategies}
\ccsdesc{Theory of computation~Type theory}

\keywords{Regions, Type Systems, Effect Systems, Programming Languages, Memory} 

\funding{
This work is supported by the
\href{https://www.ukri.org/councils/epsrc/}{EPSRC} under grants
\textit{EP/X021173/1} and
\textit{EP/V000470/1}.
}

\nolinenumbers 

\EventEditors{Jonathan Aldrich and Alexandra Silva}
\EventNoEds{2}
\EventLongTitle{39th European Conference on Object-Oriented Programming (ECOOP 2025)}
\EventShortTitle{ECOOP 2025}
\EventAcronym{ECOOP}
\EventYear{2025}
\EventDate{June 30--July 2, 2025}
\EventLocation{Bergen, Norway}
\EventLogo{}
\SeriesVolume{333}
\ArticleNo{40}

\usepackage{enumerate}
\usepackage{graphicx}
\usepackage{hyperref}
\usepackage{amsmath}
\usepackage{amssymb}
\usepackage{amsfonts}
\usepackage{amsthm}
\usepackage{mathpartir}
\usepackage{cite}
\usepackage[section]{placeins}
\usepackage{tikz}
\usepackage{tikz-cd}
\usepackage{stmaryrd}
\usepackage{leftindex}
\usepackage{fancyhdr}
\usepackage{pdflscape}
\usepackage{float}
\usepackage{xcolor}
\usepackage{thmtools}
\usepackage{thm-restate}
\usepackage{listings-rust}

\theoremstyle{definition}

\declaretheorem[name=Definition, numberwithin=section]{df}
\declaretheorem[name=Lemma, numberwithin=section]{lem}

\usepackage{listings}

\definecolor{mGreen}{rgb}{0,0.6,0}
\definecolor{mGray}{rgb}{0.5,0.5,0.5}
\definecolor{mPurple}{rgb}{0.58,0,0.82}
\definecolor{backgroundColour}{rgb}{0.95,0.95,0.92}

\lstdefinestyle{CStyle}{
    backgroundcolor=\color{backgroundColour},   
    commentstyle=\color{mGreen},
    keywordstyle=\color{magenta},
    numberstyle=\tiny\color{mGray},
    stringstyle=\color{mPurple},
    basicstyle=\footnotesize,
    breakatwhitespace=false,         
    breaklines=true,                 
    captionpos=b,                    
    keepspaces=true,                 
    numbers=left,                    
    numbersep=5pt,                  
    showspaces=false,                
    showstringspaces=false,
    showtabs=false,                  
    tabsize=2,
    language=C
}

\newcommand{\kVar}{
        \inferrule*[lab=$\kappa$-var]
                {\text{Type} = K(\alpha)}
                {K \vdash \alpha : \text{Type}}
}

\newcommand{\kForall}{
        \inferrule*[lab=$\kappa$-forall]
                {K, \alpha : \text{Type}, \rho : \text{Region}, \epsilon : \text{Effect} \vdash \tau : \text{Type} }
                {K \vdash \forall \{ \alpha, \rho, \epsilon \} . \tau : \text{Type}}
}

\newcommand{\kRegion}{
        \inferrule*[lab=$\kappa$-reg]
                {\text{Region} = K(\rho)}
                {K \vdash  \rho : \text{Region}}
}

\newcommand{\kApp}{
        \inferrule*[lab=$\kappa$-app$_{1}$]
                {K \vdash \kappa_1 : \text{Type} \qquad K \vdash \tau : \text{Type}}
                {K \vdash (\kappa_1 \rightarrow \kappa_2)\ \tau : \kappa_2}
}

\newcommand{\kAppTwo}{
        \inferrule*[lab=$\kappa$-app$_{2}$]
                {K \vdash \kappa_1 : \text{Effect} \qquad K \vdash \varphi : \text{Effect}}
                {K \vdash (\kappa_1 \rightarrow \kappa_2)\ \varphi : \kappa_2}
}

\newcommand{\kArrow}{
        \inferrule*[lab=$\kappa$-$\rightarrow$]
                {\quad}
                {K \vdash (\rightarrow) : \text{Type} \rightarrow \text{Effect} \rightarrow \text{Type} \rightarrow \text{Type}  }
}

\newcommand{\kInt}{
        \inferrule*[right=$\kappa$-int]
                {\quad}
                {K \vdash \text{Int} : \text{Type}}
}

\newcommand{\kUnit}{
        \inferrule*[right=$\kappa$-unit]
                {\quad}
                {K \vdash \text{Unit} : \text{Type} }
}

\newcommand{\kRef}{
        \inferrule*[right=$\kappa$-ref]
                {K \vdash \tau : \text{Type}}
                {K \vdash \text{Ref}\ \tau : \text{Type}}
}

\newcommand{\kTyWithPlace}{
        \inferrule*[lab=$\kappa$-tyWithPlace]
                {K \vdash \tau : \text{Type} \qquad K \vdash \rho : \text{Region}}
                {K \vdash (\tau, \rho) : \text{Type} }
}

\newcommand{\kCompose}{
        \inferrule*[right=$\kappa$-$\times$]
                {K \vdash \varphi_1 : \text{Effect} \\\\ K \vdash \varphi_2 : \text{Effect}}
                {K \vdash \varphi_1 \times \varphi_2 : \text{Effect}}
}

\newcommand{\kBot}{
        \inferrule*[right=$\kappa$-bot]
                {\quad}
                {K \vdash_\kappa \{ \bot \} : \text{Effect}}   
}

\newcommand{\kAlloc}{
        \inferrule*[right=$\kappa$-alloc]
                {K \vdash s : \text{Size} \qquad K \vdash \rho : \text{Region}}
                {K \vdash \{\textbf{alloc}\ s\ \rho\} : \text{Effect}}
}

\newcommand{\kNew}{
        \inferrule*[right=$\kappa$-fresh]
                {K \vdash \rho : \text{Region} \qquad K \vdash s : \text{Size}}
                {K \vdash \{\textbf{fresh}\ \rho\ s \} : \text{Effect}}
}

\newcommand{\kFree}{
        \inferrule*[lab=$\kappa$-free]
                {K \vdash \rho : \text{Region}}
                {K \vdash \{\textbf{free}\ \rho\} : \text{Effect}}
}

\newcommand{\kSize}{
        \inferrule*[right=$\kappa$-size]
                {\quad}
                {K \vdash s : \text{Size}}
}       

\newcommand{\kOp}{
        \inferrule*[right=$\kappa$-op]
                {K \vdash \kappa_1 : \text{Size} \qquad K \vdash \kappa_2 : \text{Size} }
                {K \vdash \kappa_1\ op\ \kappa_2 : \text{Size}}
}

\newcommand{\kSplit}{
        \inferrule*[right=$\kappa$-split]
                {K \vdash \rho : \text{Region} \qquad K \vdash s : \text{Size} \qquad K \vdash \rho' : \text{Region}}
                {K \vdash \{\textbf{split}\ \rho\ s\ \rho' \} : \text{Effect}}
}

\newcommand{\kEps}{
        \inferrule*[lab=$\kappa$-$\epsilon$]
                {\text{Effect} = K(\epsilon)}
                {K \vdash \{\epsilon\} : \text{Effect}}
}

\newcommand{\kRec}{
        \inferrule*[lab=$\kappa$-rec]
                {K \vdash \epsilon : \text{Effect} \qquad K \vdash \varphi : \text{Effect}  }     
                {K \vdash \{\textbf{rec}\ \epsilon\ \varphi\} : \text{Effect}}
}

\newcommand{\kBool}{
        \inferrule*[right=$\kappa$-bool]
                {\quad}
                {K \vdash \text{Bool} : \text{Type}}
}

\newcommand{\kJoin}{
        \inferrule*[right=$\kappa$-$\sqcup$]
                {K \vdash \varphi_1 : \text{Effect} \\\\ K \vdash \varphi_2 : \text{Effect}}
                {K \vdash \varphi_1 \sqcup \varphi_2 : \text{Effect}}
}

\newcommand{\effectEqRefl}{
        \inferrule*[right=$\equiv$-refl]
                {K \vdash \varphi : \text{Effect}}
                {K \vdash \varphi \equiv \varphi : \text{Effect}}
}

\newcommand{\effectEqSym}{
        \inferrule*[right=$\equiv$-sym]
                { K \vdash \varphi_1 : \text{Effect} \\\\ K \vdash \varphi_2 : \text{Effect}
                        \\\\ K \vdash \varphi_1 \equiv \varphi_2 : \text{Effect} }
                {K \vdash \varphi_2 \equiv \varphi_1 : \text{Effect}}
}

\newcommand{\effectEqTrans}{
        \inferrule*[right=$\equiv$-tr]
                {K \vdash \varphi_1 \equiv \varphi_2 : \text{Effect} \\\\ K \vdash \varphi_2 \equiv \varphi_3 : \text{Effect}}
                {K \vdash \varphi_1 \equiv \varphi_3 : \text{Effect}}
}

\newcommand{\effectEqCong}{
        \inferrule*[right=$\equiv$-$\times$-cong]
                {K \vdash \varphi_1 \equiv \varphi'_1 : \text{Effect} \qquad K \vdash \varphi_2 \equiv \varphi'_2 : \text{Effect}}
                {K \vdash \varphi_1 \times \varphi_2 \equiv \varphi'_1 \times \varphi'_2 : \text{Effect}}
}

\newcommand{\effectEqAssoc}{
        \inferrule*[right=$\equiv$-$\times$-assoc]
                {K \vdash \varphi_1 : \text{Effect} \qquad K \vdash \varphi_2 : \text{Effect} \qquad K \vdash \varphi_3 : \text{Effect}}
                {K \vdash \varphi_1 \times (\varphi_2 \times \varphi_3) \equiv (\varphi_1 \times \varphi_2) \times \varphi_3 : \text{Effect} }
}

\newcommand{\effectEqJoinAssoc}{
        \inferrule*[right=$\equiv$-$\sqcup$-assoc]                 
                {K \vdash \varphi_1 : \text{Effect} \qquad K \vdash \varphi_2 : \text{Effect} \qquad K \vdash \varphi_3 : \text{Effect}}
                {K \vdash \varphi_1 \sqcup (\varphi_2 \sqcup \varphi_3) \equiv (\varphi_1 \sqcup \varphi_2) \sqcup \varphi_3 : \text{Effect} }
}

\newcommand{\effectEqJoinCong}{
        \inferrule*[right=$\equiv$-$\sqcup$-cong]
                {K \vdash \varphi_1 \equiv \varphi'_1 : \text{Effect} \qquad K \vdash \varphi_2 \equiv \varphi'_2 : \text{Effect}}
                {K \vdash \varphi_1 \sqcup \varphi_2 \equiv \varphi'_1 \sqcup \varphi'_2 : \text{Effect}}
}

\newcommand{\sbEquiv}{
        \inferrule*[lab=sb-$\equiv$]
                {K \vdash \varphi_1 \equiv \varphi_2 : \text{Effect}}
                {K \vdash \varphi_1 \sqsubseteq \varphi_2 : \text{Effect}}
}

\newcommand{\sbBot}{
        \inferrule*[lab=sb-$\bot$]
                {K \vdash \varphi : \text{Effect}}
                {K \vdash \{ \bot \} \sqsubseteq \varphi}
}

\newcommand{\sbXAbove}{
        \inferrule*[lab=sb-$\times_{above}$]
                {K \vdash \varphi_1 \sqsubseteq \varphi_3 : \text{Effect} \\\\ K \vdash \varphi_2 \sqsubseteq \varphi_3 : \text{Effect}}
                {K \vdash \varphi_1 \times \varphi_2 \sqsubseteq \varphi_3 : \text{Effect}}
}

\newcommand{\sbXBelowOne}{
        \inferrule*[lab=sb-$\times_{below}^{1}$]
                {K \vdash \varphi_1 \sqsubseteq \varphi_2 : \text{Effect} \qquad K \vdash \varphi_3 : \text{Effect}}
                {K \vdash \varphi_1 \sqsubseteq \varphi_2 \times \varphi_3 : \text{Effect}}
}

\newcommand{\sbXBelowTwo}{
        \inferrule*[lab=sb-$\times_{below}^{2}$]
                {K \vdash \varphi_2 : \text{Effect} \qquad K \vdash \varphi_1 \sqsubseteq \varphi_3 : \text{Effect}}
                {K \vdash \varphi_1 \sqsubseteq \varphi_2 \times \varphi_3 : \text{Effect}}
}

\newcommand{\sbJoinAbove}{
        \inferrule*[lab=sb-$\sqcup_{above}$]
                {K \vdash \varphi_1 \sqsubseteq \varphi_3 : \text{Effect} \\\\ K \vdash \varphi_2 \sqsubseteq \varphi_3 : \text{Effect}}
                {K \vdash \varphi_1 \sqcup \varphi_2 \sqsubseteq \varphi_3 : \text{Effect}}
}

\newcommand{\sbJoinBelowOne}{
        \inferrule*[lab=sb-$\sqcup_{below}^{1}$]
                {K \vdash \varphi_1 \sqsubseteq \varphi_2 : \text{Effect} \\\\ K \vdash \varphi_3 : \text{Effect}}
                {K \vdash \varphi_1 \sqsubseteq \varphi_2 \sqcup \varphi_3 : \text{Effect}}
}

\newcommand{\sbJoinBelowTwo}{
        \inferrule*[lab=sb-$\sqcup_{below}^{2}$]
                {K \vdash \varphi_2 : \text{Effect} \\\\ K \vdash \varphi_1 \sqsubseteq \varphi_3 : \text{Effect}}
                {K \vdash \varphi_1 \sqsubseteq \varphi_2 \sqcup \varphi_3 : \text{Effect}}
}

\newcommand{\intT}{
        \inferrule*[right=t-int]
                {\quad}
                {K \mid \Gamma \mid \Sigma \vdash c : \text{Int}}
        }

\newcommand{\trueT}{
        \inferrule*[right=t-true]
                {\quad}
                {K \mid \Gamma \mid \Sigma \vdash \text{true} : \text{Bool}}
}

\newcommand{\falseT}{
        \inferrule*[right=t-false]
                {\quad}
                {K \mid \Gamma \mid \Sigma \vdash \text{false} : \text{Bool}}
}

\newcommand{\absT}{
        \inferrule*[right=t-$\lambda$]
                {K \vdash \mu_1 : \text{Type} \quad K \mid \Gamma, x : \mu_1 \mid \Sigma \vdash e : \mu_2 \mid \varphi }
                {K \mid \Gamma \mid \Sigma \vdash \lambda x . e : \mu_1 \xrightarrow{\varphi} \mu_2}
        }

\newcommand{\unitT}{
        \inferrule*[right=t-unit]
                {\quad}
                {K \mid \Gamma \mid \Sigma \vdash () : \text{Unit}}
}

\newcommand{\locValT}{
        \inferrule*[right=t-loc]
                {\tau = \Sigma(l_\rho) \qquad K \vdash \tau : \text{Type}}
                {K \mid \Gamma \mid \Sigma \vdash l_\rho : \tau}
}

\newcommand{\biglamT}{
        \inferrule*[right=t-$\Lambda$]
                {K, \alpha : \text{Type}, \rho : \text{Region}, \epsilon : \text{Effect} \mid \Gamma, x : \mu_1 \mid \Sigma \vdash e : \mu_2 \mid \varphi}
                {K \mid \Gamma \mid \Sigma \vdash \Lambda \{\alpha, \rho, \epsilon\} . \lambda x . e : \forall \{\alpha : \text{Type}, \rho : \text{Region}, \epsilon : \text{Effect} \} . \mu_1 \xrightarrow{ \varphi } \mu_2}
}

\newcommand{\uselocT}{
        \inferrule*[lab=t-use-val]
                {K \mid \Gamma \mid \Sigma \vdash l_\rho : \tau }
                {K \mid \Gamma \mid \Sigma \vdash l_\rho : (\tau, \rho) \mid \{ \bot \}}
}

\newcommand{\newrgnT}{
        \inferrule*[right=t-newrgn]
                {s \sqsupseteq 1 \qquad K, \rho : \text{Region} \vdash s : \text{Size}}
                {K \mid \Gamma \mid \Sigma \vdash \textbf{newrgn}\ [s] : (\text{Unit}, \rho) \mid  \{\textbf{fresh}\ \rho\ s \} \times \{\textbf{alloc}\ 1\ \rho \} }
}

\newcommand{\freergnT}{
        \inferrule*[right=t-freergn]
                {K \mid \Gamma \mid \Sigma \vdash e : (\tau, \rho) \mid \varphi}
                {K \mid \Gamma \mid \Sigma \vdash \textbf{freergn}\ e : (\text{Unit}, \rho_\text{glob}) \mid  \varphi \times \{\textbf{free}\ \rho\} }
}

\newcommand{\splitT}{
        \inferrule*[right=t-split]
                {K, \rho' : \text{Region} \mid \Gamma \mid \Sigma \vdash e : (\tau, \rho) \mid \varphi \qquad s \sqsupseteq 1 \qquad K \vdash s : \text{Size}}
                {K \mid \Gamma \mid \Sigma \vdash \textbf{split}\ [s]\ e : (\text{Unit}, \rho') \mid \varphi \times \{\textbf{split}\ \rho\ s\ \rho'\} \times \{\textbf{alloc}\ 1\ \rho' \} }
}
\newcommand{\varT}{
        \inferrule*[lab=t-var]
                {K \vdash \mu : \text{Type} }
                {K \mid \Gamma, x : \mu \mid \Sigma \vdash x : \mu \mid \{\bot\}}
        }

\newcommand{\valT}{
        \inferrule*[right=t-val]
                { K \mid \Gamma \mid \Sigma \vdash e : (\tau', \rho) \mid \varphi \\\\
                  K \mid \Gamma \mid \Sigma \vdash v : \tau \qquad K \vdash s : \text{Size}}
                { K \mid \Gamma \mid \Sigma \vdash v\ [s]\ \textbf{at}\ e : (\tau, \rho) \mid \varphi \times \{ \textbf{alloc}\ s\ \rho\}}
}

\newcommand{\appT}{
        \inferrule*[right=t-app]
                {K \mid \Gamma \mid \Sigma \vdash e_1 : (\mu_1 \xrightarrow{\varphi} \mu_2, \rho) \mid \varphi_1 \qquad K \mid \Gamma \mid \Sigma \vdash e_2 : \mu_1 \mid \varphi_2}
                {K \mid \Gamma \mid \Sigma \vdash e_1\ e_2 : \mu_2 \mid  \varphi_1 \times \varphi_2 \times \varphi  }
        }

\newcommand{\refT}{
        \inferrule*[lab=t-ref]
                {K \mid \Gamma \mid \Sigma \vdash e : (\tau, \rho) \mid \varphi}
                {K \mid \Gamma \mid \Sigma \vdash \textbf{ref}\ e : (\text{Ref}\ \tau, \rho) \mid \varphi \times \{\textbf{alloc}\ 1\ \rho\} }
}

\newcommand{\derefT}{
        \inferrule*[lab=t-deref]
                {K \mid \Gamma \mid \Sigma \vdash e : (\text{Ref}\ \tau, \rho) \mid \varphi}
                {K \mid \Gamma \mid \Sigma \vdash\ !e : (\tau, \rho) \mid \varphi}
}

\newcommand{\assignT}{
        \inferrule*[lab=t-assign]
                {K \mid \Gamma \mid \Sigma \vdash e_1 : (\text{Ref}\ \tau, \rho') \mid \varphi_1  \\\\ K \mid \Gamma \mid \Sigma \vdash e_2 : \mu \mid \varphi_2}
                {K \mid \Gamma \mid \Sigma \vdash e_1 := e_2 : (\text{Unit}, \globalRegion) \mid \varphi_1 \times \varphi_2 }
}

\newcommand{\ifT}{
        \inferrule*[right=t-if]
                {K \mid \Gamma \mid \Sigma \vdash e_1 : (\text{Bool}, \rho) \mid \varphi_1 \\\\ 
                 K \mid \Gamma \mid \Sigma \vdash e_2 : \mu_2 \mid \varphi_2 \qquad
                 K \mid \Gamma \mid \Sigma \vdash e_3 : \mu_2 \mid \varphi_3 }
                {K \mid \Gamma \mid \Sigma \vdash \textbf{if}\ e_1\ \textbf{then}\ e_2\ \textbf{else}\ e_3 : \mu_2 \mid \varphi_1 \times (\varphi_2 \sqcup \varphi_3) }
}

\newcommand{\seqT}{
        \inferrule*[lab=t-seq]
                {K \mid \Gamma \mid \Sigma \vdash e_1 : (\text{Unit}, \rho) \mid \varphi_1 \\\\ K \mid \Gamma \mid \Sigma \vdash e_2 : \mu \mid \varphi_2}
                {K \mid \Gamma \mid \Sigma \vdash e_1 ; e_2 : \mu \mid \varphi_1 \times \varphi_2 }
}

\newcommand{\copyT}{
        \inferrule*[right=t-copy]
                {K \mid \Gamma \mid \Sigma \vdash e_1 : (\tau, \rho) \mid \varphi_1
                \qquad
                K \mid \Gamma \mid \Sigma \vdash e_2 : (\tau', \rho') \mid \varphi_2 
                 }
                {K \mid \Gamma \mid \Sigma \vdash \textbf{copy}\ e_1\ \textbf{into}\ e_2 : (\tau, \rho') \mid  \varphi_1 \times \varphi_2 \times \{\textbf{alloc}\ 1\ \rho'\} }
}

\newcommand{\letT}{
        \inferrule*[lab=t-let]
                {K \mid \Gamma  \mid \Sigma \vdash e_1 : \mu_1 \mid \varphi_1 \\\\ 
                 K \mid \Gamma, x : \mu_1 \mid \Sigma \vdash e_2 : \mu_2 \mid \varphi_2 }
                {K \mid \Gamma \mid \Sigma \vdash \textbf{let}\ x = e_1\ \textbf{in}\ e_2 : \mu_2 \mid \varphi_1 \times \varphi_2 }
}

\newcommand{\bigAppT}{
        \inferrule*[right=t-tyApp]
                {K \mid \Gamma \mid \Sigma \vdash e : \forall \{ \alpha, \rho , \epsilon : \text{Effect} \} . \mu_1 \xrightarrow{\varphi} \mu_2 \mid \{ \bot \} \\\\ K \vdash (\tau, \rho') : \text{Type} }
                {K \mid \Gamma \mid \Sigma \vdash e\ @\ (\tau, \rho', \varphi) : [\alpha \mapsto \tau, \rho \mapsto \rho'](\mu_1 \xrightarrow{\varphi} \mu_2) \mid \{ \bot \} }
}

\newcommand{\fixT}{
        \inferrule*[lab=t-fix]
                {
                        \mu_f = (\forall \{ \alpha, \rho, \epsilon  \} . (\alpha, \rho) \xrightarrow{\{ \epsilon \}} \mu_1, \rho_f) 
                        \\\\
                 K \mid \Gamma, f : \mu_f \mid \Sigma \vdash (\Lambda\{ \alpha, \rho, \epsilon \}. \lambda x . e_1)\ [s]\ \textbf{at}\ e_2 : (\forall \{ \alpha , \rho, \epsilon  \} . (\alpha, \rho) \xrightarrow{\varphi} \mu_1, \rho_f)  \mid \varphi_1 
                 \\\\
                 \varphi' = [\{ \epsilon \} \mapsto \{ \textbf{rec}\ \epsilon\ \varphi \}] \varphi
                 \qquad
                 K \mid \Gamma, f : (\forall \{ \alpha , \rho , \epsilon  \} . (\alpha, \rho) \xrightarrow{\varphi'} \mu_1, \rho_f) \mid \Sigma \vdash e_3 : \mu_2 \mid \varphi_2
                }
                {
                 K \mid \Gamma \mid \Sigma \vdash \textbf{let}\ f = \textbf{fix}(f, (\Lambda \{ \alpha , \rho , \epsilon \} . \lambda x . e_1)\ [s]\ \textbf{at}\ e_2)\ \textbf{in}\ e_3 
                 : \mu_2 \mid \varphi_1 \times \varphi_2
                }
}

\newcommand{\storeInner}{\sigma}
\newcommand{\storeOuter}{\sigma}
\newcommand{\globalRegion}{\rho_\text{glob}}

\newcommand{\eformat}[4]{
        \langle #1 \mid #2 \rangle \longrightarrow \langle #3 \mid #4 \rangle
}

\newcommand{\eNewrgn}{
        \inferrule*[lab=e-newrgn]
                {\rho = \text{freshRegion()} \qquad l_\rho = \text{freshLoc}(\rho) }    
                {\eformat{\textbf{newrgn}\ [s]}{\storeOuter}{l_\rho}{\storeOuter, \rho \mapsto (l_\rho \mapsto (), s)}}
}

\newcommand{\eFreergn}{
        \inferrule*[right=e-freergn]
                {\eformat{e}{\storeOuter}{e'}{\storeOuter'}}
                {\eformat{\textbf{freergn}\ e}{\storeOuter}{\textbf{freergn}\ e'}{\storeOuter'}}
}

\newcommand{\eFreergnL}{
        \inferrule*[right=e-freergnL]
                {\quad}
                {\eformat{\textbf{freergn}\ l_\rho}{\storeOuter}{l^{1}_{\globalRegion}}{\storeOuter \setminus \rho}}
}

\newcommand{\eSplit}{
        \inferrule*[right=e-split]
                {\eformat{e}{\storeOuter}{e'}{\storeOuter'}}
                {\eformat{\textbf{split}\ [s]\ e}{\storeOuter}{\textbf{split}\ [s]\ e'}{\storeOuter'}}
}

\newcommand{\eSplitL}{
        \inferrule*[right=e-splitL]
                {\rho' = \text{freshRegion()} 
                 \qquad (\storeInner_{\rho}^{in.}, s_a) = \storeOuter(\rho) \qquad l_{\rho'} = \text{freshLoc}(\rho') }
                {\eformat{\textbf{split}\ [s]\ l_\rho}{\storeOuter}{l_\rho'}{[\rho \mapsto (\storeInner_{\rho}^{in.}, s_a \dot - s)] \storeOuter,  \rho' \mapsto (l_{\rho'} \mapsto (), s)}}
}

\newcommand{\eCopyOne}{
        \inferrule*[right=e-copy1]
                {\eformat{e_1}{\storeOuter}{e'_1}{\storeOuter'}}
                {\eformat{\textbf{copy}\ e_1\ \textbf{into}\ e_2}{\storeOuter}{\textbf{copy}\ e'_1\ \textbf{into}\ e_2}{\storeOuter'}}
}

\newcommand{\eCopyTwo}{
        \inferrule*[right=e-copy2]
                {\eformat{e_2}{\storeOuter}{e'_2}{\storeOuter'}}
                {\eformat{\textbf{copy}\ l_\rho\ \textbf{into}\ e_2}{\storeOuter}{\textbf{copy}\ l_\rho\ \textbf{into}\ e'_2}{\storeOuter'}}
}

\newcommand{\eCopyL}{
        \inferrule*[right=e-copyL]
                {(\storeInner_{\rho}^{in.}, s_a) = \storeOuter(\rho) \qquad (\storeInner_{\rho'}^{in.}, s'_a) = \storeOuter(\rho') \\\\
                 v = \storeInner_{\rho}^{in.} (l_\rho) \qquad l'_{\rho'} = \text{freshLoc}(\rho')}
                {\eformat{\textbf{copy}\ l_\rho\ \textbf{into}\ l_{\rho'}}{\storeOuter}{l'_{\rho'}}{[\rho' \mapsto ((\storeInner_{\rho'}^{in.}, l'_{\rho'} \mapsto l_{\rho}), s'_a)] \storeOuter }}
}

\newcommand{\eVal}{
        \inferrule*[right=e-val]
                {\eformat{e}{\storeOuter}{e'}{\storeOuter'}}
                {\eformat{v\ [s]\ \textbf{at}\ e}{\storeOuter}{v\ [s]\ \textbf{at}\ e'}{\storeOuter'}}
}

\newcommand{\eValL}{
        \inferrule*[right=e-valL]
                {(\storeInner_{\rho}^{in.}, s_a) = \storeOuter(\rho) \qquad s_v = \text{sizeOf}(v)  \\\\
                 s_v \sqsubseteq s \qquad l'_{\rho} = \text{freshLoc}(\rho) }
                {\eformat{v\ [s]\ \textbf{at}\ l_\rho}{\storeOuter}{l'_{\rho}}{[\rho \mapsto ((\storeInner_{\rho}^{in.}, l'_{\rho} \mapsto v), s_a)]\storeOuter }}
}

\newcommand{\eAppOne}{
        \inferrule*[right=e-app1]
                {\eformat{e_1}{\storeOuter}{e'_1}{\storeOuter'}}
                {\eformat{e_1\ e_2}{\storeOuter}{e'_1\ e_2}{\storeOuter'}}
}

\newcommand{\eAppTwo}{
        \inferrule*[right=e-app2]
                {\eformat{e_2}{\storeOuter}{e'_2}{\storeOuter'}}
                {\eformat{l_\rho\ e_2}{\storeOuter}{l_\rho\ e'_2}{\storeOuter'}}
}

\newcommand{\eAppL}{
        \inferrule*[right=e-appL]
                {(\storeInner_{\rho}^{in.}, s_a) = \storeOuter(\rho) \qquad (\lambda x . e) = \storeInner^{\rho}_{in.} (l_\rho)}
                {\eformat{l_\rho\ l'_{\rho'}}{\storeOuter}{[x \mapsto l'_{\rho'}]e}{\storeOuter}}
}

\newcommand{\eRef}{
        \inferrule*[right=e-ref]
                {\eformat{e}{\storeOuter}{e'}{\storeOuter'}}
                {\eformat{\textbf{ref}\ e}{\storeOuter}{\textbf{ref}\ e'}{\storeOuter'}}
}

\newcommand{\eRefL}{
        \inferrule*[right=e-refL]
                {(\storeInner_{\rho}^{in.}, s_a) = \storeOuter(\rho) \qquad l'_{\rho} = \text{freshLoc}(\rho)}
                {\eformat{\textbf{ref}\ l_\rho}{\storeOuter}{l'_\rho}{[\rho \mapsto ((\storeInner_{\rho}^{in.}, l'_{\rho} \mapsto l_\rho), s_a)]\storeOuter}}
}

\newcommand{\eDeref}{
        \inferrule*[right=e-deref]
                {\eformat{e}{\storeOuter}{e'}{\storeOuter'}}
                {\eformat{!e}{\storeOuter}{!e'}{\storeOuter'}}
}

\newcommand{\eDerefL}{
        \inferrule*[right=e-derefL]
                {(\storeInner_{\rho}^{in.}, s_a) = \storeOuter(\rho) \qquad l'_{\rho} = \storeInner^{\rho}_{in.} (l_\rho)}
                {\eformat{!l_\rho}{\storeOuter}{l'_{\rho}}{\storeOuter}}
}

\newcommand{\eAssignOne}{
        \inferrule*[right=e-assign1]
                {\eformat{e_1}{\storeOuter}{e'_1}{\storeOuter'}}
                {\eformat{e_1 := e_2}{\storeOuter}{e'_1 := e_2}{\storeOuter'}}
}

\newcommand{\eAssignTwo}{
        \inferrule*[right=e-assign2]
                {\eformat{e_2}{\storeOuter}{e'_2}{\storeOuter'}}
                {\eformat{l_\rho := e_2}{\storeOuter}{l_\rho := e'_2}{\storeOuter'}}
}

\newcommand{\eAssignL}{
        \inferrule*[lab=e-assignL]
                {(\storeInner_{\rho}^{in.}, s_a) = \storeOuter(\rho) }
                {\eformat{l_\rho := l_{\rho'}}{\storeOuter}{l^{1}_{\globalRegion}}{ [\rho \mapsto ((\storeInner_{\rho}^{in.}, l_\rho \mapsto l_{\rho'}), s_a)] \storeOuter}}
}

\newcommand{\eSeq}{
        \inferrule*[right=e-seq]
                {\eformat{e_1}{\storeOuter}{e'_1}{\storeOuter'}}
                {\eformat{e_1 ; e_2}{\storeOuter}{e'_1 ; e_2}{\storeOuter'}}
}

\newcommand{\eSeqNext}{
        \inferrule*[lab=e-seqNext]
                {\quad}
                {\eformat{l_\rho ; e_2}{\storeOuter}{e_2}{\storeOuter}}
}

\newcommand{\eBigApp}{
        \inferrule*[lab=e-tyApp]
                {\quad}
                {\eformat{e\ $@$\ (\tau, \rho')}{\storeOuter}{e}{\storeOuter}}
}

\newcommand{\eIf}{
        \inferrule*[right=e-if]
                {\eformat{e_1}{\storeOuter}{e'_1}{\storeOuter'}}
                {\eformat{\textbf{if}\ e_1\ \textbf{then}\ e_2\ \textbf{else}\ e_3}
                         {\storeOuter}
                         {\textbf{if}\ e'_1\ \textbf{then}\ e_2\ \textbf{else}\ e_3}
                         {\storeOuter'}}
}

\newcommand{\eIfTrue}{
        \inferrule*[lab=e-ifTrue]
                {\storeInner_{\rho}^{in.} = \storeOuter(\rho) \qquad \text{true} = \storeInner_{\rho}^{in.}(l_\rho)}
                {\eformat{\textbf{if}\ l_\rho\ \textbf{then}\ e_2\ \textbf{else}\ e_3}
                         {\storeOuter}
                         {e_2} 
                         {\storeOuter}}
}

\newcommand{\eIfFalse}{
        \inferrule*[lab=e-ifFalse]
                {\storeInner_{\rho}^{in.} = \storeOuter(\rho) \qquad \text{false} = \storeInner_{\rho}^{in.}(l_\rho)}
                {\eformat{\textbf{if}\ l_\rho\ \textbf{then}\ e_2\ \textbf{else}\ e_3}
                         {\storeOuter}
                         {e_3} 
                         {\storeOuter}}
}

\newcommand{\eLet}{
        \inferrule*[right=e-let]
                {\eformat{e_1}{\storeOuter}{e'_1}{\storeOuter'}}
                {\eformat{\textbf{let}\ x : \mu = e_1\ \textbf{in}\ e_2}
                         {\storeOuter}
                         {\textbf{let}\ x : \mu = e'_1\ \textbf{in}\ e_2}
                         {\storeOuter'}}
}

\newcommand{\eLetL}{
        \inferrule*[right=e-letL]
                {\quad}
                {\eformat{\textbf{let}\ x : \mu = l_\rho\ \textbf{in}\ e_2}
                         {\storeOuter}
                         {[x \mapsto l_\rho] e_2}
                         {\storeOuter}}
}

\newcommand{\eFix}{
        \inferrule*[right=e-fix]
                {\eformat{\Lambda \{\alpha, \rho, \epsilon\} . \lambda x . e_1\ [s]\ \textbf{at}\ e_2}
                         {\storeOuter}
                         {e'}
                         {\storeOuter'}}
                {\eformat{\textbf{let}\ f = \textbf{fix}(f, \Lambda \{\alpha, \rho, \epsilon\} . \lambda x . e_1\ [s]\ \textbf{at}\ e_2)\ \textbf{in}\ e_3}
                         {\storeOuter}
                         {\textbf{let}\ f = (f, e')\ \textbf{in}\ e_3}
                         {\storeOuter'}}
}

\newcommand{\eFixL}{
        \inferrule*[lab=e-fixL]
                {\storeInner_{\rho}^{in.} = \storeOuter(\rho) \qquad \Lambda\{\alpha, \rho, \epsilon\} . \lambda x . e_1 = \storeInner_{\rho}^{in.}(l_\rho) }
                {\eformat{\textbf{let}\ f = (f, l_{\rho})\ \textbf{in}\ e_3}
                         {\storeOuter}
                         {[x \mapsto l_{\rho}] e_3}
                         {\storeOuter}}
}

\newcommand{\composeEBot}{
        \inferrule*[right=$\times$-$\bot$]
                {\quad}
                {\varphi \times \{\bot\} }
}

\newcommand{\composeEFresh}{
        \inferrule*[right=$\times$-Fresh]
                {\quad}
                {\varphi \times \{\textbf{fresh}\ \rho\ s \}}
}

\newcommand{\composeEFree}{
        \inferrule*[right=$\times$-Free]
                { \{\textbf{free}\ \rho \} \notin \varphi }
                {\varphi \times \{\textbf{free}\ \rho \} }
}

\newcommand{\composeEFreshAlloc}{
        \inferrule*[lab=$\times$-FreshAlloc]
                { \exists s'' . s' + s'' + \text{sumAllocs}(\rho, \varphi) \sqsubseteq s }
                {\{\textbf{fresh}\ \rho\ s\} \times \varphi \times \{\textbf{alloc}\ s'\ \rho\} } 
}

\newcommand{\composeESplitAlloc}{
        \inferrule*[lab=$\times$-SplitAlloc]
                { \exists s'' . s' + s'' + \text{sumAllocs}(\rho, \varphi) \sqsubseteq s }
                {\{\textbf{split}\ \rho\ s\ \rho'\} \times \varphi \times \{\textbf{alloc}\ s'\ \rho\} } 
}

\newcommand{\composeEAlloc}{
        \inferrule*[lab=$\times$-Alloc]
                { \{\textbf{free}\ \rho \} \notin \varphi \qquad 
                  \{\textbf{fresh}\ \rho\ s'\} \notin \varphi }
                {\varphi \times \{\textbf{alloc}\ s\ \rho\} }
}

\newcommand{\composeEFreshSplit}{
        \inferrule*[lab=$\times$-FreshSplit]
                {
                \exists s'' . s' + s'' + \text{sumAllocs}(\rho, \varphi) \sqsubseteq s 
                }
                { \{\textbf{fresh}\ \rho\ s\} \times \varphi \times \{\textbf{split}\ \rho\ s'\ \rho'\}  }
}

\newcommand{\composeESplitSplit}{
        \inferrule*[lab=$\times$-SplitSplit]
                {
                \exists s'' . s' + s'' + \text{sumAllocs}(\rho, \varphi) \sqsubseteq s }
                { \{\textbf{split}\ \rho\ s\ \rho'\} \times \varphi \times \{\textbf{split}\ \rho'\ s'\ \rho''\}  }
}

\newcommand{\composeESplit}{
        \inferrule*[lab=$\times$-Split]
                {
                 \{\textbf{free}\ \rho\} \notin \varphi }
                {\varphi \times \{\textbf{split}\ \rho\ s\ \rho'\} }
}

\newcommand{\composeEVarLeft}{
        \inferrule*[lab=$\times$-VarL]
                { 
                 \forall (\rho \times s) \in  \text{freeAllocs}(\varphi) . s \sqsupseteq \omega}
                {\{\epsilon\} \times \varphi } 
}

\newcommand{\composeEVarRight}{
        \inferrule*[lab=$\times$-VarR]
                { 
                 \forall (\rho \times s) \in  \text{freeAllocs}(\varphi) . s \sqsupseteq \omega}
                {\varphi \times \{ \epsilon \} } 
}

\newcommand{\composeERec}{
        \inferrule*[lab=$\times$-Rec]
                { [\epsilon \mapsto \{\bot\}] \varphi' } 
                {\varphi \times \{\textbf{rec}\ \epsilon\ \varphi' \} } 
}

\newcommand{\composeEJoin}{
        \inferrule*[lab=$\times$-$\sqcup$]
                {
                 \varphi_1 \times \varphi_2  \qquad 
                 \varphi_1 \times \varphi_3 }
                {\varphi_1 \times (\varphi_2 \sqcup \varphi_3) }
}

\begin{document}

\maketitle

\begin{abstract}
    Region based memory management is a powerful tool designed with the goal of
    ensuring memory safety statically. The region calculus of Tofte and Talpin
    is a well known example of a region based system, which uses regions to
    manage memory in a stack-like fashion. However, the region calculus is
    lexically scoped and requires explicit annotation of memory regions, which
    can be cumbersome for the programmer. Other systems have addressed  
    non-lexical regions, but these approaches typically require the use of a
    substructural type system to track the lifetimes of regions.
    %
    %
    We present \textsc{Spegion}, a language with implicit non-lexical
    regions, which provides these same memory safety guarantees for
    programs that go beyond using memory allocation in a stack-like
    manner. We are able to achieve this with a concise syntax, and
    without the use of substructural types, relying instead on an
    effect system to enforce constraints on region allocation and
    deallocation.
    These regions may be divided into sub-regions, i.e.,
    \textbf{Sp}littable r\textbf{Egion}s, allowing fine grained control over
    memory allocation. Furthermore, \textsc{Spegion} permits \emph{sized}
    allocations, where each value has an associated size which is used to ensure
    that regions are not over-allocated into. We present a type system for
    \textsc{Spegion} and prove it is type safe with respect to a small-step
    operational semantics.   
\end{abstract}

\section{Introduction}
\label{sec:intro}
Writing error-free code is hard. Writing error-free systems code is even harder.
Possibly the most common source of bugs in systems programming code is the
traditional approach of manual memory management. Use-after-free bugs, where a
program attempts to access memory that has been deallocated, are pervasive in C
code, and can lead to undefined behaviour, crashes, and security
vulnerabilities. Nevertheless, manual memory management is still widely used in
systems code as it provides the programmer with fine-grained control over memory
usage. A vast swathe of research has been dedicated to the problem of memory
safety in systems programming languages. One approach is to use a garbage
collector to automatically manage memory. However, garbage collectors can    
introduce unpredictable pauses in the execution of a program, which is
unacceptable in systems programming where performance is critical. Another
approach is to ensure memory safety statically using a type system, one of the
most well-studied verification tools in computer science.

The region calculus of Tofte and Talpin~\cite{TOFTE1997109} is a type system
designed for ensuring memory safety in ML~\cite{tofte2001} through the use of
memory \emph{regions}. Memory is divided into regions which are allocated in
a stack-like manner, using the language's $\textbf{letregion}$ construct. This
allows deallocation to be done in a single operation by popping the topmost
region from the stack:
\begin{align*}
  \textbf{letregion}\ \rho\ \textbf{in}\ e
\end{align*}
The memory region $\rho$ is created and pushed onto the region stack.  
After the evaluation of $\textbf{letregion}$, $\rho$ is popped from the stack,
deallocating any memory allocated during the evaluation of $e$. Region based
memory management has since been applied to languages other than ML, such as
Java~\cite{10.1145/780822.781168, 10.1145/773039.512433} and
Prolog~\cite{10.1145/362426.362434}. However, these region calculi are lexically
scoped, restricting the programmer from expressing more complex memory usage
patterns that may arise in code. Lexically-scoped regions are tied to the static
structure of the program and are deallocated in a stack-like manner. This limits
flexibility in memory management, as regions cannot be split or partially
deallocated. Consider the following C program, for example: 
\begin{figure}[h]
\begin{lstlisting}[language=C,basicstyle=\ttfamily\footnotesize]
struct task {int n; struct task* next;};

struct task* build_tasks () {
    struct task* t      = malloc(sizeof(struct task));
    t->next             = malloc(sizeof(struct task));
    t->next->next       = malloc(sizeof(struct task));
    t->next->next->n    = 1;
    t->next->next->next = NULL;
    return t;
  }

int main () {
  struct task* t_head = build_tasks();
  struct task* t_last = t_head->next->next;
  free(t_head->next);
  free(t_head);
  return t_last->n;
}
\end{lstlisting}
\caption{Non-lexically-scoped memory use.}
\label{listing:linked_list}
\end{figure}

Here we have a linked list of tasks, where each task has a pointer to the next
task. The function \texttt{build\_tasks} allocates three tasks and links them
together. The main function then frees the first two tasks, leaving the third
task dangling. Traditional region calculi are unable to express this pattern as the
regions are allocated and de-allocated non-lexically.

There has been prior work exploring the safe use of non-lexically scoped
regions, though this work has often focused on inferring region annotations for
functional programs, or involved the use of substructural types (requiring that
code is written to use memory linearly or avoid aliasing, for example), or
imposed constraints on the runtime of programs like requiring reference
counting. This research is covered in detail in Section~\ref{subsec:regions}.


We present an approach to non-lexically scoped region which allows for the same
memory safety guarantees as prior systems, without the need for linear or
substructural types and with a concise syntax. We refer to these as \emph{implicit}
regions, as they are not first-class values in the language's syntax and are not
tied to the lexical structure of the program. Our system, named
\textsc{Spegion}, removes the need for first class regions in the language's
source syntax (with the exception of polymorphic constructs). Instead, when
allocating a piece of data, the programmer needs only to provide an existing
value in the region they wish to allocate the new value into. For example, 
\begin{align*} 
    \begin{array}{ll}
        \textbf{let}\ x =\ \textbf{newrgn}\ \textbf{in}\\
        \ \ \ \ \ v\ \textbf{at}\ x\\ 
        \ \ \ \ \ \textbf{freergn}\ x; 
    \end{array}
\end{align*}
The $\textbf{newrgn}$ construct returns a null pointer for a new region, which is
then used to allocate $v$ into that region. The type system then ensures that
the region associated with $x$ is live when $v$ is allocated into it.
Eliminating first-class regions from the syntax provides natural compatibility
with C-like languages, where regions are not first-class values.
\textsc{Spegion} sidesteps the need to retrofit complex region syntax onto C by
using a special syntactic form for allocation ($v\ \textbf{at}\ x$) to determine
which region to allocate a value into. Moreover, implicit regions eliminate the
need for substructural type constraints, a common source of complexity in region
type systems and a feature which C does not support. 

Our non-lexical approach to regions in \textsc{Spegion} also permits the
splitting of regions into sub-regions, allowing unused memory to be reclaimed
without deallocating the entire region. This allows the typing of patterns which
are difficult to express, even in C.

Furthermore, our system introduces the notion of \emph{sized} regions. Each
value in our language has an associated size, which is an abstract
representation of the amount of memory that the value occupies. We use this
information to ensure that regions are not over-allocated into, as each region
has an associated maximum size constraint. These constraints are provided
explicitly by the programmer as annotations on region creation and allocation
sites, giving them fine-grained control over memory usage in their program. In
the previous example, we allocated $v$ into the region associated with $x$. This
allocation carries the implicit constraint that the region of $x$ must be of an
arbitrary size. A more fine-grained allocation is possible using size
annotations: 
\begin{align*}
    \begin{array}{ll}
        \textbf{let}\ x =\ \textbf{newrgn}\ [10]\ \textbf{in}\\
        \ \ \ \ \ v\ [5]\ \textbf{at}\ x
    \end{array}
\end{align*}
Here, $10$ is a size representing $10$ abstract units of memory. The region
associated with $x$ has a size of $5$ and we can only allocate into this region
if the total size of allocations into the region does not exceed $10$. In the
above code, we allocate $5$ units of memory into the region. as indicated by the
size annotation on the allocation site. Contrarily, the following code would be
invalid:
\begin{align*}
    \begin{array}{ll}
        \textbf{let}\ x =\ \textbf{newrgn}\ [10]\ \textbf{in}\\
        \ \ \ \ \ v\ [10]\ \textbf{at}\ x; \\ 
        \ \ \ \ \ v\ [5]\ \textbf{at}\ x \\ 
    \end{array}
\end{align*}
Since the sum of allocations ($15$) is greater than the region's bound size.
This program is thus rejected by the type checker. 

\textsc{Spegion}'s size annotations are abstract units of memory, but the
principle can be applied to concrete memory sizes. In performance critical
applications, such as embedded systems, where the programmer has clear static
constraints on the size of data structures, the programmer can use these size
annotations to ensure that the program does not exceed the available memory. 

Size annotations are an optional feature of \textsc{Spegion}, as unannotated
regions and allocations default to being of unbounded size. These size
constraints on allocations can be enforced using an effect system, and present
an opportunity for further research into type-level reasoning about memory
usage.

\subsection{Contributions}
This paper makes the following contributions:   
\begin{itemize}

  \item We introduce \textsc{Spegion}, a novel core calculus and formal
        semantics for a language featuring implicit, non-lexical regions with
        explicit size constraints. Unlike traditional region calculi which are
        lexically scoped and require first-class region syntax, \textsc{Spegion}
        removes first-class regions from the source language syntax (except for
        polymorphic constructs), enabling natural compatibility with C-like
        languages. 
        
  \item We develop a type system that leverages effects to ensure safe
        allocation into implicit regions, eliminating the need for substructural
        or linear types commonly required in prior non-lexical region systems
        such as Cyclone. This approach supports flexible memory management
        patterns including region splitting and partial deallocation, which are
        difficult to express in lexically scoped systems. 
        
  \item We extend the concept of region-based memory management by introducing
        sized regions, where each region and allocation carries an abstract size
        annotation. This allows static enforcement of memory usage bounds via an
        effect system, providing fine-grained control over memory consumption
        and enabling static reasoning about memory usage in performance-critical
        applications. 
        
  \item We provide a formal proof of type safety for \textsc{Spegion},
        building on and adapting the syntactic proof techniques from Helsen and
        Thiemann’s work on the region calculus. 
        
  \item We demonstrate the expressiveness and practical applicability of
        \textsc{Spegion} through a series of illustrative examples, including a
        sketch of an extension combining refinement types with region sizes to
        further enhance static guarantees about memory usage.
\end{itemize}

\subsection{Overview}
The rest of the paper is structured as follows. Section~\ref{sec:static}
introduces the syntax and typing rules of \textsc{Spegion}. In
Section~\ref{sec:dynamic} we present the dynamic semantics of the language as
reduction rules. In Section~\ref{sec:properties} we present the type safety
properties of the language in the form of progress and preservation theorems.
Section~\ref{sec:applications} explores the potential applications of
\textsc{Spegion} in several domains, including a sketch of how refinement types
can be combined with our system to provide powerful static guarantees about
memory usage. We also discuss the application of our system in the context of
memory safety in C-like languages. Section~\ref{sec:related} discusses related
work, and Section~\ref{sec:conclusion} concludes.

\section{Static Semantics}
\label{sec:static}
We define a core calculus for \textsc{Spegion}, drawing from the region calculus
of~\cite{TOFTE1997109} and the presentation by~\cite{HELSEN20011}. The syntax of
types is given by:
\begin{align*}
        \begin{array}{rl}
        \tau ::=& \alpha \mid 
                  \text{Int} \mid \text{Unit} \mid \text{Bool} \mid \text{Ref}\ \tau \mid 
                  \mu \xrightarrow{\varphi} \mu  \mid \forall \{ \alpha, \rho, \epsilon \}. \mu \xrightarrow{\varphi} \mu \\ 
        \mu  ::=& (\tau, \rho)
    \end{array}
    \tag{types}
\end{align*}
Types $\tau$ are either type variables $\alpha$, integers, units, booleans,
references to other types, function types, or polymorphic types. A program
expression in our calculus is assigned a ``type-with-place'', i.e., a pair of a
type and a region, denoted by $\mu$. 

Function types are written as $\mu_1 \xrightarrow{\varphi} \mu_2$. Above the
function arrow is an \emph{arrow effect}, i.e. the latent effect that happens on
application of the function, as in~\cite{TOFTE1997109}. Polymorphic functions
are typed by type schemes $\forall \{ \alpha, \rho, \epsilon \}. \mu $ which
bind type $\alpha$, region $\rho$, and effect $\epsilon$ variables. We refer to
these as kind-annotated type variables, however, in the syntax of our calculus
we typically omit the kind annotations for brevity.  

Kinds are given by the grammar:
\begin{align*}
        \begin{array}{rl}
        \kappa ::= & \text{Type} \mid \kappa \rightarrow \kappa \mid \text{Region} \mid \text{Effect} \mid \text{Size} 
        \end{array}
        \tag{kinds}
\end{align*}
That is a kind is either a type, a function kind, a region, an effect, or a
size. Kinding rules for types are given in Figure~\ref{fig:static-kinding}.
These rules are mostly straightforward and ensure that type, regions, effects,
and sizes are well-structured. 
\begin{figure}[h]
\begin{align*}
        \begin{array}{cc}
                \kVar 
                \;\;\;
                \kForall
                \;\;\;
                \kRegion
                \\[1em]
                \kArrow
                \;\;\;
                \kApp
                \\[1em]
                \kAppTwo
                \;\;\;
                \kInt 
                \;\;\; 
                \kUnit
                \\[1em]
                \kRef
                \;\;\;
                \kBool
                \;\;\;
                \kSize
                \\[1em]
                \kOp
                \;\;\;
                \kTyWithPlace
                \\[1em]
                \kBot
                \;\;\;
                \kCompose
                \;\;\;
                \kJoin
                \\[1em]
                \kAlloc
                \;\;\;
                \kNew
                \\[1em]
                \kSplit
                \\[1em]
                \kFree
                \;\;\;
                \kEps 
                \;\;\; 
                \kRec
        \end{array}
\end{align*}
\caption{Kinding rules in \textsc{Spegion}}
\label{fig:static-kinding}
\end{figure}

\subsection{Effects and Sizes}
An effect $\varphi$ describes the sequence of actions that are performed on a
region when an expression is evaluated. Many expressions in the calculus convey
some effect, which form part of the expression's typing judgement. The following
grammar defines these actions:
\begin{align*}
        \begin{array}{rl}
        \varphi ::= & \varphi \times \varphi \mid \{\bot\} \mid \{\textbf{fresh}\ \rho\ s \} \mid \{ \textbf{free}\ \rho\} \mid \{\textbf{split}\ \rho\ s\ \rho'\} \mid \{  \textbf{alloc}\ s\ \rho\}  
          \mid \varphi \sqcup \varphi \\ \mid &  \{\epsilon\} \mid \{\textbf{rec}\ \epsilon\ \varphi\} 
        \end{array} 
    \tag{effects}
\end{align*}
Effects are composed using the $\times$ operator, which describes the sequential
composition of two effects. This operator also ensures constraints on memory
usage are respected when two effects are composed. The $\{\bot\}$ effect 
describes an empty effect, i.e., no actions are performed, and is used 
in the typing of expressions which are pure computations. The creation 
of a fresh region is denoted by $\{\textbf{fresh}\ \rho\ s\}$, where $\rho$ is
a fresh region name and $s$ is the region's \emph{size}.

A size is an abstract unit of memory. For simplicity, we consider sizes to be
natural numbers, extended with a special size $\omega$ which represents an
unknown size. Sizes thus comprise an extended natural numbers preordered
semiring $(\overline{\mathbb{N}}, +, \cdot, \dot -, 0, 1, \sqsubseteq)$ where
$\overline{\mathbb{N}} = \mathbb{N} \cup \{ \omega \}$, such that $\forall n \in
\overline{\mathbb{N}} . n + \omega = \omega = \omega + n = \omega$. The preorder
$\sqsubseteq$ is the standard preorder on $\mathbb{N}$ with $\omega$ as the
greatest element. The semiring is also equipped with
a truncated subtraction (or \emph{monus}) operation $\dot -$, given by Definition~\ref{df:monus}: 
\begin{df}[Monus ($\dot -$)]
For any two numbers $n, n' \in \overline{\mathbb{N}}$, $\dot -$ is defined:
\begin{align*}
        n \dot - n' = \left\{\begin{matrix}
        \begin{array}{ll}
                n - n' & n, n' \sqsubset \omega \wedge n \sqsupseteq n',\\ 
                0 & n, n' \sqsubset \omega \wedge n \sqsubset n'\\
                n & n' = 0\\
                0 & n \sqsubset n' \wedge n' = \omega,\\
                \omega &\text{if } n = \omega
        \end{array}\end{matrix}\right.
\end{align*}
\label{df:monus}
\end{df}
The typing rules do not make use of subtraction, however, it is required during
evaluation. 

The effect $\{\textbf{free}\ \rho\}$ describes the freeing of a region $\rho$,
which deallocates all memory in the region. The creation of sub-regions is
denoted by $\{\textbf{split}\ \rho\ s\ \rho'\}$, where $\rho'$ is a fresh region
name for a sub-region of $\rho$ with size $s$. Allocation of a value into a
region is described by $\{\textbf{alloc}\ s\ \rho\}$, where $s$ is the size of
the value being allocated, and $\rho$ is the region it is being allocated into.
If statements introduce branching into our calculus, with the different cases
potentially having different effects. This is captured by the $\{\varphi_1
\sqcup \varphi_2\}$ effect, where $\varphi_1$ denotes the effect of the true
branch and $\varphi_2$ the effect of the false branch. Recursion is described in
effects through the $\{\epsilon\}$ and $\{\textbf{rec}\ \epsilon\ \varphi\}$
effects. The former denotes the usage of a recursive function definition, while
the latter captures the latent effect of this function upon application. We
explain further when we discuss the typing rules for recursive functions.

\subsection{Syntax and Typing}
The term syntax of our calculus comprises the $\lambda$-calculus, with the
addition of references, sequential composition, polymorphic recursion, as well
as constructs for region manipulation. These include expressions for creating,
freeing, and splitting regions, as well as primitives for allocating data into
regions and copying data between regions. The full syntax is given by the
following grammar, which divides constructs between values $v$ and expressions
$e$:
\begin{align*}
    \begin{array}{rl}
        v ::= & n \mid \text{true} \mid \text{false} \mid \Lambda \{ \alpha, \rho , \epsilon \} . \lambda x . e \mid \lambda x . e \mid () \mid l_\rho\\
        e ::= & x \mid l_\rho \mid v\ [s]\ \textbf{at}\ e \mid e\ e \mid \textbf{ref}\ e \mid\ !e \mid e := e \mid e ;\ e \mid \textbf{if}\ e\ \textbf{then}\ e\ \textbf{else}\ e \\
             \mid & \textbf{let}\ x = e\ \textbf{in}\ e \mid  e\ @\ \mu \mid \textbf{let}\ f = \textbf{fix}(f, (\Lambda \{ \alpha , \rho , \epsilon \} . \lambda x . e_1)\ [s]\ \textbf{at}\ e_2)\ \textbf{in}\ e_3 \\ 
             \mid & \textbf{newrgn}\ [s] \mid  \textbf{freergn}\ e \mid \textbf{split}\ [s]\ e \mid \textbf{copy}\ e\ \textbf{into}\ e 
    \end{array}
    \tag{terms}
\end{align*}

The syntax of values $v$ includes integers $n$, booleans $\text{true}$ and
$\text{false}$, polymorphic functions $\Lambda \{ \alpha , \rho , \epsilon \} .
\lambda x . e$, functions $\lambda x . e$, unit $()$, and locations $l_\rho$.
Expressions $e$ include variables $x$, locations $l_\rho$, value allocations $v\
[s]\ \textbf{at}\ e$, function application $e\ e$, reference creation
$\textbf{ref}\ e$, dereferencing $!e$, assignment $e := e$, sequential
composition $e;\ e$, conditionals $\textbf{if}\ e\ \textbf{then}\ e\
\textbf{else}\ e$, let-binding $\textbf{let}\ x = e\ \textbf{in}\ e$, type
applications $e\ @\ \mu$, and recursive function definitions $\textbf{let}\ f =
\textbf{fix}(f, (\Lambda \{ \alpha , \rho , \epsilon \} . \lambda x . e_1)\ [s]\
\textbf{at}\ e_2)\ \textbf{in}\ e_3$. Region manipulation constructs include
creating a new region $\textbf{newrgn}\ [s]$, freeing a region
$\textbf{freergn}\ e$, splitting a sub-region from another region
$\textbf{split}\ [s]\ e$, and copying data between regions $\textbf{copy}\ e\
\textbf{into}\ e$.

Expressions are typed by the judgement: 
\begin{align*}
    K \mid \Gamma \mid \Sigma \vdash e : (\tau, \rho) \mid \varphi
\end{align*}
assigning an expression $e$ the type $\tau$ in region $\rho$ with effect
$\varphi$. A context of kind-annotated type variables $K$ provides a mapping
from type variables to kinds, given by: 
\begin{align*}
\begin{array}{rl}
    K ::= & \emptyset \mid K, \alpha : \text{Type} \mid K, \rho : \text{Region} \mid K, \epsilon : \text{Effect} 
\end{array}
\tag{kinding contexts}
\end{align*}
That is, a context may be empty $\emptyset$, extended with a type variable
$\alpha : \text{Type}$, a region variable $\rho : \text{Region}$, or an effect
variable $\epsilon : \text{Effect}$. Free variable contexts $\Gamma$ are given
by: 
\begin{align*}
\begin{array}{rl}
    \Gamma ::= & \emptyset \mid \Gamma, x : \mu
\end{array}
\tag{free variable contexts}
\end{align*} 
A store typing context $\Sigma$ is a mapping from locations to types, defined: 
\begin{align*}
\begin{array}{rl}
    \Sigma ::= & \emptyset \mid \Sigma, l_\rho : \tau   
\end{array}
\tag{store typing}
\end{align*}
That is $l_\rho$ is a location in region $\rho$ with type $\tau$.
Figure~\ref{fig:static-typing-expressions} provides typing rules for the
expressions. Values, which are non-computational objects, are typed by a
separate judgement
\begin{align*}
    K \mid \Gamma \mid \Sigma \vdash v : \tau
\end{align*} 
This judgement is similar to that of expressions but does not include an effect.
Figure~\ref{fig:static-typing-values} gives the typing rules for values,
relating terms to types.

{\small{
\begin{figure}[t]
    \begin{align*}
    \begin{array}{cc}
        \intT
        \;\;
        \unitT
        \\[-0.3em]
        \trueT 
        \;\;
        \falseT
        \\[1em]
        \locValT
        \;\;
        \absT
        \\[1em]
        \biglamT
    \end{array}
    \end{align*}
    \caption{Value typing rules in \textsc{Spegion}}
    \label{fig:static-typing-values}
\end{figure}
}} 

The rules for integers (\textsc{t-int}), unit (\textsc{t-unit}), and bools
(\textsc{t-true}, \textsc{t-false}) are standard. For (\textsc{t-}$\lambda$),
the type of the argument $x$ is added to the variable context $\Gamma$ and the
expression $e$ is typed. The type of the function is then a function type from
the type of the argument to the type of the body. The effect of the function is
the latent effect of the body. A location $l_\rho$ is typed by (\textsc{t-loc}),
where the location's type $\tau$ is looked up in the store typing $\Sigma$.
Finally, polymorphic functions are typed by (\textsc{t-$\Lambda$}), which also
adds the kind-annotated type, region, and effect variables to the kind variable
context $K$. 
    
\begin{figure}[t]
    {{
    \begin{align*}
        \begin{array}{cc}
                \uselocT 
                \;\;
                \valT 
                \\[1em]
                \newrgnT
                \\[1em]
                \freergnT
                \\[1em]
                \splitT
                \\[1em]
                \copyT
                \\[1em]
                \varT 
                \;\;\;
                \letT
                \\[1em]
                \ifT
                \\[1em]
                \refT
                \;\;\;
                \derefT
                \\[1em]
                \assignT
                \;\;\;
                \seqT
                \\[1em]
                \appT
                \\[1em]
                \bigAppT
                \\[1em]
                \fixT
        \end{array}
    \end{align*}
    }}
    \caption{Expression typing rules in \textsc{Spegion}}
    \label{fig:static-typing-expressions}
\end{figure}
In the typing of expressions, the rule for variables (\textsc{t-var}) is
straightforward, if the variable is present in the  
free variable context $\Gamma$ with type $\mu$ then $x : \mu$. The rule for
locations (\textsc{t-use-val}) is also straightforward, if the location $l_\rho$ is
present in the store typing $\Sigma$ with type $\tau$ then $l_\rho : (\tau,
\rho)$.

Creating a new region is handled by the (\textsc{t-newrgn}) rule. The
$\textbf{newrgn}$ construct allocates a new region of size $s$ and has type
$(\text{Unit}, \rho)$ where $\rho$ is the freshly allocated region. This
behaviour is captured in the rule's effect $\{\textbf{fresh}\ \rho\ s\}$, which
binds a fresh region name $\rho$ to the size $s$. Evaluating this expression
returns a unit value, which acts as a pointer into the region, allowing
this region to be referenced in the subsequent program. This eliminates the need
for syntactic region variables in all non-polymorphic code. The allocation of
this null value is described in the second part of the effect $\{\textbf{alloc}\
1\ \rho\}$, i.e., allocate a value of size $1$ into $\rho$ (since a value cannot
be $0$-sized). Thus, we have the additional constraint $s \sqsupseteq 1$
ensuring regions must be able to allocate at least $1$ unit of memory.

The rule for freeing a region (\textsc{t-freergn}) deletes the region $\rho$
associated with the expression $e$, freeing its associated memory locations.
This yields an effect comprised of the effect of typing $e$ ($\varphi$) followed
by the effect $\{\textbf{free}\ \rho\}$. Evaluating a $\textbf{freergn}$
expression returns a pointer to value of type $\text{Unit}$ in the global
region $\globalRegion$ which is always available.

A sub-region may be split from a region using the (\textsc{t-split}) rule. As
with $\textbf{freergn}$, the region to be split is associated with the
sub-expression $e$. The programmer-defined size $s$ determines the size of the
new sub-region, represented by the effect $\{\textbf{split}\ \rho\ s\ \rho'\}$.
As in (\textsc{t-newrgn}), a null pointer into the new region is created (thus
the rule also has an $\textbf{alloc}$ effect), and the overall expression has
the type $(\text{Unit}, \rho')$, where $\rho'$ is the fresh sub-region.

Copying data between regions is handled by the (\textsc{t-copy}) rule. The
behaviour of $\textbf{copy}$ can also be simulated in terms of
(\textsc{t-assign}) and (\textsc{t-val}). However, we include the
$\textbf{copy}$ construct as a useful primitive. The expression $\textbf{copy}\
e_1\ \textbf{into}\ e_2$ copies the location obtained from evaluating $e_1$
from its region $\rho$ into the region $\rho'$ associated with $e_2$. The effect
of this operation is the composition of the effects obtained from typing $e_1$
and $e_2$, followed by the effect $\{\textbf{alloc}\ 1\ \rho'\}$. Locations
have a size of $1$, thus this effect conveys the creation of a new location in
$\rho'$ pointing to location value.

Allocating a value into a region is handled by the $\textsc{t-val}$ rule. A
programmer is required to annotate the allocation of a value with a size, which is
used in the effect $\{\textbf{alloc}\ s\ \rho\}$. Unlike~\cite{TOFTE1997109},
where the region is explicitly passed as an argument to the allocation
primitive, our rule allocates the value into the region associated with the
sub-expression $e$. For example, consider the following trivial program:
\begin{align*}
    \textbf{let}\ x = \textbf{newrgn}\ [3]\ \textbf{in}\ ()\ [2]\ \textbf{at}\ x
\end{align*}
which creates a new region of size $3$ and allocates a unit value of size $2$
into it. This program has the type $(\text{Int}, \rho)$ with the effect: 
\begin{align*}
    \{\textbf{fresh}\ \rho\ 3\} \times \{\textbf{alloc}\ 1\ \rho\} \times \{\textbf{alloc}\ 2\ \rho\}
\end{align*}
conveying the creation and allocation into a region without ever giving it an
explicit name in the program. This allows our calculus to model region-based
memory management in languages where these first-class regions are not present
in the source language.

Typing rules for let-binding (\textsc{t-let}), creating a reference
(\textsc{t-ref}), dereferencing (\textsc{t-deref}), assigning to a reference
(\textsc{t-assign}), and sequential composition (\textsc{t-seq}) thus handle the
sequencing of the effects of their sub-expressions, but are otherwise
straightforward. Similarly, application (\textsc{t-app}) sequences the effects
of the sub-expressions, followed by the function type's latent effect. If 
statements (\textsc{t-if}) are typed using a join effect $\{\varphi_1 \sqcup
\varphi_2\}$, which is the effect of the true branch $\varphi_1$ and the false
branch $\varphi_2$ combined.

Recursive function definitions are typed by the (\textsc{t-fix}) rule. The body
of the definition is typed with an effect variable in place of the functions
actual latent effect. In the rule's third premise, this variable is then
substituted for the recursive effect $\{\textbf{rec}\ \epsilon\ \varphi\}$
yielding $\varphi'$. In the typing of $e_3$, $f$ is then bound with the same
function type as before but with this new latent effect $\varphi’$. This
recursive effect conveys that the the effect $\varphi$ may be repeated an
unbounded number of times. This is akin to the notion of ``Kleene star
effects''~\cite{10.1007/978-3-319-27810-0_1, 10.1145/2914770.2837634} and
iterable sequential effects~\cite{10.1145/3450272}.

A consequence of recursion is that regions which appear freely in $e_1$ must be
capable of allocating an unbounded amount of memory. This constraint is enforced
at the point where the recursive variable effect is composed with another effect
in the typing of the function body (via $\times$), since the presence of a
$\epsilon$ in $\varphi$ indicates that the function recurses. We provide the
definition of valid effect composition in detail in
Section~\ref{sec:effect-combination}. 

Recursive variable usage thus requires that the variable’s type scheme be
instantiated via the (\textsc{t-tyApp}) rule, which instantiates the type scheme
and region at the type and region provided by the programmer. This is the only
rule in the syntax where a region name is required, which can be obtained using
a primitive operator $\text{regionOf}(e)$. 

\subsection{Effect Composition}
\label{sec:effect-combination}
The typing rules above define the individual effects of expressions. However,
taken alone an effect does not provide a complete picture of the memory usage of
a program. For example, the effect of an allocation is the $\{\textbf{alloc}\ s\
\rho\}$ effect preceded by the effect of the expression which identifies the
region being allocated into. However, at the point of typing this allocation,
there is no way of knowing if the allocation described by this effect is valid
- this can only be known when the effect is composed with the effect which
describes the creation of $\rho$. We describe below the behaviour of the
$\times$ operator, which composes two effects. Effects in our calculus are
sequential, thus the effect $\varphi \times \varphi'$ describes the effect of
$\varphi$ followed by the effect of $\varphi'$.
Figure~\ref{fig:effect-comp} provides the
rules for valid effect compositions. 

In the ($\times$-$\bot$) and ($\times$-\textsc{Fresh}) rules, composing an
effect with $\{ \bot \}$ or $\{\textbf{fresh}\ \rho\ s\}$ is always valid. In
the latter case we thus assume that \textsc{Spegion} always has the ability to
allocate new regions. An effect which frees a region $\{\textbf{free}\ \rho\}$
is valid only if the region is not already freed in the preceding effect
($\times$-\textsc{Free}). 

Effects which handle allocation must consider three cases. The first
($\times$-\textsc{FreshAlloc}) is when the entire lifetime of the region up
until the point of the current allocation is described in the preceding effects.
In this case, the size of the new allocation effect $s'$ must not exceed the
current total size of allocations in the $\rho$. This is enforced by the
constraint $\exists s'' . s'' + s' + \text{sumAllocs}(\rho, \varphi) \sqsubseteq s$, i.e., there
exists some leftover space $s''$ after adding the size of the new allocation
$s'$ to the total size of allocations in $\rho$. Definition~\ref{df:sumAllocs}
defines the function calculating this total size, mapping a region and effect to
a size. 
\begin{figure}[H]
    \begin{align*}
                \begin{array}{cc}
                        \composeEBot
                        \;\;\;
                        \composeEFresh
                        \;\;\;
                        \composeEFree 
                        \\[1em]
                        \composeEFreshAlloc 
                        \;\;\;
                        \composeEAlloc 
                        \\[1em]
                        \composeESplitAlloc 
                        \;\;\;
                        \composeEFreshSplit  
                        \\[1em]
                        \composeESplit
                        \;\;\;
                        \composeESplitSplit 
                        \;\;\;
                        \composeEJoin
                        \\[1em] 
                        \composeEVarRight
                        \;\;\; 
                        \composeEVarLeft 
                        \;\;\;
                        \composeERec
                \end{array}
            \end{align*}
    \caption{Effect composition rules in \textsc{Spegion} ($\times$)}
    \label{fig:effect-comp}
\end{figure}

\begin{df}[Total Size of Allocations]
    \begin{align*}
        \text{sumAllocs}(\rho, \varphi) = \left\{\begin{matrix}\begin{array}{lll}
            s + \text{sumAllocs}(\rho, \varphi') & \varphi = \varphi' \times \{\textbf{alloc}\ s\ \rho\} \\ 
            s + \text{sumAllocs}(\rho, \varphi') & \varphi = \varphi' \times \{\textbf{split}\ \rho\ s\ \rho'\}\\ 
            \text{sumAllocs}(\rho, \varphi') + \text{sumAllocs}(\rho, \varphi'') & \varphi = \varphi' \times \varphi'' \\
            \text{max}(\text{sumAllocs}(\rho, \varphi'), \text{sumAllocs}(\rho, \varphi'')) & \varphi = \varphi' \sqcup \varphi'' \\ 
            0 & \text{otherwise}
        \end{array}\end{matrix}\right.
    \end{align*}
    \label{df:sumAllocs}
\end{df}
A similar case applies when allocating into a sub-region
where the entire lifetime of the region is described in the preceding 
effects ($\times$-\textsc{SplitAlloc}). 

The second case for allocation ($\times$-\textsc{Alloc}) is when the region is
not created in the preceding effect, i.e. $\varphi \times \{\textbf{alloc}\ s\
\rho\}$ where $\{\textbf{fresh}\ s'\ \rho\} \notin \varphi$. This situation
arises when combining two effects in the typing of a sub-term of the expression
where the region is created. In this case, the current total size of allocations
in $\rho$ is unknown, thus the permissibility of the current allocation effect
is unknown. Thus, the allocation is simply accepted as valid with the knowledge
that $\varphi \times \{\textbf{alloc}\ s\ \rho\}$ will eventually be composed
with some effect describing the preceding lifetime of $\rho$, at which point it
will be validated by the previous case. 

The rules for composing split effects $\{\textbf{split}\ \rho\ s\ \rho'\}$ behave
in a similar manner as allocation ($\times$-\textsc{FreshSplit}),
($\times$-\textsc{SplitSplit}), and ($\times$-\textsc{Split}). The first two
rules handle the cases where the entire lifetime of the region is described in
the preceding effects, whilst the last rule handles the case where this
information is not available. 

Composing a join effect $\varphi_2 \sqcup \varphi_3$ with another effect
$\varphi_1$ is straightforward ($\times$-$\sqcup$). Since both branches of the
program must be well-typed, the effect $\varphi_1$ is simply composed with both
effects. If this is valid then so the composition of the join: $\varphi_1 \times
\{\varphi_2 \sqcup \varphi_3\}$. 

Effect variables $\epsilon$ occur as part of an effect when a recursive function
definition is typed. Composition of these effects are handled by the
($\times$-\textsc{VarR}) and ($\times$-\textsc{VarL}) rules. A constraint of
recursive function definitions is that regions which occur freely of the
function body must be created with an unbounded size $\omega$ since it is not
possible to know statically how many times the function will recurse. Regions
which are created inside the function body itself may be given any size. The
addition of a lightweight refinement types would allow us to express and type
programs which allocate recursively into a free region under some constraints.
We consider this idea in more depth in Section~\ref{subsec:refinement-types}.

The free regions of the effects $\varphi \times \{\epsilon \}$ and $\{\epsilon\}
\times \varphi$ are collected by the function $\text{freeAllocs}$, given by
Definition~\ref{df:free-allocs}. This function maps an effect to the set of
allocations (and splits) of free regions which occur in the effect. A region is
free if there is no $\{\textbf{fresh}\ s\ \rho\}$ effect in $\varphi$. This set
takes the form of pairs of regions and sizes, where $\bot$ is
used to indicate that the region was freed inside $\varphi$. Since regions of
unbounded size are represented by size $\omega$, the only sized allocations that
are permitted must also be those of $\omega$.
This is enforced by the constraint $\forall (\rho \times s) \in
\text{freeAllocs}(\varphi) . s \sqsupseteq \omega$.
\begin{df}[Free Region Allocations of an Effect]
    Given an effect $\varphi$, the allocations of regions which are free in $\varphi$ is given by: 
    \begin{align*}
        \text{freeAllocs}(\varphi) = \left\{\begin{matrix}
                \begin{array}{lll}
            \{ (\rho \times s) \} & \varphi = \{ \textbf{alloc}\ s\ \rho \}\\
            \{ (\rho \times s) \} \cup \text{freeAllocs}(\varphi') & \varphi = \varphi' \times \{ \textbf{alloc}\ s\ \rho \} \wedge \\ & \{\textbf{fresh}\ \rho\ s'\} \notin \varphi'\\           
            \{ (\rho \times s) \} & \varphi = \{\textbf{split}\ \rho\ s\ \rho'\}\\
            \{ (\rho \times s) \} \cup \text{freeAllocs}(\varphi') & \varphi = \varphi' \times \{\textbf{split}\ \rho\ s\ \rho'\} \wedge \\ & \{\textbf{fresh}\ \rho\ s'\} \notin \varphi'\\
            \{ (\rho \times 0 )\} & \varphi = \{\textbf{free}\ \rho\}\\ 
            \{ (\rho \times 0) \} \cup \text{freeAllocs}(\varphi') & \varphi = \varphi' \times \{\textbf{free}\ \rho\} \wedge \\ &  \{\textbf{fresh}\ \rho\ s\} \notin \varphi'\\
            \text{freeAllocs}(\varphi') & \varphi = \{\textbf{rec}\ \epsilon\ \varphi'\}\\
            \text{freeAllocs}(\varphi') \cup \text{freeAllocs}(\varphi'') & \varphi = \varphi' \times \varphi''\\
            \text{freeAllocs}(\varphi') \cup \text{freeAllocs}(\varphi'') & \varphi = \varphi' \sqcup \varphi''\\
            \emptyset & \text{otherwise}
        \end{array}\end{matrix}\right.
    \end{align*}
    \label{df:free-allocs}
\end{df}
Finally, the usage of a recursive function definition is handled via the
$\{\textbf{rec}\ \epsilon\ \varphi''\}$ effect. This rule
($\times$-\textsc{Rec}) is straightforward, since the burden of ensuring that
free regions in $\varphi''$ are unbounded is placed on the
($\times$-\textsc{VarR}) and ($\times$-\textsc{VarL}) rules. The latent effect
of the recursive function definition $\varphi''$ is simply composed with the
preceding effect $\varphi$. Effect variables in $\varphi''$ are substituted for
$\{\bot\}$ since there is no need to recheck the validity of $\{\epsilon\}$
effects. 

Effect combination is a partial function, since the combination of some effects
is not valid. For example, the combination of two effects which free the same
region is invalid. In this case, the function $\times$ is undefined. Undefined
compositions are rejected during typing.

We conclude this section with an example of a \textsc{Spegion} program and its typing:
\begin{align*}
    \begin{array}{ll}
    \textbf{let}\ x = \textbf{newrgn}\ [5]\ \textbf{in}\\
    \textbf{in}\\ 
    \ \ \ \ (\lambda z . \textbf{newrgn}\ [5];\\
    \ \ \ \ \ \ \ \ \textbf{newrgn}\ [5])\ [1] \textbf{at}\ x\\
    \end{array}
\end{align*}
which defines a function that creates two fresh regions, and is typed as: 
\begin{align*}
    ((\text{Unit}, \rho')\ \xrightarrow{\{\textbf{fresh}\ \rho_1\ 5\} \times \{\textbf{fresh}\ \rho_2\ 5\}} (\text{Unit}, \rho_2), \rho) \mid \{\textbf{alloc}\ 1\ \rho\}
\end{align*}

\section{Dynamic Semantics}
\label{sec:dynamic}
The relation $\langle e \mid \storeOuter \rangle \rightarrow \langle e' \mid
\storeOuter' \rangle$ means that $e$ reduces to $e'$ in a single step. Unlike
the calculus of~\cite{HELSEN20011}, our reduction rules include an explicit
store to enable the use of references. This store is dual layered, with an outer
store $\storeOuter$ mapping region names $\rho$ to pairs of inner stores
$\storeInner_{\rho}^{in.}$, and maximum sizes $s$: 
\begin{align*}
    \begin{array}{rl}
        \storeOuter ::= & \emptyset \mid \rho \mapsto (\storeInner_{\rho}^{in.}, s)
    \end{array}
    \tag{outer store}
\end{align*}
This maximum size which is specified by the programmer as an annotation when the
store is allocated initially via a \textbf{newrgn} expression. Inner stores thus
map locations to values:
\begin{align*}
    \begin{array}{rl}
        \storeInner_{\rho}^{in.} ::= & \emptyset \mid l \mapsto v 
    \end{array}
    \tag{inner store}
\end{align*}
Using our store typing context $\Sigma$ which maps locations to types, we define the following 
judgement for well-typed stores:
\begin{align*}
    K \mid \Gamma \mid \Sigma \vdash \storeOuter
\end{align*}
This judgement asserts that all values in the store are coherent with the store
typing, and that the total size allocated for each region does not exceed the region's
specified maximum size. Figure~\ref{fig:st-outer} gives the typing rules 
for stores.
{\small{
\begin{figure}[H]
    \begin{align*}
        \begin{array}{c}
        \inferrule*[right=st-outer-empty]
            {\quad}
            {K \mid \Gamma \mid \Sigma \vdash \emptyset }
            \;\;\;
        \inferrule*[right=st-outer-region]
            {
             \text{currentSize}(\storeOuter(\rho)) \sqsubseteq s \\\\
                K \mid \Gamma \mid \Sigma \vdash \storeOuter \qquad
             K \mid \Gamma \mid \Sigma \vdash \storeInner_{\rho}^{in.} 
            }
            {K \mid \Gamma \mid \Sigma \vdash \storeOuter, \rho \mapsto (\storeInner_{\rho}^{in.}, s)}
        \end{array}
    \end{align*}
    \caption{Typing rules of $\storeOuter$}
    \label{fig:st-outer}  
\end{figure}
}} From these rules a store is well-typed if it is empty
(\textsc{st-outer-empty}) or (\textsc{st-outer-region}) if each of its inner
stores are well-typed, and the total size of values allocated in a region does
not exceed the region's maximum size $s$. 

The $\text{currentSize}$ function is given by Definition~\ref{def:currentSize},
mapping inner stores to the total size of values allocated in that inner store,
the dynamic counterpart to Definition~\ref{df:sumAllocs} in
Section~\ref{sec:static}:
\begin{df}[Current Size]
\begin{align*}
    \begin{array}{ll}
        \hspace{-1em}\text{currentSize}(\storeInner_{\rho}^{in.}) = \left\{\begin{matrix}\begin{array}{lll}
            \text{currentSize}(\emptyset, s) & = 0\\
            \text{currentSize}((\storeInner_{\rho}^{in'.}, l_\rho \mapsto v), s) & = \text{sizeOf}(v) + \text{currentSize}(\storeInner_{\rho}^{in'.}, s)
        \end{array}\end{matrix}\right.
    \end{array}
\end{align*}
\label{def:currentSize}
\end{df}
The function $\text{sizeOf}(v)$ maps values $v$ to sizes, given 
by Definition~\ref{def:sizeOf}: 
\begin{df}[Value Size]
Given a value $v$, sizeOf($v$) is defined:
\begin{align*}
        \text{sizeOf}(v) = \left\{\begin{matrix}\begin{array}{lll} 
            1 + |\text{freeLocs}(e)| & v = (\lambda x . e)\\
            1 + |\text{freeLocs}(e)| & v = (\Lambda \{ \alpha , \rho , \epsilon \} . \lambda x . e)\\
            1 & \text{otherwise}
    \end{array}\end{matrix}\right.  
\end{align*}
\label{def:sizeOf}
\end{df}
where $\text{freeLocs}(e)$ returns a list of locations in an expression $e$. 

A separate judgement types the inner store:
\begin{align*}
    K \mid \Gamma \mid \Sigma \vdash \storeInner_{\rho}^{in.}
\end{align*}
asserting that all values in the inner store are coherent with the store typing.
The typing rules for inner stores are straightforward and are given in
Figure~\ref{fig:st-inner}. {\small{
\begin{figure}[H]
    \begin{align*}
        \begin{array}{c}
        \inferrule*[right=st-inner-empty]
            {\quad}
            {K \mid \Gamma \mid \Sigma \vdash \emptyset}
        \;\;\;
        \inferrule*[right=st-inner-loc]
            { \tau = \Sigma(l_\rho) \qquad K \mid \Gamma \mid \Sigma \vdash \storeInner_{\rho}^{in.}  \\\\
             K \mid \Gamma \mid \Sigma \vdash v : \tau
            }
            {K \mid \Gamma \mid \Sigma \vdash \storeInner_{\rho}^{in.}, l_\rho \mapsto v}
        \end{array}
    \end{align*}
    \caption{Typing rules of $\storeInner_{\rho}^{in.}$}
    \label{fig:st-inner}
\end{figure}
}}

Figure~\ref{fig:dynamic-reduction-rules} gives the reduction rules for
\textsc{Spegion}. For brevity, we omit congruence rules which are standard. The
complete set of rules can be found in the appendix.
{\small{
\begin{figure}[h]
    \begin{align*}
        \begin{array}{c}
        \eNewrgn
        \\[1em]
        \eFreergnL
        \\[1em]
        \eSplitL
        \\[1em]
        \eCopyL
        \\[1em]
        \eValL
        \\[1em]
        \eAppL
        \\[1em]
        \eLetL
        \\[1em]
        \eIfFalse
        \;\;\;
        \eIfTrue
        \\[1em]
        \eRefL
        \\[1em]
        \eDerefL
        \\[1em]
        \eAssignL
        \;\;\;
        \eSeqNext
        \\[1em]
        \eBigApp
        \;\;\;
        \eFixL
        \end{array}
    \end{align*}
    \caption{Evaluation rules of \textsc{Spegion}}
    \label{fig:dynamic-reduction-rules} 
\end{figure}
}}

The ($\textsc{e-newrgn}$) rule allocates a new region in the store, with the
specified size $s$, and allocates and steps to a new location that points to a unit
value in the fresh region. The ($\textsc{e-freergnL}$) rule simply frees the
region $\rho$ from the store and steps to a pointer to the unit value in the
global region $l_{\globalRegion}^{1}$. 

Splitting a sub-region off from another region is handled by the
($\textsc{e-splitL}$) rule. The rule makes use of the size semiring's truncated
subtraction operation, given by Definition~\ref{df:monus}, to calculate the new
reduced maximum size of the parent region $\rho$. As in (\textsc{e-newrgn}), the
rule then allocates a new region $\rho'$, then allocates and steps to a location
in $\rho'$ pointing to a unit value to act as a pointer to the region.

The (\textsc{e-copyL}) rule copies a location from one region to another. A
fresh location $l'_{\rho'}$ is allocated in the region being copied to ($\rho'$)
pointing to the location being copied from $\rho$ ($l_{\rho}$).  

The reduction rule for allocation (\textsc{e-valL}) compares the programmer
annotated size against the actual size of the value being allocated (calculated
via~\ref{def:sizeOf}). If the annotated size is greater than or equal to the
actual size, the value is allocated into the region at a fresh location
$l'_{\rho}$ in the store. Type safety guarantees that $\rho$ has enough space to
allocate a value of the annotated size. If the annotated size is less than the
actual size, the program gets stuck. Type safety also ensures that such a
situation should never arise.

The reduction rules for application (\textsc{e-appL}), let binding
(\textsc{e-letL}), if statements (\textsc{e-ifFalse}), (\textsc{e-ifTrue}),
referencing (\textsc{e-refL}), dereferencing (\textsc{e-derefL}), assignment
(\textsc{e-assignL}), and sequential composition (\textsc{e-seqNext}) are
standard.

Type application (\textsc{e-bigApp}) is handled by the rule of the same name,
which is a no-op, since it is a purely static construct. Recursive functions are
handled by the (\textsc{e-fixL}) rule, which substitutes the location $l_\rho$,
where the recursive function $\Lambda \{\alpha, \rho, \epsilon\} . \lambda x .
e_1$ is stored, for the variable $f$ in $e_3$. 



\section{Properties of the Type System}
\label{sec:properties}
Type safety is guaranteed by Theorems~\ref{thm:progress}
and~\ref{thm:preservation}. Theorem~\ref{thm:progress} states the property of
\emph{progress}, i.e, a well-typed closed term is either a value, or can be
further reduced. Theorem~\ref{thm:preservation} states the property of type
\emph{preservation} (or \emph{subject reduction}), which states that if a term
is well-typed and can take a step of evaluation, then the resulting term is also
well-typed. The proofs are based on the syntactic type soundness proofs of
Helsen and Thiemann~\cite{HELSEN20011}, and can be found in
the extended version of this paper.

\begin{restatable}[Progress]{thm}{progress}
    \label{thm:progress}
If
\begin{align*}
    K \mid \Gamma \mid \Sigma \vdash e : (\tau, \rho) \mid \varphi 
\end{align*}
and 
\begin{align*}
    K \mid \Gamma \mid \Sigma \vdash \storeOuter 
\end{align*}
then either 
\begin{enumerate}[i]
    \item $e$ is a value or
    \item $e$ has the form ($x$) (a variable), with $x \in fv(e)$ or
    \item there exists an $e'$ such that $\langle e \mid \storeOuter \rangle
    \longrightarrow \langle e' \mid \storeOuter' \rangle$
\end{enumerate}
\end{restatable}

\begin{restatable}[Preservation]{thm}{preservation}
    \label{thm:preservation}
If
\begin{align*}
\begin{array}{ll}
   & K \mid \Gamma \mid \Sigma \vdash e : \mu \mid \varphi \\
   & K \mid \Gamma \mid \Sigma \vdash \storeOuter\\
   & \langle e \mid \storeOuter \rangle \longrightarrow \langle e' \mid \storeOuter' \rangle 
\end{array}
\end{align*}
for some $\Sigma'$ such that $\Sigma' \supseteq \Sigma$, we have that
\begin{align*}
\begin{array}{ll}
    & K \mid \Gamma \mid \Sigma' \vdash e' : \mu \mid \varphi' \\
    & K \mid \Gamma \mid \Sigma' \vdash \storeOuter' \\
\end{array}
\end{align*}
where 
\begin{align*}
\begin{array}{ll}
    & K \vdash \varphi' \sqsubseteq \varphi : \text{Effect}
\end{array}
\end{align*}
From our assumption of store typing correctness, we assert that $e$ is a
sub-derivation of some larger derivation $e''$ for which correctness with regard
to region size holds, i.e., given 
\begin{align*}
    K \mid \Gamma \mid \Sigma \vdash \storeOuter
\end{align*}
we have that: 
\begin{align*}
    \forall \rho \in \storeOuter . \forall l_\rho \in \storeInner_{\rho}^{in.} . l_{\rho} \in dom(\Sigma) \wedge \text{sizeOf}(\Sigma(l_{\rho})) \sqsupseteq \text{sizeOf}(\storeInner_{\rho}^{in.}(l_{\rho}))
\end{align*}
Finally, we assume the existence of a global region parametrising the calculus
which cannot be freed, as well as a location inside this global region of type
$\text{Unit}$: 
\begin{align*}
\begin{array}{ll}
& \forall \Sigma . l_{\rho_{\text{glob}}}^1 : \text{Unit} \in \Sigma
\end{array}
\end{align*}
\end{restatable}
Preservation states that if an expression $e$ is well-typed and has effect
$\varphi$, and a reduction step can be made, then the resulting expression $e'$
is also well-typed with an updated store typing $\Sigma'$ and an effect
$\varphi'$ such that $\varphi' \sqsubseteq \varphi$. 

The relation $\sqsubseteq$ defines effect \emph{subsumption}. Effect subsumption
is used to relate the effects of intermediate terms as a program is evaluated,
since evaluation does not preserve syntactic effects. For example, an allocation
expression has an effect which contains an allocation effect $\{\textbf{alloc}\
s\ \rho\}$. However, the reduction step for allocation expressions produces a
location, and locations are typed with an empty effect $\{\bot\}$. To the
programmer, these intermediate effects are not of interest - the effect of
importance is the overall effect of the program. In preservation, however, we
ensure that the intermediate effects are related to this overall effect via the
subsumption judgement $K \vdash \varphi_1 \sqsubseteq \varphi_2 : \text{Effect}$.
These rules are given in Figure~\ref{fig:prop-effect-subsumption}.
\begin{figure}[H]
    {\small{
        \begin{align*}
                \begin{array}{cc}
                        \sbEquiv
                        \;\;\;
                        \sbBot
                        \;\;\;
                        \sbXAbove
                        \\[1em]
                        \sbXBelowOne
                        \;\;\;
                        \sbXBelowTwo
                        \\[1em]
                        \sbJoinAbove
                        \;\;\;
                        \sbJoinBelowOne
                        \;\;\;
                        \sbJoinBelowTwo
                \end{array}
        \end{align*}
    }}
        \caption{Subsumption rules for \textsc{Spegion} ($\sqsubseteq$)}
        \label{fig:prop-effect-subsumption}
\end{figure}
\begin{figure}[H]
    {{
        \begin{align*}
                \begin{array}{cc}
                        \effectEqRefl 
                        \;\;\; 
                        \effectEqSym
                        \;\;\;
                        \effectEqTrans 
                        \\[1em]
                        \effectEqCong 
                        \\[1em] 
                        \effectEqAssoc
                        \\[1em]
                        \effectEqJoinAssoc
                        \\[1em]
                        \effectEqJoinCong
                \end{array}
        \end{align*}
    }}
        \caption{Effect equality rules for \textsc{Spegion} ($\equiv$)}
        \label{fig:prop-effect-equiv} 
\end{figure}
These rules are standard, however, note that we include kinding of
effects and we include rules for the join operator $\sqcup$. The
(\textsc{sb-$\equiv$}) rule requires a judgement of effect equivalence, which is
given in Figure~\ref{fig:prop-effect-equiv}. Again, these rules are completely
standard.

\section{Applications}
\label{sec:applications}
\subsection{System Code Examples}
\label{subsec:code-examples}

\begingroup
\allowdisplaybreaks

Various systems programming idioms are expressible in the
language. Below, we present examples first in C and then in
\textsc{Spegion}. In each case, the translation is shallow: region-size
annotations required by our language can be inferred from the C program. In
one example, our language exceeds what can be expressed in C. In all
cases, the type system statically recognises memory safety errors.

\medskip
\subparagraph*{Use After Free.}
In the C code below, a struct with two members is allocated, written
to, freed, and then erroneously read from. This is \emph{use after
free}, a failure of memory safety.

\begin{lstlisting}[language=C,basicstyle=\ttfamily\footnotesize]
struct mine {int a; int b;};
  
int free_help(int* p) { free(p); }

int main () {
  struct mine *mp = malloc(sizeof(struct mine));
  mp->a = 0;
  mp->b = 1;
  int *bp = &(mp->b);
  free_help(mp);
  return *bp;    // uaf
}
\end{lstlisting}

\noindent The translation of this code to our language is direct: the size of
the new region is taken from the call to \texttt{malloc} in the C code.

\begin{align*}
  &    \textbf{let}\ r = \textbf{newrgn}\ [2]\ \textbf{in}\ \\
  &    \textbf{let}\ (x, y) = (\textbf{ref}\ 0, \textbf{ref}\ 0)\ [2]\ \textbf{at}\ r\ \textbf{in}\ \\
  &    \textbf{let}\ x := 0\ [1]\ \textbf{at}\ \globalRegion\ \textbf{in}\ \\
  &    \textbf{let}\ y := 1\ [1]\ \textbf{at}\ \globalRegion\ \textbf{in}\ \\
  &    \textbf{let}\ bp = y\ \textbf{in}\ \\
  &    \textbf{freergn}\ r;\ \\
  &    \textbf{let}\ b =\, !bp\ \\
\end{align*}

\noindent This program does not pass the static semantics because on
the last line, the type effects record that region $r$ was previously
freed, and in typing the final line, the \textsc{t-ref} rule requires
the region of $bp$ to be live.

\medskip
\subparagraph*{Recursion.} Below, a function allocates, calls itself
recursively, and then frees its prior allocation. There is no memory
safety error: despite being nested by the recursion, all of the
allocations are matched by frees.

\begin{lstlisting}[language=C,basicstyle=\ttfamily\footnotesize]
int alloc_rec_free (int n) {
  if (n) {
    int *p = malloc(sizeof(int));
    int  v = alloc_rec_free (n - 1);
    free(p);
    return v;
  } else
    return 0;
}

int main () {
  return alloc_rec_free(10);
}
\end{lstlisting}

\noindent The translation uses the \textbf{fix} constructor. Again
the translation is shallow: none of the additional annotations in our
language require any analysis to concoct. Indeed the function's type
scheme is empty because it conveys no change to the live regions
to its context.

\begin{align*}
  &  \textbf{let}\ \text{arf} = \textbf{fix}(\text{arf}, (\Lambda \emptyset . \lambda n .\\
  &  \quad\textbf{if}\ (n == 0\ [1]\ \textbf{at}\ \globalRegion)\ \textbf{then}\ \\
  &  \quad\quad\textbf{let}\ r = \textbf{newrgn}\ [1]\ \textbf{in}\ \\
  &  \quad\quad\textbf{let}\ p = \textbf{ref}\ 0\ [1]\ \textbf{at}\ r\ \textbf{in}\ \\
  &  \quad\quad\textbf{let}\ v = (\text{arf}\ @\ \emptyset)\ (n - 1\ [1]\ \textbf{at}\ \globalRegion)\ \textbf{in}\ \\
  &  \quad\quad\textbf{freergn}\ p;\ \\
  &  \quad\quad v\ \\
  &  \quad\textbf{else}\ \\
  &  \quad\quad 0\ [1]\ \textbf{at}\ \globalRegion\ \\
  &  )\ [1]\ \textbf{at}\ \globalRegion)\ \textbf{in}\ \\ 
  &  (\text{arf}\ @\ \emptyset)\ (10\ [1]\ \textbf{at}\ \globalRegion)\
\end{align*}

\noindent The static semantics accepts this program and it is memory safe.

\medskip
\subparagraph*{Modelling Loops.}  Loops can be translated to recursive
functions. The C code below is analogous to the code above but it uses
a loop. The loop body allocates, it records the value of the index and
then frees the prior allocation.

\begin{lstlisting}[language=C,basicstyle=\ttfamily\footnotesize]
int main () {
  int v;
  for (int n = 10; n > 0; n--) {
    int *p = malloc(sizeof(int));
    v = n;
    free(p);
  }  
  return v;
}
\end{lstlisting}

\noindent The loop is translated into a recursive function whose argument serves
as the loop index.

\begin{align*}
  &  \textbf{let}\ \text{loop} = \textbf{fix}(\text{loop}, (\Lambda \emptyset . \lambda n .\\
  &  \quad\textbf{if}\ (n > 0\ [1]\ \textbf{at}\ \globalRegion)\ \textbf{then}\ \\
  &  \quad\quad\textbf{let}\ r = \textbf{newrgn}\ [1]\ \textbf{in}\ \\
  &  \quad\quad\textbf{let}\ p = \textbf{ref}\ (0\ [1]\ \textbf{at}\ \globalRegion)\ [1]\ \textbf{at}\ r\ \textbf{in}\ \\
  &  \quad\quad\textbf{let}\ v := n\ \textbf{in}\ \\
  &  \quad\quad\textbf{freergn}\ r;\ \\
  &  \quad\quad(\text{loop}\ @\ \emptyset)\ (n - 1\ [1]\ \textbf{at}\ \globalRegion)\ \\
  &  \quad\textbf{else}\ ()\ \\
  &  )\ [1]\ \textbf{at}\ \globalRegion)\ \textbf{in}\ \\
  &  (\text{loop}\ @\ \emptyset)\ (10\ [1]\ \textbf{at}\ \globalRegion);\ \\
  &  v\ 
\end{align*}

\medskip
\subparagraph*{Pointer Arithmetic and Finite Buffers.}  The C code
below produces and consumes packets, serialising them in a
buffer. \textsc{Spegion} tracks pointer arithmetic and will catch
patterns of use that exhaust the size of the buffer.

\begin{lstlisting}[language=C,basicstyle=\ttfamily\footnotesize]
struct packet {int len; int payload[];};

struct packet * produce (struct packet *p, int data) {
  p->len = sizeof(int);  p->payload[0] = data;
  return p->payload + p->len;
}

struct packet * consume (struct packet *p, int* sum) {
  *sum += p->payload[0];
  return p->payload + p->len;
}

int main () {
  int sum = 0;
  void *buf = malloc(1000);
  struct packet *pp = buf, *cp = buf;
  pp = produce(pp, 1);  pp = produce(pp, 2);
  cp = consume(cp, &sum);  cp = consume(cp, &sum);
  return sum;
}
\end{lstlisting}

\noindent The \texttt{main} function translates into the following
\textsc{Spegion}. Statically, the loads and stores of calls to $\text{produce}$
and $\text{consume}$ are checked against the bounds of the region of
$\text{buf}$. The return pointers type check, even if they escape the region of
$\text{buf}$, matching idiomatic use of pointer arithmetic in C.

\begin{align*}
  & \textbf{let}\ \text{r1} = \textbf{newrgn}\ \textbf{in} \\ 
  & \textbf{let}\ \text{sum} = \textbf{ref}\ (0\ \textbf{at}\ \text{r1})\ \textbf{in} \\ 
  & \textbf{let}\ \text{buf} = \textbf{newrgn}\ [1000]\ \textbf{in} \\ 
  & \textbf{let}\ \text{init} = \textbf{ref}\ (\text{initPacket}\ [1]\ \textbf{at}\ \text{buf})\ \textbf{in} \\ 
  & \textbf{let}\ (\text{pp},\ \text{cp}) = !\text{init}\ \textbf{in} \\
  & \textbf{let}\ \text{pp} = \text{produce}(\text{pp},\ 1)\ \textbf{in} \\
  & \textbf{let}\ \text{pp} = \text{produce}(\text{pp},\ 2)\ \textbf{in} \\  
  & \textbf{let}\ \text{cp} = \text{consume}(\text{cp},\ \text{sum})\ \textbf{in} \\
  & \textbf{let}\ \text{cp} = \text{consume}(\text{cp},\ \text{sum})\ \textbf{in} \\  
  & !\text{sum}
\end{align*}

\medskip
\subparagraph*{Non-Lexically-Scoped Lifetimes.}  The C code of
Figure~\ref{listing:linked_list} allocates a linked list of three
tasks in list order, takes a reference to the last task, and then
frees the first two tasks. \textsc{Spegion} can represent and statically type
programs like this. The translation is direct:

\begin{align*}
  &  \textbf{let}\ \text{re} = \textbf{newrgn}\ [1]\ \textbf{in}\ \\
  &  \textbf{let}\ \text{end} = 0\ [1]\ \textbf{at}\ \text{re}\ \textbf{in}\ \\
  &  \\    
  &  \textbf{let}\ r_0 = \textbf{newrgn}\ [2]\ \textbf{in}\ \\
  &  \textbf{let}\ t_0 = (0\ [1]\ \textbf{at}\ \globalRegion,\ \text{end})\ [2]\ \textbf{at}\ r_0\ \textbf{in}\ \\
  &  \textbf{let}\ (n_0,\text{next}_0) = t_0\ \textbf{in}\ \\
  &  \\
  &  \textbf{let}\ r_1 = \textbf{newrgn}\ [2]\ \textbf{in}\ \\
  &  \textbf{let}\ t_1 = (0\ [1]\ \textbf{at}\ \globalRegion,\ \text{end})\ [2]\ \textbf{at}\ r_1\ \textbf{in}\ \\
  &  \textbf{let}\ (n_1,\text{next}_1) = t_1\ \textbf{in}\ \\
  &  \textbf{let}\ \text{next}_0 = t_1\ \textbf{in}\ \\
  &  \\
  &  \textbf{let}\ r_2 = \textbf{newrgn}\ [2]\ \textbf{in}\ \\
  &  \textbf{let}\ t_2 = (0\ [1]\ \textbf{at}\ \globalRegion,\ \text{end})\ [2]\ \textbf{at}\ r_2\ \textbf{in}\ \\
  &  \textbf{let}\ (n_2,\text{next}_2) = t_2\ \textbf{in}\ \\
  &  \textbf{let}\ \text{next}_1 = t_2\ \textbf{in}\ \\
  &  \\
  &  \textbf{let}\ n_2 = 1\ [1]\ \textbf{at}\ \globalRegion\ \textbf{in}\ \\
  &  \\
  &  \textbf{freergn}\ t_0;\ \\
  &  \textbf{freergn}\ t_1;\ \\
  &  \textbf{let}\ (n, \_) =\, !t_2\ \textbf{in}\ \\
  &  n
\end{align*}

\medskip
\subparagraph*{Splitting.}  Our language represents idioms that C
cannot represent, in particular \emph{splitting}. In this idiom, the
programmer creates a region. Subsequently they split the region into a
part that is freed, and a part that is kept and used.

C cannot represent this idiom, but the C code below demonstrates the
idea. The \texttt{alloc2} function mirrors the initial allocation, but
to match the limitations of C it is implemented with an indirection
and two allocations matching the subsequent split. The first part of
the allocation made by \texttt{alloc2} is freed, and the second part
is used.

\begin{lstlisting}[language=C,basicstyle=\ttfamily\footnotesize]
struct indirect {int* p1; int* p2;};

struct indirect * alloc2 () {
  struct indirect *p = malloc(sizeof(struct indirect));
  p->p1 = malloc(sizeof(int));
  p->p2 = malloc(sizeof(int));
  return p;
}

int main () {
  struct indirect *pr = alloc2();
  *(pr->p2) = 1;
  free(pr->p1);
  return *(pr->p2);
}
\end{lstlisting}

\noindent Our language can represent this programming idiom. The
region $r$ is created in one call to \textbf{newrgn}, and then split
into two parts. These parts can be used as a pair, but they can also
be independently freed.

\endgroup

\begin{align*}
  &  \textbf{let}\ r = \textbf{newrgn}\ [2]\ \textbf{in}\ \\
  &  \textbf{let}\ r_1 = \textbf{split}\ [1]\ \textbf{in}\ \\    
  &  \textbf{let}\ v_1 = (0\ [1]\ \textbf{at}\ \globalRegion)\ [1]\ \textbf{at}\ r_1\ \textbf{in}\ \\
  &  \textbf{let}\ v_2 = (0\ [1]\ \textbf{at}\ \globalRegion)\ [1]\ \textbf{at}\ r\ \textbf{in}\ \\
  &  \textbf{let}\ (p_1,p_2) = (v_1,v_2)\ [2]\ \textbf{at}\ \globalRegion\ \textbf{in}\ \\
  &  \textbf{let}\ p_2 := 1\ [1]\ \textbf{at}\ \globalRegion \textbf{in}\ \\
  &  \textbf{freergn}\ p_1;\ \\
  &  !p_2
\end{align*}

\subsection{Refinement Types}
\label{subsec:refinement-types} 
Refinement types give the programmer the ability to use \emph{predicates} in the
type as a means of restricting the values described by the
type~\cite{10.1145/113446.113468, 10.1145/944746.944725}. For example, the type:
\begin{align*}
     n : \text{Int} \rightarrow \text{Int}\ [ n' \mid n' \leq n  ]
\end{align*}
describes a function from integers to integers which returns a value  
that must be less than or equal to the input. This is a simple example of a refinement type,
but they can be used to encode many forms of specification. 

A natural extension of our work would be to incorporate refinement types into
our type system, allowing the programmer to reason about region sizes in the
types of their programs. In this section we sketch out a simple refinement
system based on the approach of Jhala and Vazou~\cite{10.1561/2500000032} that
could be used in tandem with \textsc{Spegion}'s region sizes and sized
allocations. We define a new sort of \emph{predicates} which can be used in the
type system to restrict the values described by a type. This section is mostly a
sketch of the ideas, and we leave the full development of the system to future
work. We begin with the syntax of predicates: 
\begin{align*}
    \begin{array}{rl}
        p :=& x \mid n \mid \text{true} \mid \text{false} \mid p + p \mid p \cdot p \mid \neg p \mid \textbf{if}\ p\ \textbf{then}\ p\ \textbf{else}\ p \mid f (\overline{p})
    \end{array}
    \tag{predicates}
\end{align*}
Predicates are drawn from quantifier-free fragment of linear arithmetic and
uninterpreted functions~\cite{Barrett2010TheSS}. That is, a predicate $p$ may be
a variable $x$, a natural number constant $n$, boolean constants
$\text{true}$ and $\text{false}$, the sum or product of two predicates, the
negation of a predicate, a conditional, or the application of an uninterpreted
function $f$ to a list of predicates. For this section, we extend
$\textsc{Spegion}$ with type annotations and a new list datatype:
\begin{align*}
    \begin{array}{rl}
        \text{List}\ (\tau, \rho) = \text{Nil} : \text{List}\ (\tau, \rho) \mid \text{Cons} : (\tau, \rho) \xrightarrow{\{\bot\}} \text{List}\ (\tau, \rho) \xrightarrow{\{\bot\}} \text{List}\ (\tau, \rho)
    \end{array}
\end{align*}
That is, a list is a value which is either empty or a pair of a value of type
$(\tau, \rho)$ and a reference to another list. For simplicity the payload of
the list must be stored in the same region as the list itself. One could
envision a more complex type where the payload and indeed the tail of the list
could be stored in separate regions. With the above type we can construct a
function which takes an integer $n$ and creates a list in a region of arbitrary
size that has length $n$:
\begin{align*}
    \begin{array}{ll}
        \textbf{let}\ r = \textbf{newrgn}\ \\
        \ \ \ \ \textbf{in} \\
        \ \ \ \ \textbf{let}\ \text{buildList} = \textbf{fix} (\text{buildList}, \Lambda \{ \alpha, \rho, \epsilon \} \lambda n . \\
        \ \ \ \ \ \ \ \ \textbf{if}\ n == 0 \\
        \ \ \ \ \ \ \ \ \ \ \ \ \  \textbf{then}\\
        \ \ \ \ \ \ \ \ \ \ \ \ \ \ \ \ \ \text{Nil}\ \textbf{at}\ r\\
        \ \ \ \ \ \ \ \ \ \ \ \ \  \textbf{else}\\
        \ \ \ \ \ \ \ \ \ \ \ \ \ \ \ \ \ \text{Cons}\ (n, \text{buildList}\ @(\text{Int}, \text{regionOf}(r))\ n-1))\ \textbf{at}\ r \\
        \textbf{in}\ ... 
    \end{array}
\end{align*}
With refinement types we could constrain this function to only return lists 
where the size of the list is less than or equal to the size of the region:
\begin{align*}
    \begin{array}{ll}
        \textbf{let}\ r = \textbf{newrgn}\ [5] \\
        \textbf{in} \\
        \ \ \ \ \textbf{let}\ \text{buildList} : ((n : \text{Int}, \text{regionOf}(r)) \xrightarrow{\varphi} (\text{List}\ (\text{Int}, \text{regionOf}(r)))\ [1 + (2 \cdot n) < 5] ) =\\
        \ \ \ \ \ \ \ \ \textbf{fix} (\text{buildList}, \Lambda \{ \alpha, \rho, \epsilon \} \lambda n . \\
        \ \ \ \ \ \ \ \ \ \ \ \ \textbf{if}\ n == 0 \\
        \ \ \ \ \ \ \ \ \ \ \ \ \ \ \ \ \  \textbf{then}\\
        \ \ \ \ \ \ \ \ \ \ \ \ \ \ \ \ \ \ \ \ \ \text{Nil}\ [1]\ \textbf{at}\ r\\
        \ \ \ \ \ \ \ \ \ \ \ \ \  \ \ \ \ \textbf{else}\\
        \ \ \ \ \ \ \ \ \ \ \ \ \ \ \ \ \ \ \ \ \ \text{Cons}\ (n\ \textbf{at}\ [1], \text{buildList}\ @(\text{Int}, \text{regionOf}(r))\ n-1))\ [1]\ \textbf{at}\ r \\
        \ \ \ \ \textbf{in}\ ... 
    \end{array}
\end{align*}
where $\varphi = { \{ \textbf{alloc}\ 1\ \text{regionOf}(r)\} \sqcup (\{\epsilon
\} \times \{ \textbf{alloc}\ 1\ \text{regionOf}(r) \} ) }$. Note we remove
superfluous $\{\bot\}$ effects from the type annotation for brevity. Thus if we
try to apply $\text{buildList}$ to the integer $10$, this application will fail
as the size of the list produced is greater than the size of the region
associated with $r$. This is just a taste of what the interaction between refinement types and region
sizes can offer. We believe that adding refinements to \textsc{Spegion} is a
natural extension of the language, with few changes required to the core
calculus.

%
%

\section{Related Work}
\label{sec:related}
Key areas of related work for this research relate to the Rust programming language, and effect and region types in programming language design.

\subsection{Rust}

Rust is a popular systems programming language that aims to statically enforce memory safety, and it has also been the subject of significant academic research.
Notable formal work includes Rust Belt~\cite{10.1145/3158154} and Oxide~\cite{weiss2021oxideessencerust}.
The key concepts involved in Rust's approach to enforcing memory safety are \emph{ownership} and \emph{lifetimes}, with the latter corresponding to regions.
As a key component of enforcing safety guarantees, Rust often limits or prohibits the introduction of \emph{aliases} of pointers to mutable data, giving it a similar flavour to the substructural type system approaches detailed later in this section.
%
%
This is, in part, because Rust is also concerned with preventing data races in concurrent code.

The type system might be extended to the concurrent context by reconciling the sequence of effects with the orderings imposed by a declarative concurrency semantics, like those of sequential consistency, ARMv8, or C++, but this is left to future work.
Even in sequential code, however, there are difficulties involved in aliasing pointers and (for example) building cyclic data structures.

The following program is rejected by the Rust compiler:
\begin{lstlisting}[language=Rust,basicstyle=\ttfamily\footnotesize]
fn main() {
    let mut a: i32 = 10;
    let b: &i32 = &a;

    {
        let c: &mut i32 = &mut a;
        *c = 20;
    }
    // use a,b ...
}
\end{lstlisting}
The compiler will point out \lstinline{a} being borrowed multiple times (immutably by \lstinline{b} and mutably by \lstinline{c}).
This program can be represented in our calculus with no issue, as our approach allows arbitray aliasing of mutable references.
In general, Rust's expressive type system supports many patterns of memory usage, and we believe there is potential in combining ideas from Rust and from our effect-based approach.

\subsection{Effect Types and Region Calculi}\label{subsec:regions}

There is extensive prior work on using effects and regions in programming language design.
Starting in the 1980s, Lucassen and Gifford, among others, explored the development of polymorphic effect type systems~\cite{10.1145/73560.73564, 10.1145/319838.319848}.
They made use of a \emph{kind} system consisting of types, which describe what sort of value an expression may evaluate to; effects, which describe what side effects may be performed when evaluating an expression; and regions, which describe groups of values in memory that are related.
Building upon the idea of associating types with regions, Tofte and Talpin introduced the region calculus and region-based memory management~\cite{10.1145/174675.177855,TOFTE1997109}, which provided a structured way to control memory in higher order programs by associating types with parameters to statically track their lifetimes.

Region-based memory management formed the basis of ML-Kit, a compiler for Standard ML that used \emph{region inference} to infer region annotations of source programs and statically determine memory usage~\cite{10.1145/291891.291894, 10.1145/237721.237771}.
Rather than requiring programmers to manually annotate their programs with region/lifetime parameters, the compiler determines them automatically.
Subsequent work investigated improving region inference, as it was prone to potentially producing programs that leak memory or unnecessarily extend the lifetime of objects~\cite{10.1145/223428.207137, 10.1023/B:LISP.0000029446.78563.a4}.
ML-Kit was augmented with a runtime garbage collector in order to collect memory left around inside regions which had become garbage but had not yet been de-allocated because it lived inside a still-alive region~\cite{10.1145/543552.512547}.

Existing papers have proven type safety theorems for variations of the region calculus.
Helsen and Thiemann give an elegant syntactic proof of type safety for a region calculus with polymorphic types and recursive let-bindings, which was a primary source for the formalism in this paper~\cite{HELSEN20011}.
They represent region de-allocation by introducing a special ``dead'' region identifier, and when a region goes out of scope they substitute this special identifier for the now-dead reagion.
This also means they do not need to have a phase distinction between handles and region identifiers.
However, they do not have a store in their operational semantics, which leads to a straightforward proof but unfortunately does not allow for mutable data.
Subsequent work extended this approach with an explicit store, but had to remove polymorphic types~\cite{10.1006/inco.2001.3112}.
Prior work also investigated proof of soundness for a fragment of the region calculus via a translation into a target language~\cite{10.5555/788021.788961}.

One of the major issues with using region-based memory management as proposed in the original region calculus is that the lexical scoping of regions means they are required to follow a stack or last-in-first-out (LIFO) ordering.
This was observed by Aiken et al., demonstrating that extending the traditional region calculus by removing the LIFO restriction dramatically reduces memory usage (sometimes asymptotically) in functional programs~\cite{10.1145/223428.207137}.
Like our approach, they de-couple allocation and de-allocation, though their approach is for functional programs without mutable memory.
They are also focused on region annotation inference in order to statically determine memory lifetimes of ML programs, and so they do not consider bounding or splitting regions and cannot track fine-grained effect details like an allocation consuming a specific amount of memory in a region.

Relatedly, Boudol presented a monomorphic variant of the region calculus that permitted early de-allocation~\cite{10.1007/978-3-540-78739-6_10}.
A type system tracks both negative de-allocation effects and positive effects such as reads and writes, ensuring that the positive effects on a region cannot occur after the negative effect.

Further work explored adding explicit region allocation and de-allocation to a C-like language, where the safety of region de-allocation relied on keeping track of a dynamic \emph{reference count} for objects in the heap~\cite{10.1145/277652.277748, 10.1145/381694.378815}.

A prominent line of research that addressed the problem of LIFO lifetimes of regions is the use of \emph{substructural types} in region type systems~\cite{10.1007/11693024_2}.
The Cyclone programming language, and related work on the calcluls of capabilities, allowed for type-safe manual memory management using a combination of type system features including substuctural types, allowing for both LIFO regions as well as first-class dynamic regions and unique pointers~\cite{10.1145/292540.292564, 10.1007/11693024_2}.
This style of type system requires that programmers thread resources linearly through programs and avoid creating aliases.
In contrast, a primary goal of our system is to avoid imposing this requirement of linearity in programs, and allowing for the aliasing of pointers that occurs in idiomatic imperative code.

Region types have also been used in programming language design for reasons beyond memory management.
For example, in Gibbon, region types are used in combination with an effect type system to type programs that create and manipulate serialised data, where each region contains a serialised data structure~\cite{10.1145/3314221.3314631}.







\section{Conclusion}
\label{sec:conclusion}
We have proposed a type system for region-based memory management that allows
implicit regions with non-lexical scoping and explicit size constraints. We have
proven type safety for this system, and shown several directions in
which this system can be extended. Based on this type system, we are currently
developing an implementation of \textsc{Spegion} in the form of both an
interpreter for \textsc{Spegion} as well as a static analysis tool targeting the
C language.

\bibliography{references}

\newpage
\appendix
\label{appendix:appendix}
\section{Grammar}
\label{appendix:grammar}

\begin{figure}[h]
\begin{align*}
        \begin{array}{rl}
        \kappa ::= & \text{Type} \mid \kappa \rightarrow \kappa \mid \text{Region} \mid \text{Effect} \mid \text{Size} \\
        \tau ::=& \alpha \mid 
                  \text{Int} \mid \text{Unit} \mid \text{Bool} \mid \text{Ref}\ \tau \mid 
                  \mu \xrightarrow{\varphi} \mu  \mid \forall \{ \alpha, \rho, \epsilon \}. \mu \xrightarrow{\varphi} \mu \\ 
        \mu  ::=& (\tau, \rho)\\
        op ::=& + \mid \cdot \mid \dot - \mid = \mid \neq \mid \sqsubseteq \mid \sqsupseteq \\
        s ::=& x \mid n \mid \omega \mid s\ op\ s\\
        \\
        v ::= & n \mid \text{true} \mid \text{false} \mid \Lambda \{ \alpha, \rho, \epsilon \} . \lambda x . e \mid \lambda x . e \mid () \mid l_\rho\\
        e ::= & x,f  \\
             \mid & v\ [s]\ \textbf{at}\ e\\
             \mid & l_\rho\\
             \mid & e\ e \\ 
             \mid & \textbf{newrgn}\ [s]\\ 
             \mid & \textbf{freergn}\ e \\ 
             \mid & \textbf{split}\ [s]\ e\\
             \mid & \textbf{copy}\ e\ \textbf{into}\ e \\
             \mid & \textbf{ref}\ e \\ 
             \mid & \textbf{if}\ e\ \textbf{then}\ e\ \textbf{else}\ e \\ 
             \mid & !e \\
             \mid & e := e \\ 
             \mid & e ;\ e \\
             \mid & \textbf{let}\ x : \mu = e\ \textbf{in}\ e\\
             \mid & e\ @\ \mu\\
             \mid & \textbf{let}\ f = \textbf{fix}(f, (\Lambda \{ \alpha , \rho , \epsilon \} . \lambda x . e_1)\ [s]\ \textbf{at}\ e_2)\ \textbf{in}\ e_3\\
        \\
        \varphi ::= & \varphi \times \varphi \mid \{\bot\} \mid \{\textbf{fresh}\ \rho\ s \} \mid \{ \textbf{free}\ \rho\} \mid \{\textbf{split}\ \rho\ s\ \rho'\} \mid \{  \textbf{alloc}\ s\ \rho\}  
          \mid \varphi \sqcup \varphi \\ \mid &  \{\epsilon\} \mid \{\textbf{rec}\ \epsilon\ \varphi\} 
        \\
        \Gamma ::= & \emptyset \mid \Gamma, x : \mu \\ 
        \Sigma ::= & \emptyset \mid \Sigma, l_\rho : \tau  \\
        K ::= & \emptyset \mid K, \alpha : \text{Type} \mid K, \rho : \text{Region} \mid K, \epsilon : \text{Effect} 
        \end{array}
\end{align*}
\label{fig:grammar}
\caption{Grammar of \textsc{Spegion}}
\end{figure}

\newpage
\section{Typing}
\label{appendix:typing}

\subsection{Kinding Rules}

\begin{figure}[h]
\begin{align*}
        \begin{array}{cc}
                \kVar 
                \;\;\;
                \kForall
                \;\;\;
                \kRegion
                \\[1em]
                \kArrow
                \;\;\;
                \kApp
                \\[1em]
                \kAppTwo
                \;\;\;
                \kInt 
                \;\;\; 
                \kUnit
                \\[1em]
                \kRef
                \;\;\;
                \kBool
                \;\;\;
                \kSize
                \\[1em]
                \kOp
                \;\;\;
                \kTyWithPlace
                \\[1em]
                \kBot
                \;\;\;
                \kCompose
                \;\;\;
                \kJoin
                \\[1em]
                \kAlloc
                \;\;\;
                \kNew
                \\[1em]
                \kSplit
                \\[1em]
                \kFree
                \;\;\;
                \kEps 
                \;\;\; 
                \kRec
        \end{array}
\end{align*}
\caption{Kinding rules for \textsc{Spegion}}
\label{fig:kinding}
\end{figure}

\subsection{Effect Rules}
\begin{figure}[H]
        \begin{align*}
                \begin{array}{cc}
                        \effectEqRefl 
                        \;\;\; 
                        \effectEqSym
                        \;\;\;
                        \effectEqTrans 
                        \\[1em]
                        \effectEqCong 
                        \\[1em] 
                        \effectEqAssoc
                        \\[1em]
                        \effectEqJoinAssoc
                        \\[1em]
                        \effectEqJoinCong
                \end{array}
        \end{align*}
        \caption{Rules for effect equality ($\equiv$)}
\end{figure}

\begin{figure}[h]
        \begin{align*}
                \begin{array}{cc}
                        \sbEquiv
                        \;\;\;
                        \sbBot
                        \;\;\;
                        \sbXAbove
                        \\[1em]
                        \sbXBelowOne
                        \;\;\;
                        \sbXBelowTwo
                        \\[1em]
                        \sbJoinAbove
                        \;\;\;
                        \sbJoinBelowOne
                        \;\;\;
                        \sbJoinBelowTwo
                \end{array}
        \end{align*}
        \caption{Rules for effect subsumption ($\sqsubseteq$)}
\end{figure}

\begin{figure}
        \begin{align*}
                \begin{array}{cc}
                        \composeEBot
                        \;\;\;
                        \composeEFresh
                        \;\;\;
                        \composeEFree 
                        \\[1em]
                        \composeEFreshAlloc 
                        \;\;\;
                        \composeEAlloc 
                        \\[1em]
                        \composeESplitAlloc 
                        \;\;\;
                        \composeEFreshSplit  
                        \\[1em]
                        \composeESplit
                        \;\;\;
                        \composeESplitSplit 
                        \;\;\;
                        \composeEJoin
                        \\[1em] 
                        \composeEVarRight
                        \;\;\; 
                        \composeEVarLeft 
                        \;\;\;
                        \composeERec
                \end{array}
        \end{align*}
        \caption{Rules for effect composition}
\end{figure}

\newpage

\subsection{Typing Rules}

\begin{figure}[h]
\begin{gather*}
        \begin{array}{cc}
        \intT
        \;\;
        \unitT
        \\[-0.3em]
        \trueT 
        \;\;
        \falseT
        \\[1em]
        \locValT
        \;\;
        \absT
        \\[1em]
        \biglamT
        \end{array}
\end{gather*}
\caption{Typing rules for values in \textsc{Spegion}}
\end{figure}

\begin{figure}
        \vspace{-5em}
\begin{gather*}
        \hspace{-5em}
        \begin{array}{cc}
                \uselocT 
                \;\;
                \valT 
                \\[1em]
                \newrgnT
                \\[1em]
                \freergnT
                \\[1em]
                \splitT
                \\[1em]
                \copyT
                \\[1em]
                \varT 
                \;\;\;
                \letT
                \\[1em]
                \ifT
                \\[1em]
                \refT
                \;\;\;
                \derefT
                \\[1em]
                \assignT
                \;\;\;
                \seqT
                \\[1em]
                \appT
                \\[1em]
                \bigAppT
                \\[1em]
                \fixT
        \end{array}
\end{gather*}
\caption{Typing rules for expressions in \textsc{Spegion}}
\end{figure}


\newpage
\section{Semantics}
\label{appendix:semantics}

\begin{align*}
    \storeInner_{\rho}^{in.} ::= \emptyset \mid \storeInner_{\rho}^{in.}, l_\rho \mapsto v \mid \storeInner_{\rho}^{in.}, l_\rho \mapsto a
\tag{inner store}
\end{align*}

\begin{align*}
    \storeOuter ::= \emptyset \mid \storeOuter, \rho \mapsto (\storeInner_{\rho}^{in.},s_a)
\tag{outer store}
\end{align*}




\subsection{Evaluation Rules}

{\small{
\begin{figure}[H]
\begin{align*}
    \begin{array}{c}
        \eNewrgn
        \\[1.5em]
        \eFreergn
        \\[1.5em]
        \eFreergnL
        \\[1.5em]
        \eSplit
        \\[1.5em]
        \eSplitL
        \\[1.5em]
        \eCopyOne
        \\[1.5em]
        \eCopyTwo
        \\[1.5em]
        \eCopyL
        \\[1.5em]
        \eVal
        \\[1.5em]
        \eValL
    \end{array}
\end{align*}
\label{fig:op-sem}
\caption{Small step rules}
\end{figure}
}}

{\small{
\begin{figure}[H]
\begin{align*}
    \begin{array}{c}
        \eAppOne
        \\[1.5em]
        \eAppTwo
        \\[1.5em]
        \eAppL
        \\[1.5em]
        \eRef 
        \\[1.5em]
        \eRefL
        \\[1.5em]
        \eDeref
        \\[1.5em]
        \eDerefL
        \\[1.5em]
        \eAssignOne
        \\[1.5em]
        \eAssignTwo
        \\[1.5em]
        \eAssignL
        \\[1.5em]
        \eSeq
        \\[1.5em]
        \eSeqNext
    \end{array}
\end{align*}
\label{fig:op-sem-2}
\caption{Small step rules (continued)}
\end{figure}
}}

{\small{
\begin{figure}[H]
\begin{align*}
    \begin{array}{c}
        \eBigApp
        \\[1.5em]
        \eIf
        \\[1.5em]
        \eIfTrue
        \\[1.5em]
        \eIfFalse
        \\[1.5em]
        \eLet
        \\[1.5em]
        \eLetL
        \\[1.5em]
        \eFix
        \\[1.5em]
        \eFixL
    \end{array}
\end{align*}
\label{fig:op-sem-3}
\caption{Small step rules (continued)}
\end{figure}
}}

\newpage
\subsection{Store Typing Rules}
{\small{
\begin{figure}[H]
\begin{align*}
    \hspace{-5em}
    \begin{array}{c}
        \inferrule*[right=emptyOuter]
            {\quad}
            {K \mid \Gamma \mid \Sigma \vdash \emptyset }
        \;\;\;
        \inferrule*[right=inner]
            {K \mid \Gamma \mid \Sigma \vdash \storeOuter \qquad
             K \mid \Gamma \mid \Sigma \vdash \storeInner_{\rho}^{in.} 
             \qquad \text{currentSize}(\rho) \sqsubseteq s_a
            }
            {K \mid \Gamma \mid \Sigma \vdash \storeOuter, \rho \mapsto (\storeInner_{\rho}^{in.}, s_a)}
    \end{array}
\end{align*}
\caption{Typing rules for $\storeOuter$}
\end{figure}
}}

{\small{
\begin{figure}[H]
\begin{align*}
    \hspace{-5em}
    \begin{array}{c}
        \inferrule*[right=emptyInner]
            {\quad}
            {K \mid \Gamma \mid \Sigma \vdash (\emptyset, s_a)}
        \;\;\;
        \inferrule*[right=loc]
            { \tau = \Sigma(l_\rho) \qquad K \mid \Gamma \mid \Sigma \vdash \storeInner_{\rho}^{in.}  \qquad
             K \mid \Gamma \mid \Sigma \vdash v : \tau
            }
            {K \mid \Gamma \mid \Sigma \vdash ((\storeInner_{\rho}^{in.}, l_\rho \mapsto v), s_a)}
    \end{array}
\end{align*}
\caption{Typing rules for $\storeInner_{\rho}^{in.}$}
\end{figure}
}}

\section{Proofs}
\label{appendix:proofs}
\begin{lem}[Preservation of types under substitution]
    \label{lemma:subst}
If 
\begin{align*}
    \begin{array}{l}
    K \mid \Gamma, x : \mu_1 \mid \Sigma \vdash e : \mu_2 \mid \varphi 
    \\
    K \mid \Gamma \mid \Sigma \vdash v : \mu_1 \mid \varphi
    \end{array}
\end{align*}
then
\begin{align*}
    K \mid \Gamma \mid \Sigma \vdash [x \mapsto v] e : \mu_2 \mid \varphi 
\end{align*}
\end{lem}
\begin{proof}
By induction on a derivation of
\begin{align*}
    K \mid \Gamma, x : \mu_1 \mid \Sigma \vdash e : \mu_2 \mid \varphi 
\end{align*}
\end{proof}

\begin{lem}[Inversion of the typing relation]
    \label{lemma:typing-inversion}
For each syntactic form:
\begin{enumerate}
    \item (\textsc{t-loc}): If 
        \begin{align*}
            K \mid \Gamma \mid \Sigma \vdash l_\rho\ : T
        \end{align*}
        , then $l_\rho : T \in \Sigma$.

    \item (\textsc{t-int}): If 
        \begin{align*}
            K \mid \Gamma \mid \Sigma \vdash n : T
        \end{align*}
        , then $T = \text{Int}$.

    \item (\textsc{t-unit}): If 
        \begin{align*}
            K \mid \Gamma \mid \Sigma \vdash ()\ : T
        \end{align*}
        , then $T = \text{Unit}$.

    \item (\textsc{t-true}): If 
        \begin{align*}
            K \mid \Gamma \mid \Sigma \vdash \text{true}\ : T
        \end{align*}
        , then $T = \text{Bool}$.

    \item (\textsc{t-false}): If 
        \begin{align*}
            K \mid \Gamma \mid \Sigma \vdash \text{false}\ : T
        \end{align*}
        , then $T = \text{Bool}$.

    \item (\textsc{t-}$\lambda$): If 
        \begin{align*}
            K \mid \Gamma  \mid \Sigma \vdash \lambda x : (\tau_1, \rho_1) . e : T_1
        \end{align*}
        , then $T_1 = (\tau_1, \rho_1) \xrightarrow{\varphi} (T_2, R_2)$ for some $T_2$, $R_2$, and $\varphi$ such that: 
        \begin{align*}
            K \mid \Gamma, x : (\tau_1, \rho_1) \mid \Sigma \vdash e : (T_2, R_2) \mid \varphi 
        \end{align*}

    \item (\textsc{t-}$\Lambda$): If 
        \begin{align*}
            K \mid \Gamma \mid \Sigma \vdash \Lambda \{ \alpha : \text{Type}, \rho : \text{Region} \} . \lambda x : (\tau_1, \rho_1) . e : T_1
        \end{align*}
        , then $T_1 = (\tau_1, \rho_1) \xrightarrow{\varphi} (T_2, R_2)$ for some $T_2$, $R_2$, and $\varphi$ such that: 
        \begin{align*}
            K \mid \Gamma, x : (\tau_1, \rho_1) \mid \Sigma \vdash e : (T_2, R_2) \mid \varphi 
        \end{align*}

    \item (\textsc{t-var}): If
        \begin{align*}
            K \mid \Gamma \mid \Sigma \vdash x : (T, R) \mid \{\bot\}
        \end{align*}
        , then $x : (T, R) \in \Gamma$.

    \item (\textsc{t-use-loc}): If 
        \begin{align*}
            K \mid \Gamma \mid \Sigma \vdash l_{\rho} : (T, R) \mid \{ \bot \} 
        \end{align*}
        , then $l_{\rho} : (T, R) \in \Sigma$.

    \item (\textsc{t-newrgn}): If 
        \begin{align*}
            K \mid \Gamma \mid \Sigma \vdash \textbf{newrgn}\ [s] : (T, R) \mid \{\textbf{fresh}\ R\ s\} 
        \end{align*}
        , then $T = \text{Unit}$ and $R$ is some fresh region $\rho$.

    \item (\textsc{t-freergn}):
        \begin{align*}
            K \mid \Gamma \mid \Sigma \vdash \textbf{freergn}\ e : (T, R) \mid \varphi \times \{\textbf{free}\ R'\}
        \end{align*}
        , then $T = \text{Unit}$, $R = \globalRegion$, and there is some type $T'$ and some region $R'$ such that: 
        \begin{align*}
            K \mid \Gamma \mid \Sigma \vdash e : (T', R') \mid \varphi
        \end{align*}

    \item (\textsc{t-split}): 
        \begin{align*}
            K \mid \Gamma \mid \Sigma \vdash \textbf{split}\ [s]\ e : (T, R) \mid \varphi \times \{\textbf{alloc}\ s\ R'\} \times \{\textbf{fresh}\ R\ s \} \times \{\textbf{alloc}\ 1\ R \}
        \end{align*}
        , then $T = \text{Unit}$, $R$ is some fresh region $\rho'$, and there is some type $T'$ and some region $R'$ such that: 
        \begin{align*}
            K \mid \Gamma \mid \Sigma \vdash e : (T', R') \mid \varphi
        \end{align*}



    \item (\textsc{t-app}): If 
        \begin{align*}
            K \mid \Gamma \mid \Sigma \vdash e_1\ e_2 : (T_2, R_2) \mid 
        \end{align*}
        , then there is some type $\mu_1$, such that 
        \begin{align*}
            K \mid \Gamma \mid \Sigma \vdash e_1 : (\mu_1 \rightarrow (T_2, R_2), \rho) \mid \varphi_1 
        \end{align*}
        and 
        \begin{align*}
            K \mid \Gamma \mid \Sigma \vdash e_2 : \mu_1 \mid \varphi_2
        \end{align*}




    \item (\textsc{t-seq}): If 
        \begin{align*}
            K \mid \Gamma \mid \Sigma \vdash e_1 ; e_2 : (T, R) \mid \varphi_1 \times \varphi_2 
        \end{align*}
        , then:
        \begin{align*}
            K \mid \Gamma \mid \Sigma \vdash e_1 : (\text{Unit}, \rho) \mid \varphi_1
        \end{align*}
        and: 
        \begin{align*}
            K \mid \Gamma \mid \Sigma \vdash e_2 : (T, R) \mid \varphi_2 
        \end{align*}

    \item (\textsc{t-if}): If 
        \begin{align*}
            K \mid \Gamma \mid \Sigma \vdash \textbf{if}\ e_1\ \textbf{then}\ e_2\ \textbf{else}\ e_3 : (T, R) \mid \varphi_1 \times (\varphi_2 \sqcup \varphi_3)
        \end{align*}
        , then .
        \begin{align*}
            K \mid \Gamma \mid \Sigma \vdash e_1 : (\text{Bool}, \rho)
        \end{align*}
        , and 
        \begin{align*}
            K \mid \Gamma \mid \Sigma \vdash e_2 : (T, R)
            K \mid \Gamma \mid \Sigma \vdash e_3 : (T, R)
        \end{align*}


        

\end{enumerate}
\end{lem}
\begin{proof}
Immediate from the definition of the typing relation.
\end{proof}

\begin{lem}[Canonical forms]
    \label{lemma:canonical}
    The canonical forms lemma determines the form of a value, given a location $l$,
    a type $(\tau, \rho)$, and a store $\storeOuter$.
\begin{enumerate}[i]
    \item If $K \mid \Gamma \mid \Sigma \vdash l\ [s]\ \textbf{at}\ e : ((\tau_1,
        \rho_1) \xrightarrow{\varphi} (\tau_2, \rho_2), \rho) \mid \varphi'$
        and $K \mid \Gamma \mid \Sigma \vdash \storeOuter$, then there is a value $v$
        at $(\storeOuter(\rho))(l)$ with type $(\tau_1, \rho_1) \xrightarrow{\varphi}
        (\tau_2, \rho_2)$ and the form $\lambda x . e'$ for some $x$ and $e'$.
    \item If $K \mid \Gamma \mid \Sigma \vdash l\ [s]\ \textbf{at}\ e : (\forall \{\alpha, \rho', \epsilon\} . (\tau_1,
        \rho_1) \xrightarrow{\varphi} (\tau_2, \rho_2), \rho) \mid \varphi'$
        and $K \mid \Gamma \mid \Sigma \vdash \storeOuter$, then there is a value $v$
        at $(\storeOuter(\rho))(l)$ with type $ \forall \{\alpha, \rho', \epsilon\} . (\tau_1, \rho_1) \xrightarrow{\varphi}
        (\tau_2, \rho_2)$ and the form $\Lambda \{\alpha, \rho', \epsilon\} . \lambda x . e'$ for some $x$ and $e'$.
    \item  If $K \mid \Gamma \mid \Sigma \vdash l\ [s]\ \textbf{at}\ e :
        (\text{Unit}, \rho) \mid \varphi $ and $K \mid \Gamma \mid \Sigma
        \vdash \storeOuter$, then there is a value $v$ at
        $(\storeOuter(\rho))(l)$ with type $\text{Unit}$ and the form $()$.
    \item  If $K \mid \Gamma \mid \Sigma \vdash l\ [s]\ \textbf{at}\ e :
        (\text{Int}, \rho) \mid \varphi \mid \Delta$ and $K \mid \Gamma \mid \Sigma
        \vdash \storeOuter$, then there is a value $v$ at
        $(\storeOuter(\rho))(l)$ with type $\text{Int}$ and the form $n$.
\end{enumerate}
\end{lem}

\subsection{Progress}

\newpage
\progress*
\begin{proof}
By induction on the structure of $e$: 
\begin{enumerate} 
    \item Case ($l_{\rho}$): Item (i) applies.
    \item Case ($x$): Item (ii) applies.
    \item Case ($v\ [s]\ \textbf{at}\ e$): By induction, there are the following cases for $e$: 
          \begin{enumerate}
            \item One of Item (ii) or Item (iii) applies to $e$.
            \item Item (i) applies to $e$. Hence, Item (i) applies with reduction (\textsc{e-val}).
          \end{enumerate}
    \item Case ($e_1\ e_2$): By induction, there are the following cases for $e_1$:
          \begin{enumerate}
            \item One of Item (ii) or Item (iii) applies to $e_1$.
            \item Item (i) applies to $e_1$. By typability, we have that 
            $K \mid \Gamma \mid \Sigma \vdash e_1\ e_2 : \mu_2 \mid \varphi_1 \times \varphi_2 \times \varphi$
            which must be due to rule (\textsc{t-app}). Thus, $K \mid \Gamma \mid \Sigma e_1 : (\mu_1 \xrightarrow{\varphi} \mu_2, \rho) \mid \varphi_1$. 
            By the canonical forms lemma (Lemma~\ref{lemma:canonical}), $e_1$ has the form of a location which 
            points to a lambda abstraction $\lambda x . e$. Again by induction, there are the following cases for $e_2$:
            \begin{enumerate}
                \item One of Item (ii) or Item (iii) applies to $e_2$.
                \item Item (i) applies to $e_2$, making $e_1\ e_2$ a beta-redex. Thus, Item (i)
                      applies with reduction (\textsc{e-appL}).
            \end{enumerate}
          \end{enumerate}
    \item Case ($\textbf{newrgn}\ [s]$): Item (i) applies with rule (\textsc{e-newrgn}).
    \item Case ($\textbf{freergn}\ e$): By induction there are the following cases for $e$:
            \begin{enumerate}
                \item One of Item (ii) or Item (iii) applies to $e$.
                \item Item (i) applies to $e$. Hence, Item (i) applies with reduction (\textsc{e-freergn}).
            \end{enumerate}
    \item Case ($\textbf{split}\ [s]\ e$): By induction there are the following cases for $e$:
            \begin{enumerate}
                \item One of Item (ii) or Item (iii) applies to $e$.
                \item Item (i) applies to $e$. Hence, Item (i) applies with reduction (\textsc{e-splitL}).
            \end{enumerate}
    \item Case ($\textbf{copy}\ e_1\ \textbf{into}\ e_2$): By induction there are the following cases for $e_1$:
            \begin{enumerate}
                \item One of Item (ii) or Item (iii) applies to $e_1$.
                \item Item (i) applies to $e_1$. By typability, we have that $K \mid \Gamma \mid \Sigma \vdash \textbf{copy}\ e_1\ \textbf{into}\ e_2 : (\tau, \rho') \mid  \varphi_1 \times \varphi_2 \times \{\textbf{alloc}\ 1\ \rho'\} $ which 
                must be due to rule (\textsc{t-copy}). Thus, $K \mid \Gamma \mid \Sigma \vdash e_1 : (\tau, \rho) \mid \varphi_1$ 
                By the canonical forms lemma (Lemma~\ref{lemma:canonical}),
                $e_1$ has the form of a location which points to a value of type $\tau$.
                Again by induction, there are the following cases for $e_2$:
                \begin{enumerate}
                    \item One of Item (ii) or Item (iii) applies to $e_2$.
                    \item Item (i) applies to $e_2$. Hence, Item (i) applies with reduction (\textsc{e-copyL}).
                \end{enumerate}
            \end{enumerate} 
    \item Case ($\textbf{ref}\ e$): By induction there are the following cases for $e$:
            \begin{enumerate}
                \item One of Item (ii) or Item (iii) applies to $e$.
                \item Item (i) applies to $e$. Hence, Item (i) applies with reduction (\textsc{e-refL}).
            \end{enumerate}
    \item Case ($\textbf{if}\ e_1\ \textbf{then}\ e_2\ \textbf{else}\ e_3$):
            By induction, there are the following cases for $e_1$:
            \begin{enumerate}
                \item One of Item (ii) or Item (iii) applies to $e_1$.
                \item Item (i) applies to $e_1$. By typability, we have that $K \mid \Gamma \mid \Sigma \vdash \textbf{if}\ e_1\ \textbf{then}\ e_2\ \textbf{else}\ e_3 : (\tau, \rho) \mid \varphi_1 \times \varphi_2 \times \varphi_3$ which must be due to 
                    rule (\textsc{t-if}). Thus, $K \mid \Gamma \mid \Sigma \vdash e_1 : (\text{Bool}, \rho') \mid \varphi_1$.
                    By the canonical forms lemma (Lemma~\ref{lemma:canonical}), $e_1$ has the form of a location which points to a value of type $\text{Bool}$.
                    Again by induction, there are the following cases for $e_2$:
                    \begin{enumerate}
                        \item One of Item (ii) or Item (iii) applies to $e_2$.
                        \item Item (i) applies to $e_2$. By typability we have that 
                              $K \mid \Gamma \mid \Sigma \vdash e_2 : (\tau, \rho) \mid \varphi_2$ and by the canonical
                              forms lemma (Lemma~\ref{lemma:canonical}), $e_2$ has the form of a location which points to a value of type $\tau$.
                              By induction, there are the following cases for $e_3$:
                              \begin{enumerate}
                                \item One of Item (ii) or Item (iii) applies to $e_3$.
                                \item Item (i) applies to $e_3$. Hence, Item (i) applies with reduction (\textsc{e-ifL}).
                              \end{enumerate}                              
                    \end{enumerate}
            \end{enumerate}
    \item Case ($!e$): By induction there are the following cases for $e$:
            \begin{enumerate}
                \item One of Item (ii) or Item (iii) applies to $e$.
                \item Item (i) applies to $e$. Hence, Item (i) applies with reduction (\textsc{e-derefL}).
            \end{enumerate}
    \item Case ($e_1 := e_2$): By induction, there are the following cases for $e_1$:
          \begin{enumerate}
            \item One of Item (ii) or Item (iii) applies to $e_1$.
            \item Item (i) applies to $e_1$. By typability, we have that $K \mid \Gamma \mid \Sigma \vdash e_1 := e_2 : (\text{Unit}, \rho) \mid \varphi_1 \times \varphi_2 $ which must be due to 
                  rule (\textsc{t-assign}). Thus, $K \mid \Gamma \mid \Sigma \vdash e_1 : (\text{Ref}\ \tau, \rho') \mid \varphi_1$.
                  By the canonical forms lemma (Lemma~\ref{lemma:canonical}), $e_1$ has the form 
                  of a location which points to a reference $\textbf{ref}\ e$.
                  Again by induction, there are the following cases for $e_2$:
                  \begin{enumerate}
                    \item One of Item (ii) or Item (iii) applies to $e_2$.
                    \item Item (i) applies to $e_2$. Hence, Item (i) applies with reduction (\textsc{e-assignL}).
                  \end{enumerate}
          \end{enumerate}
    \item Case ($e_1;\ e_2$): By induction there are the following cases for $e_1$:
          \begin{enumerate}
            \item One of Item (ii) or Item (iii) applies to $e_1$.
            \item Item (i) applies to $e_1$. By typability, we have that $K \mid \Gamma \mid \Sigma \vdash e_1 ; e_2 : \mu \mid \varphi_1 \times \varphi_2$ which must be due to 
                  rule (\textsc{t-seq}). Thus, $K \mid \Gamma \mid \Sigma \vdash e_1 : (\text{Unit}, \rho) \mid \varphi_1$.
                  By the canonical forms lemma (Lemma~\ref{lemma:canonical}), $e_1$ has the form 
                  of a location which points to a value of type $\text{Unit}$. 
                  Again by induction, there are the following cases for $e_2$:
                  \begin{enumerate}
                    \item One of Item (ii) or Item (iii) applies to $e_2$.
                    \item Item (i) applies to $e_2$. Hence, Item (i) applies with reduction (\textsc{e-seqL}).
                  \end{enumerate}
          \end{enumerate}
    \item Case ($\textbf{let}\ x : \mu\ = e_1\ \textbf{in}\ e_2$): By induction there are the following cases 
          for $e_1$:
          \begin{enumerate}
            \item One of Item (ii) or Item (iii) applies to $e_1$: 
            \item Item (i) applies to $e_1$. By typability, we have that $K \mid \Gamma \mid \Sigma \vdash \textbf{let}\ x : \mu_1 = e_1\ \textbf{in}\ e_2 : \mu_2 \mid \varphi_1 \times \varphi_2$ which must be by 
                  rule (\textsc{t-let}). Thus, $K \mid \Gamma \mid \Sigma \vdash e_1 : \mu \mid \varphi_1$.
                  By the canonical forms lemma (Lemma~\ref{lemma:canonical}), $e_1$ has the form of a location which points to a value of type $\mu$.
                  Again by induction, there are the following cases for $e_2$:
                  \begin{enumerate}
                    \item One of Item (ii) or Item (iii) applies to $e_2$.
                    \item Item (i) applies to $e_2$. Hence, Item (i) applies with reduction (\textsc{e-letL}).
                  \end{enumerate}
          \end{enumerate}
    \item Case ($e\ @\ \mu$): Item (i) applies with rule (\textsc{e-apply}).
    \item Case ($\textbf{let}\ f = \textbf{fix}(f, (\Lambda \{ \alpha , \rho ,
    \epsilon \} . \lambda x . e_1)\ [s]\ \textbf{at}\ e_2)\ \textbf{in}\ e_3$):
    By induction there are the following cases for $(\Lambda \{ \alpha , \rho ,
    \epsilon \} . \lambda x . e_1)\ [s]\ \textbf{at}\ e_2$:
      \begin{enumerate}
            \item Item (i) applies to $(\Lambda \{ \alpha , \rho , \epsilon \} .
    \lambda x . e_1)\ [s]\ \textbf{at}\ e_2$. By typability, we have that $K
    \mid \Gamma \mid \Sigma \vdash \textbf{let}\ f = \textbf{fix}(f, (\Lambda \{
    \alpha , \rho , \epsilon \} . \lambda x . e_1)\ [s]\ \textbf{at}\ e_2)\
    \textbf{in}\ e_3 : \mu_2 \mid \varphi_1 \times \varphi_2 $ which must be by
    rule (\textsc{t-fix}). Thus, $K \mid \Gamma, f : (\forall \{ \alpha, \rho,
    \epsilon  \} . (\alpha, \rho) \xrightarrow{\{ \epsilon \}} \mu_1, \rho_f)
    \mid \Sigma \vdash (\Lambda\{ \alpha, \rho, \epsilon \}. \lambda x . e_1)\
    [s]\ \textbf{at}\ e_2 : (\forall \{ \alpha , \rho, \epsilon  \} . (\alpha,
    \rho) \xrightarrow{\varphi} \mu_1, \rho_f)  \mid \varphi_1 $. By the
    canonical forms lemma (Lemma~\ref{lemma:canonical}), $(\Lambda \{ \alpha , \rho ,
    \epsilon \} . \lambda x . e_1)\ [s]\ \textbf{at}\ e_2$ has the form of a
    location which points to a value of type $(\forall \{ \alpha , \rho,
    \epsilon  \} . (\alpha, \rho) \xrightarrow{\varphi} \mu_1, \rho_f)$. Again
    by induction, there are the following cases for $e_3$:
                  \begin{enumerate}
                    \item One of Item (ii) or Item (iii) applies to $e_3$.
                    \item Item (i) applies to $e_3$. Hence, Item (i) applies with reduction (\textsc{e-fixL}).
                  \end{enumerate}
          \end{enumerate}
\end{enumerate}
\end{proof}

\subsection{Preservation}

\newpage
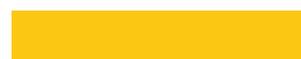
\begin{lem}[Store update]
    \label{lemma:store-update} For updating a store we have that:
If
\begin{align*}
\begin{array}{l}
    K \mid \Gamma \mid \Sigma \vdash \storeOuter \\ 
    \tau = \Sigma(l_\rho) \\
    K \mid \Gamma \mid \Sigma \vdash v : \tau
\end{array}
\end{align*}
then 
\begin{align*}
    K \mid \Gamma \mid \Sigma \vdash [ \rho \mapsto ((\storeInner, l_\rho \mapsto v), s_a) ] \storeOuter
\end{align*}
\end{lem}

\begin{lem}[Store weakening]
    \label{lemma:store-weak}
If 
\begin{align*}
    \begin{array}{l}
    K \mid \Gamma \mid \Sigma \vdash e : \mu \mid \varphi \\
    \Sigma' \supseteq \Sigma\\ 
    \end{array}
\end{align*}
then 
\begin{align*}
    \begin{array}{l}
    K \mid \Gamma \mid \Sigma' \vdash e : \mu \mid \varphi \\ 
    K \mid \Gamma \mid \Sigma' \vdash \storeOuter
    \end{array}
\end{align*}
\end{lem}
\begin{proof}
Straightforward induction.
\end{proof}

\begin{lem}[Fresh Region Consistency]
    \label{lemma:fresh-region-consist}
If 
\begin{align*}
    \begin{array}{l}
    K \mid \Gamma \mid \Sigma \vdash \textbf{newrgn}\ [s] : (\text{Unit}, \rho) \mid \{\textbf{fresh}\ (\rho, s)\} \\
    K \mid \Gamma \mid \Sigma \vdash \storeOuter
    \end{array}
\end{align*}
then 
\begin{align*}
    \begin{array}{l}
        \rho' \equiv \rho
    \end{array}
\end{align*}
\end{lem}
\begin{proof}
    By induction on the definition of $\textbf{fresh}\ (\rho, s)$ and $\text{freshRegion}()$ or something.
\end{proof}


\newpage

\preservation*

\begin{proof}
By induction on a derivation of $K \mid \Gamma \mid \Sigma \vdash e : \mu \mid \varphi$: 
\begin{itemize}

    \item Case (\textsc{t-var}):\\
        From (\textsc{t-var}) we have the following assumptions:
        \begin{enumerate}[i]
            \itemsep0em 
            \item\label{ite:t-var1} $e = x$
            \item\label{ite:t-var2} $x : \mu \in \Gamma$
        \end{enumerate}
        Not possible (no evaluation rules with a variable as the left-hand side).
 
    \item Case (\textsc{t-use-val}):\\
        From (\textsc{t-use-val}) we have the following assumptions:
        \begin{enumerate}[i]
            \itemsep0em 
                \item\label{ite:t-use-val1} $e = l_\rho$
                \item\label{ite:t-use-val2} $\tau = \Sigma(l_\rho)$
        \end{enumerate}
        Not possible (no evaluation rules with a location as the left-hand side).

    \item Case (\textsc{t-newrgn}):\\
        From (\textsc{t-newrgn}) we have the following assumptions:
        \begin{enumerate}[i]
            \itemsep0em 
            \item\label{ite:t-newrgn1} $e = \textbf{newrgn}\ [s]$
            \item\label{ite:t-newrgn2} $s \sqsupseteq 1$
            \item\label{ite:t-newrgn3} $K \vdash s : \text{Size}$
            \item\label{ite:t-newrgn4} $\varphi = \{ \textbf{fresh}\ \rho\ s\} \times \{\textbf{alloc}\ 1\ \rho\}$
        \end{enumerate}
        There is one evaluation rule by which $\langle e \mid \storeOuter
        \rangle \longrightarrow \langle e' \mid \storeOuter' \rangle$ can be derived:
        (\textsc{e-newrgn}).
        \begin{itemize}
            \item Subcase (\textsc{e-newrgn}): \\
                From (\textsc{e-newrgn}) we have the following assumptions:
                \begin{enumerate}[a.]
                    \itemsep0em 
                        \item\label{ite:e-newrgn1} $\rho' = \text{freshRegion}()$
                        \item\label{ite:e-newrgn2} $l_\rho' = \text{freshLoc}(\rho)$
                        \item\label{ite:e-newrgn3} $e' = l_{\rho'}$  
                        \item\label{ite:e-newrgn4} $\storeOuter' = \storeOuter, \rho' \mapsto (l_{\rho'} \mapsto (), s)$
                \end{enumerate}

                From assumptions~\ref{ite:t-newrgn4}, and~\ref{ite:e-newrgn4} we have that $\rho$ and $\rho'$ are fresh region names, 
                and from fresh region name consistency (Lemma~\ref{lemma:fresh-region-consist}), we know
                that $\rho' \equiv \rho$. Thus we can freely substitute $\rho$ for $\rho'$ throughout the proof case.

                Since we make no inductive step in this case  
                and from assumptions~\ref{ite:e-newrgn3},~\ref{ite:e-newrgn4}, along with the derivation:
                \begin{equation}
                    \inferrule*[right=t-unit]{\quad}
                    {K \mid \Gamma \mid \Sigma \vdash () : \text{Unit}}
                \end{equation}
                we let $\Sigma'$ equal $\Sigma$ extended with a single location $l_\rho$,
                such that $\text{Unit} = \Sigma'(l_\rho)$ \hypertarget{newrgn1}{(*)}. 
                From this we can construct the following typing derivation for $e'$: 
                \begin{align*}
                    \inferrule*[right=t-use-val]
                        { \inferrule*[right=t-loc]
                            {  \texorpdfstring{\protect\hyperlink{newrgn1}{\text{(*)}}}{}\qquad \inferrule*[right=$\kappa$-unit]{\quad}{K \vdash \text{Unit} : \text{Type}}}
                            {K \mid \Gamma \mid \Sigma \vdash l_\rho : \text{Unit} }}
                        { K \mid \Gamma \mid \Sigma' \vdash l_\rho : (\text{Unit}, \rho) \mid \{\bot\} }
                \end{align*}
                From which we have that $\varphi' = \{\bot \}$ and from assumption~\ref{ite:t-newrgn4} that 
                $\varphi =  \{ \textbf{fresh}\ (\rho, s) \}$. From this, we can construct the following derivation for 
                $\{ \bot \} \sqsubseteq \{ \textbf{fresh}\ (\rho, s) \} $:

                \begin{equation}
                    \inferrule*[right=$\kappa$-alloc]
                                {
                                    \inferrule*[right=$\kappa$-reg]
                                        {\quad}
                                        {K \vdash \rho : \text{Region}}
                                    \qquad
                                    \inferrule*[right=$\kappa$-size]
                                        {\quad}
                                        {K \vdash 1 : \text{Size}}
                                }
                                {K \vdash \{\textbf{alloc}\ 1\ \rho\} : \text{Effect}}
                    \label{eq:newrgn-phi-1}
                \end{equation}

                \begin{align*} 
                    \inferrule*[right=sb-$\bot$]  
                        {\inferrule*[right=$\kappa$-times]
                            {
                            \inferrule*[right=$\kappa$-fresh]
                                {\inferrule*[right=$\kappa$-size]
                                    {\quad}
                                    {K \vdash s : \text{Size}}}
                                {K \vdash \{\textbf{fresh}\ \rho\ s \} : \text{Effect}}
                            \qquad 
                            \eqref{eq:newrgn-phi-1} 
                            }
                            {K \vdash \{ \textbf{fresh}\ \rho\ s \} \times \{\textbf{alloc}\ 1\ \rho\} : \text{Effect} } }
                        {K \vdash \{ \bot \} \sqsubseteq \{ \textbf{fresh}\ \rho\ s \} \times \{\textbf{alloc}\ 1\ \rho \} : \text{Effect} }                     
                \end{align*}

                From our inductive hypothesis we have that:
                \begin{align*}
                    K \mid \Gamma \mid \Sigma \vdash \storeOuter
                \end{align*}
                i.e., that the store prior to evaluation is well-typed~(\hypertarget{newrgn2}{\textdagger}). 
                As per assumption~\ref{ite:e-newrgn4}, (\textsc{e-newrgn}) extends the store with a new region $\rho$ of 
                size $s$, where from assumption~\ref{ite:t-newrgn2} we have that $s \sqsupseteq 1$~(\hypertarget{newrgn3}{\textdagger\textdagger}).
                This region contains a pointer to value $()$, and from Definition~\ref{def:sizeOf}, we know that 
                this value has the size $1$. 
                
                Thus, as $\rho$ contains only one element, we have from 
                Definition~\ref{def:region-alloc} that currentSize($\rho$) = $1$, which we can use with $\texorpdfstring{\protect\hyperlink{newrgn3}{\text{(\textdagger\textdagger)}}}{}$ to satisfy the store 
                typing constraint $1 \sqsubseteq s$ (\hypertarget{newrgn4}{\textdagger\textdagger\textdagger}). 
            
                From the above we can construct the following typing derivation for the output store $\storeOuter'$: 

                \begin{equation}
                    \inferrule*[right=loc]
                        {
                            \texorpdfstring{\protect\hyperlink{newrgn1}{\text{(*)}}}{}
                            \qquad 
                            \inferrule*[right=emptyInner]
                                {\quad}
                                {K \mid \Gamma \mid \Sigma \vdash (\emptyset, s) }
                            \qquad 
                            \inferrule*[right=t-unit]
                                {\quad}
                                {K \mid \Gamma \mid \Sigma \vdash () : \text{Unit} }
                            }
                        {K \mid \Gamma \mid \Sigma \vdash (l_\rho \mapsto (), s) } 
                    \label{eq:newrgn-store-1}
                \end{equation}

                \begin{align*}
                    \inferrule*[right=inner]
                        {
                        \texorpdfstring{\protect\hyperlink{newrgn2}{\text{(\textdagger)}}}{}
                        \qquad 
                        \eqref{eq:newrgn-store-1}
                        \qquad 
                        \texorpdfstring{\protect\hyperlink{newrgn4}{\text{(\textdagger\textdagger\textdagger)}}}{}
                        }
                        {K \mid \Gamma \mid \Sigma \vdash \storeOuter, \rho \mapsto (l_\rho \mapsto (), s) }             
                \end{align*}
                
        \end{itemize}

    \item Case (\textsc{t-freergn}):\\
        From (\textsc{t-freergn}) we have the following assumptions:
        \begin{enumerate}[i]
            \itemsep0em 
                \item\label{ite:t-freergn1} $e = \textbf{freergn}\ e_1$
                \item\label{ite:t-freergn2} $K \mid \Gamma \mid \Sigma \vdash e_1 : (\tau, \rho) \mid \varphi_1$
                \item\label{ite:t-freergn3} $\varphi = \varphi_1 \times \{\textbf{free}\ \rho\}$
        \end{enumerate}
        There are two evaluation rules by which $\langle e \mid \storeOuter
        \rangle \longrightarrow \langle e' \mid \storeOuter' \rangle$ can be derived:
        (\textsc{e-freergn}) and (\textsc{e-freergnL}).
        \begin{itemize}

            \item Subcase (\textsc{e-freergn}): \\
                From (\textsc{e-freergn}) we have the following assumptions:
                \begin{enumerate}[a.]
                    \itemsep0em 
                        \item\label{ite:e-freergn1} $\eformat{e_1}{\storeOuter}{e'_1}{\storeOuter'}$
                        \item\label{ite:e-freergn2} $e' = \textbf{freergn}\ e'_1$
                \end{enumerate}

                By induction on assumption~\ref{ite:t-freergn2} we have: 
                \begin{equation}
                    K \mid \Gamma \mid \Sigma' \vdash e'_1 : (\tau, \rho) \mid \varphi_1' 
                    \label{eq:freergn1}
                \end{equation}
                Thus, let $\Sigma'$ equal any store typing such that $\Sigma' \supseteq \Sigma$ and the above typing holds.
                From this, we can construct the following typing derivation for $e'$:
                \begin{align*}
                    \inferrule*[right=t-freergn]
                        {
                           \eqref{eq:freergn1} }
                        {K \mid \Gamma \mid \Sigma' \vdash \textbf{freergn}\ e'_1 : (\text{Unit}, \globalRegion) \mid  \varphi_1' \times \{\textbf{free}\ \rho\}   }
                \end{align*}
                From which, we have that $\varphi' = \varphi'_1 \times \{ \textbf{free}\ s\}$ and from assumption~\ref{ite:t-freergn3} that
                $\varphi = \varphi_1 \times \{\textbf{free}\ \rho\} $. 
                 
                By the typability of assumption~\ref{ite:t-freergn2} we have that 
                $K \vdash \varphi_1 : \text{Effect}$ \hypertarget{freergn1}{(*)} and by induction that $\varphi'_1 \sqsubseteq \varphi_1$ \hypertarget{freergn2}{(**)}.
                From this, we can construct the following derivation for $\varphi'_1 \times \{ \textbf{free}\ s\} \sqsubseteq \varphi_1 \times \{ \textbf{free}\ s\}$: 
                \begin{equation}
                    \inferrule*[right=sb-$\times_{below}^{2}$]
                            {
                                \texorpdfstring{\protect\hyperlink{freergn1}{\text{(*)}}}{}  
                                \qquad 
                                \inferrule*[right=sb-$\equiv$]
                                {
                                    \inferrule*[right=eq-refl]
                                    {\inferrule*[right=$\kappa$-free]
                                    {\inferrule*[right=$\kappa$-reg]
                                        {\quad}
                                        {K \vdash \rho : \text{Region}}}
                                    {K \vdash \{\textbf{free}\ \rho\} : \text{Effect} }}{K \vdash \{\textbf{free}\ \rho\} \equiv \{\textbf{free}\ \rho\} : \text{Effect}} }
                                {K \vdash \{\textbf{free}\ \rho\} \sqsubseteq \{\textbf{free}\ \rho\} : \text{Effect}} 
                                } 
                            {K \vdash \{ \textbf{free}\ \rho\} \sqsubseteq  \varphi_1 \times \{ \textbf{free}\ \rho\} : \text{Effect}} 
                    \label{eq:e-freergn1}
                \end{equation}

                \begin{equation} 
                    \inferrule*[right=sb-$\times_{below}^{1}$]
                         { \texorpdfstring{\protect\hyperlink{freergn2}{\text{(**)}}}{}
                        \qquad \inferrule*[right=$\kappa$-free]
                            {\inferrule*[right=$\kappa$-reg]
                                {\quad}
                                {K \vdash \rho : \text{Region} }} 
                            {K \vdash \{\textbf{free}\ \rho\} : \text{Effect}} }
                        { K \vdash \varphi'_1 \sqsubseteq \varphi_1  \times \{\textbf{free}\ s\}  : \text{Effect}}
                    \label{eq:e-freergn2}
                \end{equation}

                \begin{align*} 
                    \inferrule*[right=sb-$\times_{above}$]
                        {   \eqref{eq:e-freergn1} 
                            \qquad 
                            \eqref{eq:e-freergn2}
                            }
                        {K \vdash \varphi'_1 \times \{ \textbf{free}\ \rho\}   \sqsubseteq  \varphi_1 \times \{ \textbf{free}\ \rho\} : \text{Effect} }
                \end{align*}
        
                Finally, we can type the output store $\storeOuter'$ by induction:
                \begin{align*}
                    K \mid \Gamma \mid \Sigma' \vdash \storeOuter' 
                \end{align*}

            \item Subcase (\textsc{e-freergnL}): \\
                From (\textsc{e-freergnL}) we have the following assumptions:
                \begin{enumerate}[a.]
                    \itemsep0em 
                        \item\label{ite:e-freergnL1} $e' = l_{\globalRegion}^1$
                        \item\label{ite:e-freergnL2} $\storeOuter' = \storeOuter \setminus \rho$
                \end{enumerate}
                Since we make no inductive step in this case, we can assume that 
                $\Sigma' = \Sigma$. 

                From the premises of our theorem, we have that $\text{Unit} = \Sigma'(l^{1}_{\globalRegion})$ \hypertarget{freergnL1}{(*)}. From this we can construct the following typing derivation for $e'$:

                \begin{align*}
                    \inferrule*[right=t-use-val]{
                        \inferrule*[right=t-loc]
                            {\texorpdfstring{\protect\hyperlink{freergnL1}{\text{(*)}}}{} \qquad \inferrule*[right=$\kappa$-unit]{\quad}{K \vdash \text{Unit} : \text{Type}}}
                            {K \mid \Gamma \mid \Sigma' \vdash l_{\globalRegion}^1 : \text{Unit} } }
                        {K \mid \Gamma \mid  \Sigma' \vdash l_{\globalRegion}^1 : (\text{Unit},\globalRegion) \mid \{\bot \} } 
                \end{align*}
                From which we have that $\varphi' = \{\bot \}$ and from assumption~\ref{ite:t-freergn3} that 
                $\varphi = \varphi_1 \times \{ \textbf{free}\ \rho\}$. From the typability of assumption~\ref{ite:t-freergn2} we
                have that $K \vdash \varphi_1 : \text{Effect}$ \hypertarget{freergnL2}{(**)}, from which we can construct the 
                following derivation for $ \{\bot \} \sqsubseteq \varphi_1 \times \{ \textbf{free}\ \rho\}$:

                \begin{align*}
                    \inferrule*[right=sb-$\bot$]
                        {\inferrule*[right=$\kappa$-$\times$]
                        {
                               \texorpdfstring{\protect\hyperlink{freergnL2}{\text{(**)}}}{}     
                        \qquad
                            \inferrule*[right=$\kappa$-free]
                            {\inferrule*[right=$\kappa$-reg]
                                {\quad}
                                {K \vdash \rho : \text{Region}}}
                            {K \vdash \{\textbf{free}\ \rho\} : \text{Effect}}
                        }
                        {K \vdash \varphi_1 \times \{\textbf{free}\ \rho\} : \text{Effect}}
                        }
                        {K \vdash \{\bot\} \sqsubseteq \varphi_1 \times \{\textbf{free}\ \rho\} : \text{Effect} }    
                \end{align*}
                Using store weakening (Lemma~\ref{lemma:store-weak}) as we no longer require 
                the ability to type $\rho$ despite its presence in $\Sigma'$, we can type the output 
                store $\storeOuter'$: 
                \begin{align*}
                    K \mid \Gamma \mid \Sigma' \vdash \storeOuter \setminus \rho
                \end{align*}
        \end{itemize}

    \item Case (\textsc{t-split}):\\
        From (\textsc{t-split}) we have the following assumptions:
        \begin{enumerate}[i]
            \itemsep0em 
            \item\label{ite:t-split1} $e = \textbf{split}\ [s]\ e'$
            \item\label{ite:t-split2} $K \mid \Gamma \mid \Sigma \vdash e_1 : (\tau, \rho) \mid \varphi_1 $
            \item\label{ite:t-split3} $s \sqsupseteq 1$ 
            \item\label{ite:t-split4} $K \vdash s : \text{Size}$
            \item\label{ite:t-split5} $\varphi = \varphi_1 \times \{ \textbf{split}\ \rho\ s\ \rho' \} \times \{ \textbf{alloc}\ 1\ \rho' \}  $
        \end{enumerate}
        There are two evaluation rules by which $\langle e \mid \storeOuter
        \rangle \longrightarrow \langle e' \mid \storeOuter' \rangle$ can be derived:
        (\textsc{e-split}) and (\textsc{e-splitL}).
        \begin{itemize}

            \item Subcase (\textsc{e-split}): \\
                From (\textsc{e-split}) we have the following assumptions:
                \begin{enumerate}[a.]
                    \itemsep0em 
                    \item\label{ite:e-split1} $\eformat{e_1}{\storeOuter}{e'_1}{\storeOuter'}$
                    \item\label{ite:e-split2} $e' = \textbf{split}\ [s]\ e'_1$
                \end{enumerate}
                By induction on assumption~\ref{ite:t-split2} we have:
                \begin{equation}
                    K \mid \Gamma \mid \Sigma' \vdash e'_1 : (\tau, \rho) \mid \varphi'_1
                    \label{eq:split1}
                \end{equation}
                Let $\Sigma'$ equal any store typing such that $\Sigma' \supseteq
                \Sigma$. 
                From this, along with assumptions~\ref{ite:t-split3} and~\ref{ite:t-split4}, we can construct the following typing derivation for $e'$: 
                \begin{align*}
                    \hspace{-3em}
                    \inferrule*[right=t-split]
                        {\eqref{eq:split1} \qquad \ref{ite:t-split3} \qquad \ref{ite:t-split4}}
                        {K \mid \Gamma \mid \Sigma' \vdash \textbf{split}\ [s]\ e'_1 : (\text{Unit}, \rho) \mid \varphi'_1  \times \{ \textbf{split}\ \rho\ s\ \rho' \} \times \{ \textbf{alloc}\ 1\ \rho'\}  }
                \end{align*}

                From which we have that $\varphi' = \varphi'_1  \times \{ \textbf{split}\ \rho\ s\ \rho' \} \times \{\textbf{alloc}\ 1\ \rho' \} $ 
                and from assumption~\ref{ite:t-split5} that 
                $\varphi =  \varphi_1 \times \{ \textbf{split}\ \rho\ s\ \rho' \} \times \{\textbf{alloc}\ 1\ \rho'\} $. 
                From the typability of assumption~\ref{ite:t-split2}, we have that $K \vdash \varphi_1 : \text{Effect}$ \hypertarget{split1}{(*)},             
                and by induction that $\varphi'_1 \sqsubseteq \varphi_1$ \hypertarget{split2}{(**)}. From this, we can construct the following derivation for 
                $ \varphi_1 \times \{ \textbf{split}\ \rho\ s\ \rho' \} \times \{\textbf{alloc}\ 1\ \rho'\}
                \sqsubseteq \varphi'_1 \times \{ \textbf{split}\ \rho\ s\ \rho' \} \times \{\textbf{alloc}\ 1\ \rho'\} $:

                \begin{equation}
                    \inferrule*[right=$\kappa$-split]
                                                        {
                                                            \inferrule*[right=$\kappa$-reg]
                                                                {\quad}
                                                                {K \mid \Gamma \mid \Sigma \vdash \rho : \text{Region}}
                                                            \qquad 
                                                            \inferrule*[right=$\kappa$-size]
                                                                {\quad}
                                                                {K \mid \Gamma \mid \Sigma \vdash s : \text{Size}}
                                                        } 
                                                        {K \mid \Gamma \mid \Sigma' \vdash \{\textbf{split}\ \rho\ s\ \rho' \} : \text{Effect} }
                    \label{eq:split15}
                \end{equation}

                \begin{equation}
                            \inferrule*[right=sb-$\times_{below}^{1}$]
                                                {
                                                    \texorpdfstring{\protect\hyperlink{split2}{\text{(**)}}}{}
                                                    \qquad
                                                    \eqref{eq:split15} 
                                                }
                                                {K \mid \Gamma \mid \Sigma \vdash \varphi'_1 \sqsubseteq \varphi_1 \times \{\textbf{split}\ \rho\ s\ \rho'\} : \text{Effect} }
                    \label{eq:split13}
                \end{equation}

                \begin{equation}
                        \inferrule*[right=$\kappa$-alloc]
                                                {
                                                    \inferrule*[right=$\kappa$-reg]
                                                        {\quad}
                                                        {K \mid \Gamma \mid \Sigma \vdash \rho' : \text{Region}}
                                                    \qquad 
                                                    \inferrule*[right=$\kappa$-size]
                                                        {\quad}
                                                        {K \mid \Gamma \mid \Sigma \vdash 1 : \text{Size}}  
                                                }
                                                {K \mid \Gamma \mid \Sigma \vdash \{\textbf{alloc}\ 1\ \rho'\} : \text{Effect} }
                    \label{eq:split14}
                \end{equation}

                \begin{equation}
                        \inferrule*[right=sb-$\times_{below}^{1}$]
                                        {
                                            \eqref{eq:split13} 
                                            \qquad 
                                            \eqref{eq:split14}    
                                        }
                                        {K \mid \Gamma \mid \Sigma \vdash \varphi'_1 \sqsubseteq \varphi_1 \times \{\textbf{split}\ \rho\ s\ \rho' \} \times \{\textbf{alloc}\ 1\ \rho'\} : \text{Effect}  }
                    \label{eq:split8}
                \end{equation}

                \begin{equation}
                            \inferrule*[right=sb-$\equiv$]
                                                        {
                                                            \inferrule*[right=eq-refl]
                                                                {
                                                                    \inferrule*[right=$\kappa$-split]
                                                                        {
                                                                            \inferrule*[right=$\kappa$-reg]
                                                                                {\quad}
                                                                                {K \vdash \rho : \text{Region}}
                                                                            \qquad 
                                                                            \inferrule*[right=$\kappa$-size]
                                                                                {\quad}
                                                                                {K \vdash s : \text{Size}}
                                                                        }
                                                                        {K \mid \Gamma \mid \Sigma \vdash \{\textbf{split}\ \rho\ s\ \rho' \} : \text{Effect} }
                                                                }
                                                                {K \mid \Gamma \mid \Sigma \vdash \{\textbf{split}\ \rho\ s\ \rho' \} \equiv \{\textbf{split}\ \rho\ s\ \rho' \} : \text{Effect} }
                                                        }
                                                        {K \mid \Gamma \mid \Sigma \vdash \{\textbf{split}\ \rho\ s\ \rho' \} \sqsubseteq \{\textbf{split}\ \rho\ s\ \rho' \} : \text{Effect} }
                    \label{eq:split12}
                \end{equation}

                \begin{equation}
                        \inferrule*[right=sb-$\times_{below}^{2}$]
                                                {
                                                    \texorpdfstring{\protect\hyperlink{split1}{\text{(*)}}}{}
                                                    \qquad 
                                                    \eqref{eq:split11} 
                                                }
                                                {K \mid \Gamma \mid \Sigma \vdash  \{\textbf{split}\ \rho\ s\ \rho' \} \sqsubseteq \varphi_1 \times \{\textbf{split}\ \rho\ s\ \rho' \} : \text{Effect} }
                    \label{eq:split10}
                \end{equation}

                \begin{equation}
                            \inferrule*[right=$\kappa$-alloc]
                                                {
                                                    \inferrule*[right=$\kappa$-reg]
                                                        {\quad}
                                                        {K \mid \Gamma \mid \Sigma \vdash \rho' : \text{Region}}
                                                    \qquad 
                                                    \inferrule*[right=$\kappa$-size]
                                                        {\quad}
                                                        {K \mid \Gamma \mid \Sigma \vdash 1 : \text{Size}}
                                                }
                                                {K \mid \Gamma \mid \Sigma \vdash \{\textbf{alloc}\ 1\ \rho'\} : \text{Effect}  }
                    \label{eq:split11}
                \end{equation}

                \begin{equation}
                        \inferrule*[right=sb-$\times_{below}^{1}$]
                                        {
                                            \eqref{eq:split10} 
                                            \qquad 
                                            \eqref{eq:split11}                                         
                                        }
                                        {K \mid \Gamma \mid \Sigma \vdash \{\textbf{split}\ \rho s\ \rho' \}  \sqsubseteq \varphi_1 \times \{\textbf{split}\ \rho\ s\ \rho' \} \times \{\textbf{alloc}\ 1\ \rho'\} : \text{Effect} }
                    \label{eq:split9}
                \end{equation}

                \begin{equation}
                    \inferrule*[right=sb-$\times_{above}$]
                                {
                                    \eqref{eq:split8} 
                                    \qquad 
                                    \eqref{eq:split9}                                    
                                }
                                {K \mid \Gamma \mid \Sigma \vdash \varphi'_1 \times \{\textbf{split}\ \rho s\ \rho' \}  \sqsubseteq \varphi_1 \times \{\textbf{split}\ \rho\ s\ \rho' \} \times \{\textbf{alloc}\ 1\ \rho'\} : \text{Effect} }
                    \label{eq:split2}
                \end{equation}

                \begin{equation}
                            \inferrule*[right=$\kappa$-split]
                                                {
                                                    \inferrule*[right=$\kappa$-reg]
                                                        {\quad}
                                                        {K \mid \Gamma \mid \Sigma \vdash \rho : \text{Region}}
                                                    \qquad 
                                                    \inferrule*[right=$\kappa$-size]
                                                        {\quad}
                                                        {K \mid \Gamma \mid \Sigma \vdash s : \text{Size}}
                                                }
                                                {K \mid \Gamma \mid \Sigma \vdash \{\textbf{split}\ \rho\ s\ \rho' \} : \text{Effect} }
                    \label{eq:split5}
                \end{equation}

                \begin{equation}
                                \inferrule*[right=eq-refl]
                                                        {
                                                            \inferrule*[right=$\kappa$-alloc]
                                                                {
                                                                    \inferrule*[right=$\kappa$-reg]
                                                                        {\quad}
                                                                        {K \mid \Gamma \mid \Sigma \vdash \rho' : \text{Region}}
                                                                    \qquad 
                                                                    \inferrule*[right=$\kappa$-size]
                                                                        {\quad}
                                                                        {K \mid \Gamma \mid \Sigma \vdash 1 : \text{Size}}
                                                                }
                                                                {K \mid \Gamma \mid \Sigma \vdash \{\textbf{alloc}\ 1\ \rho' \} : \text{Effect} }
                                                        }
                                                        {K \mid \Gamma \mid \Sigma \vdash \{\textbf{alloc}\ 1\ \rho'\} \equiv \{\textbf{alloc}\ 1\ \rho'\} : \text{Effect} }
                    \label{eq:split7}
                \end{equation}

                \begin{equation}
                            \inferrule*[right=sb-$\equiv$]
                                                {
                                                    \eqref{eq:split7} 
                                                }
                                                {K \mid \Gamma \mid \Sigma \vdash \{\textbf{alloc}\ 1\ \rho'\} \sqsubseteq \{\textbf{alloc}\ 1\ \rho'\} : \text{Effect} }
                    \label{eq:split6}
                \end{equation}

                \begin{equation}
                    \inferrule*[right=sb-$\times_{below}^{2}$]
                                        {
                                            \eqref{eq:split5}
                                            \qquad 
                                            \eqref{eq:split6} 
                                        }
                                        {K \mid \Gamma \mid \Sigma \vdash \{\textbf{alloc}\ 1\ \rho'\} \sqsubseteq \{\textbf{split}\ \rho\ s\ \rho' \} \times \{\textbf{alloc}\ 1\ \rho'\} : \text{Effect} }
                    \label{eq:split4}
                \end{equation}

                \begin{equation}
                        \inferrule*[right=sb-$\times_{below}^{2}$]
                                {
                                    \texorpdfstring{\protect\hyperlink{split1}{\text{(*)}}}{}
                                    \qquad 
                                    \eqref{eq:split4} 
                                }
                                {K \mid \Gamma \mid \Sigma \vdash \{\textbf{alloc}\ 1\ \rho'\} \sqsubseteq \varphi_1 \times \{\textbf{split}\ \rho\ s\ \rho' \} \times \{\textbf{alloc}\ 1\ \rho'\} : \text{Effect}  }
                    \label{eq:split3}
                \end{equation}

                \begin{align*}
                    \hspace{-6em}
                    \inferrule*[right=sb-$\times_{above}$]
                        {
                            \eqref{eq:split2} 
                            \qquad 
                            \eqref{eq:split3} 
                        }
                        {K \mid \Gamma \mid \Sigma \vdash \varphi'_1 \times \{\textbf{split}\ \rho s\ \rho' \} \times \{\textbf{alloc}\ 1\ \rho'\} \sqsubseteq \varphi_1 \times \{\textbf{split}\ \rho\ s\ \rho' \} \times \{\textbf{alloc}\ 1\ \rho'\} : \text{Effect} }
                \end{align*}

                Finally, we can type the output store $\storeOuter'$ by induction:
                \begin{align*}
                    K \mid \Gamma \mid \Sigma' \vdash \storeOuter'
                \end{align*}

            \item Subcase (\textsc{e-splitL}): \\
                From (\textsc{e-splitL}) we have the following assumptions:
                \begin{enumerate}[a.]
                    \itemsep0em 
                        \item\label{ite:e-splitL1} $\rho'' = \text{freshRegion()}$
                        \item\label{ite:e-splitL2} $(\storeInner_{\rho}^{in.}) = \storeOuter(\rho)$
                        \item\label{ite:e-splitL3} $l_{\rho''} = \text{freshLoc}(\rho'')$
                        \item\label{ite:e-splitL4} $e' = l_{\rho''}$
                        \item\label{ite:e-splitL5} $\storeOuter' = [\rho \mapsto (\storeInner_{\rho}^{in.}, s_a - s), \rho'' \mapsto (l_{\rho''} \mapsto (), s)]\storeOuter$
                \end{enumerate}

                From assumptions~\ref{ite:t-split5}, and~\ref{ite:e-splitL1} we have that $\rho'$ and $\rho''$ are 
                fresh region names, and from fresh region name consistency (Lemma~\ref{lemma:fresh-region-consist}), we know
                that $\rho'' \equiv \rho'$. Thus we can freely substitute $\rho'$ for $\rho''$ throughout the proof case.

                Since we make no inductive step in this case
                and from assumptions~\ref{ite:e-splitL3},~\ref{ite:e-splitL5}, along with the derivation:
                \begin{equation}
                    \inferrule*[right=t-unit]{\quad}
                    {K \mid \Gamma \mid \Sigma \vdash () : \text{Unit}}
                \end{equation}
                we let $\Sigma'$ equal $\Sigma$ extended with a single location $l_\rho$,
                such that $\tau = \Sigma'(l_\rho)$ \hypertarget{splitL1}{(*)}. 
                From this we can construct the following typing derivation for $e'$: 
                \begin{align*}
                    \inferrule*[right=t-use-val]
                        { \inferrule*[right=t-loc]
                            {  \texorpdfstring{\protect\hyperlink{splitL1}{\text{(*)}}}{}\qquad \inferrule*[right=$\kappa$-unit]{\quad}{K \vdash \text{Unit} : \text{Type}}}
                            {K \mid \Gamma \mid \Sigma \vdash l_{\rho'} : \text{Unit} }}
                        { K \mid \Gamma \mid \Sigma' \vdash l_{\rho'} : (\text{Unit}, \rho') \mid \{\bot\} }
                \end{align*}
                From which we have that $\varphi' = \{\bot \}$ and from assumption~\ref{ite:t-split5} that 
                $\varphi =  \varphi_1 \times \{\textbf{split}\ \rho\ s\ \rho' \} \times \{\textbf{alloc}\ 1\ \rho'\}$. 
                From the typability of assumption~\ref{ite:t-split2}, we have that $K \vdash \varphi_1 : \text{Effect}$ \hypertarget{splitL2}{(**)},             
                from which we can construct the following derivation for 
                $\{ \bot \} \sqsubseteq \varphi_1 \times \{ \textbf{split}\ \rho\ s\ \rho' \} \times \{\textbf{alloc}\ 1\ \rho'\}$:

                \begin{equation}
                    \inferrule*[right=$\kappa$-$\times$]
                        {
                            \texorpdfstring{\protect\hyperlink{splitL2}{\text{(**)}}}{}
                            \qquad 
                            \inferrule*[right=$\kappa$-split]
                                {
                                    \inferrule*[right=$\kappa$-reg]
                                        {\quad}
                                        {K \mid \Gamma \mid \Sigma \vdash \rho : \text{Region}}
                                    \qquad 
                                    \inferrule*[right=$\kappa$-size]
                                        {\quad}
                                        {K \mid \Gamma \mid \Sigma \vdash s : \text{Size}}
                                }
                                {K \mid \Gamma \mid \Sigma \vdash \{\textbf{split}\ \rho\ s\ \rho'\} : \text{Effect}}
                        }
                        {K \mid \Gamma \mid \Sigma \vdash \varphi_1 \times \{ \textbf{split}\ \rho\ s\ \rho' \} : \text{Effect}}
                    \label{eq:splitL2}
                \end{equation}

                \begin{align*}
                    \inferrule*[right=sb-$\bot$]
                        {\inferrule*[right=$\kappa$-$\times$]
                            {\eqref{eq:splitL2}
                            \qquad
                            \inferrule*[right=$\kappa$-alloc]
                                {
                                    \inferrule*[right=$\kappa$-reg]
                                        {\quad}
                                        {K \vdash \rho' : \text{Region}}
                                    \qquad
                                    \inferrule*[right=$\kappa$-size]
                                        {\quad}
                                        {K \vdash 1 : \text{Size}}
                                }
                                {K \vdash \{\textbf{alloc}\ 1\ \rho'\} : \text{Effect}}                            
                            }
                            {K \vdash \varphi_1 \times \{ \textbf{split}\ \rho\ s\ \rho' \} \times \{\textbf{alloc}\ 1\ \rho'\} : \text{Effect} }
                            }
                        {K \vdash \{ \bot \} \sqsubseteq \varphi_1 \times \{ \textbf{split}\ \rho\ s\ \rho' \} \times \{\textbf{alloc}\ 1\ \rho'\} : \text{Effect} }
                \end{align*}

        \end{itemize}

    \item Case (\textsc{t-copy}):\\
        From (\textsc{t-copy}) we have the following assumptions:
        \begin{enumerate}[i]
            \itemsep0em 
            \item\label{ite:t-copy1} $e = \textbf{copy}\ e_1\ \textbf{into}\ e_2$
            \item\label{ite:t-copy2} $K \mid \Gamma \mid \Sigma \vdash e_1 : (\tau, \rho) \mid \varphi_1 $
            \item\label{ite:t-copy3} $K \mid \Gamma \mid \Sigma \vdash e_2 : (\tau', \rho') \mid \varphi_2 $
            \item\label{ite:t-copy4} $\varphi = \varphi_1 \times \varphi_2 \times \{\textbf{alloc}\ 1\ \rho' \} $
        \end{enumerate}
        There are three evaluation rules by which $\langle e \mid \storeOuter
        \rangle \longrightarrow \langle e' \mid \storeOuter' \rangle$ can be derived:
        (\textsc{e-copy1}), (\textsc{e-copy2}) and (\textsc{e-copyL}).
        \begin{itemize}

            \item Subcase (\textsc{e-copy1}):\\ 
                From (\textsc{e-copy1}) we have the following assumptions:
                \begin{enumerate}[a.]
                    \itemsep0em 
                    \item\label{ite:e-copy11} $\eformat{e_1}{\storeOuter}{e'_1}{\storeOuter'}$
                    \item\label{ite:e-copy12} $e' = \textbf{copy}\ e'_1\ \textbf{into}\ e_2$
                \end{enumerate}

                By induction on assumption~\ref{ite:t-copy2} we have that:
                \begin{equation}
                    K \mid \Gamma \mid \Sigma' \vdash e'_1 : (\tau, \rho) \mid \varphi'_1 
                    \label{eq:copy11}
                \end{equation}
                which we can use with assumption~\ref{ite:t-copy3} to obtain: 
                \begin{equation}
                    K \mid \Gamma \mid \Sigma' \vdash e_2 : (\tau', \rho') \mid \varphi_2 
                    \label{eq:copy12}
                \end{equation}
                Let $\Sigma'$ equal any context such that $\Sigma' \supseteq
                \Sigma$ and the above typings hold.
                From this we can construct the following typing derivation for 
                $e'$:
                \begin{align*}
                    \inferrule*[right=t-copy]
                        { \eqref{eq:copy11}
                            \qquad
                            \eqref{eq:copy12}
                        }
                        {K \mid \Gamma \mid \Sigma' \vdash \textbf{copy}\ e'_1\ \textbf{into}\ e_2 : (\tau, \rho') \mid \varphi'_1 \times \varphi_2 \times \{\textbf{alloc}\ 1\ \rho' \} }
                \end{align*}
                From this we have that $\varphi' = \varphi'_1 \times \varphi_2 \times \{\textbf{alloc}\ 1\ \rho' \} $ and from assumption
                ~\ref{ite:t-copy4} that $\varphi = \varphi_1 \times \varphi_2 \times \{\textbf{alloc}\ 1\ \rho' \} $. By the typability of
                assumption~\ref{ite:t-copy2} we have that $K \vdash \varphi_1 : \text{Effect}$ \hypertarget{copy11}{(*)}, from the typability of assumption~\ref{ite:t-copy3}
                we have that $K \vdash \varphi_2 : \text{Effect}$ \hypertarget{copy12}{(**)}, and by induction we have that $\varphi'_1 \sqsubseteq \varphi_1$ \hypertarget{copy13}{(***)}. From this we can construct the following derivation for $\varphi'_1 \times \varphi_2 \times \{\textbf{alloc}\ 1\ \rho' \} \sqsubseteq \varphi_1 \times \varphi_2 \times \{\textbf{alloc}\ 1\ \rho' \}$:

                \begin{equation}
                    \inferrule*[right=$\kappa$-$\times$]
                                        {
                                            \qquad
                                            \inferrule*[right=$\kappa$-alloc]
                                                {   \inferrule*[right=$\kappa$-reg]  
                                                        {\quad} 
                                                        {K \vdash \rho' : \text{Region}}
                                                    \qquad 
                                                    \inferrule*[right=$\kappa$-size] 
                                                        {\quad}
                                                        {K \vdash 1 : \text{Size}}
                                                    }
                                                {K \vdash \{\textbf{alloc}\ 1\ \rho' \} : \text{Effect} }
                                        }
                                        {K \vdash \varphi_2 \times  \{\textbf{alloc}\ 1\ \rho' \} : \text{Effect}  }
                    \label{eq:copy15}
                \end{equation}

                \begin{equation}
                        \inferrule*[right=sb-$\times_{below}^{}$]
                                {
                                    \texorpdfstring{\protect\hyperlink{copy13}{\text{(***)}}}{}
                                    \qquad 
                                    \eqref{eq:copy15}
                                }
                                {K \vdash \varphi'_1 \sqsubseteq  \varphi_1 \times \varphi_2 \times \{\textbf{alloc}\ 1\ \rho' \} : \text{Effect}  }
                    \label{eq:copy13}             
                \end{equation}

                \begin{equation}
                        \inferrule*[right=sb-$\equiv$]
                                                        {\inferrule*[right=eq-refl]
                                                            {
                                                                \texorpdfstring{\protect\hyperlink{copy12}{\text{(**)}}}{}
                                                            }
                                                            {K \vdash \varphi_2 \equiv \varphi_2 : \text{Effect}}}
                                                        {K \vdash \varphi_2 \sqsubseteq \varphi_2 : \text{Effect}}
                    \label{eq:copy110}
                \end{equation}

                \begin{equation}
                    \inferrule*[right=$\kappa$-alloc]
                                                        {
                                                            \inferrule*[right=$\kappa$-reg]
                                                                {\quad}
                                                                {K \vdash \rho' : \text{Region}}
                                                            \qquad
                                                            \inferrule*[right=$\kappa$-size]
                                                                {\quad}
                                                                {K \vdash 1 : \text{Size}}
                                                        }
                                                        {K \vdash  \{\textbf{alloc}\ 1\ \rho' \} : \text{Effect}  }                
                    \label{eq:copy111}
                \end{equation}

                \begin{equation}
                            \inferrule*[right=sb-$\times_{below}^{2}$]
                                        {
                                            \texorpdfstring{\protect\hyperlink{copy11}{\text{(*)}}}{}
                                            \qquad 
                                            \inferrule*[right=sb-$\times_{below}^{1}$]
                                                {
                                                    \eqref{eq:copy110}    
                                                    \qquad                                                 
                                                    \eqref{eq:copy111}
                                                }
                                                {K \vdash \varphi_2 \sqsubseteq \varphi_2 \times \{\textbf{alloc}\ 1\ \rho' \} : \text{Effect} }
                                        }
                                        {K \vdash \varphi_2 \sqsubseteq \varphi_1 \times \varphi_2 \times \{\textbf{alloc}\ 1\ \rho' \} : \text{Effect}}
                    \label{eq:copy16}
                \end{equation}

                \begin{equation}
                        \inferrule*[right=$\kappa$-$\times$]
                                                { 
                                                    \texorpdfstring{\protect\hyperlink{copy11}{\text{(*)}}}{}
                                                    \qquad
                                                     \texorpdfstring{\protect\hyperlink{copy12}{\text{(**)}}}{}
                                                }
                                                {K \vdash \varphi_1 \times \varphi_2 : \text{Effect}}
                    \label{eq:copy18}                
                \end{equation}

                \begin{equation}
                    \inferrule*[right=sb-$\equiv$]
                                                {\inferrule*[right=sb-refl]
                                                    { \inferrule*[right=$\kappa$-alloc]
                                                        {\inferrule*[right=$\kappa$-reg]{\quad}{K \vdash \rho' : \text{Region}} \qquad \inferrule*[right=$\kappa$-size]{\quad}{K \vdash 1 : \text{Size}}}{K \vdash \{\textbf{alloc}\ 1\ \rho' \} : \text{Effect} }}
                                                    {K \vdash \{\textbf{alloc}\ 1\ \rho' \} \equiv \{\textbf{alloc}\ 1\ \rho' \} : \text{Effect}}}
                                                {K \vdash \{\textbf{alloc}\ 1\ \rho' \} \sqsubseteq \{\textbf{alloc}\ 1\ \rho' \} : \text{Effect}}                
                    \label{eq:copy19}
                \end{equation}

                \begin{equation}
                        \inferrule*[right=sb-$\times_{above}$]
                                        {
                                           \eqref{eq:copy18} 
                                            \qquad                                          
                                            \eqref{eq:copy19}
                                        }
                                        {K \vdash \{\textbf{alloc}\ 1\ \rho' \} \sqsubseteq \varphi_1 \times \varphi_2 \times \{\textbf{alloc}\ 1\ \rho' \} : \text{Effect} } 
                    \label{eq:copy17}
                \end{equation}

                \begin{equation}
                        \inferrule*[right=sb-$\times_{above}$]
                                {
                                    \eqref{eq:copy16}    
                                    \qquad
                                    \eqref{eq:copy17}                                    
                                }
                                {K \vdash \varphi_2 \times \{\textbf{alloc}\ 1\ \rho' \} \sqsubseteq \varphi_1 \times \varphi_2 \times \{\textbf{alloc}\ 1\ \rho' \} : \text{Effect}}
                    \label{eq:copy14}
                \end{equation}

                \begin{align*}
                    \inferrule*[right=sb-$\times_{above}$]
                        {                         
                           \eqref{eq:copy13} 
                            \qquad 
                            \eqref{eq:copy14}                            
                        }
                        {K \vdash \varphi'_1 \times \varphi_2 \times \{\textbf{alloc}\ 1\ \rho' \} \sqsubseteq \varphi_1 \times \varphi_2 \times \{\textbf{alloc}\ 1\ \rho' \} : \text{Effect} }             
                \end{align*}

                Finally, we can type the output store $\storeOuter'$ by induction;
                \begin{align*}
                    K \mid \Gamma \mid \Sigma' \vdash \storeOuter' 
                \end{align*}

            \item Subcase (\textsc{e-copy2}): \\
                From (\textsc{e-copy2}) we have the following assumptions:
                \begin{enumerate}[a.]
                    \itemsep0em 
                    \item\label{ite:e-copy21} $\eformat{e_2}{\storeOuter}{e'_2}{\storeOuter'}$
                    \item\label{ite:e-copy22} $e' = \textbf{copy}\ l_\rho\ \textbf{into}\ e'_2$
                \end{enumerate}

                By induction on assumption~\ref{ite:t-copy3} we have that:
                \begin{equation}
                    K \mid \Gamma \mid \Sigma' \vdash e'_2 : (\tau', \rho') \mid \varphi'_2 
                    \label{eq:copy21}
                \end{equation}
                which we can use with assumption~\ref{ite:t-copy3} to obtain: 
                \begin{equation}
                    K \mid \Gamma \mid \Sigma' \vdash l_\rho : (\tau, \rho) \mid \varphi_1 
                    \label{eq:copy22}
                \end{equation}
                Let $\Sigma'$ equal any context such that $\Sigma' \supseteq
                \Sigma$ and the above typings hold.
                From this we can construct the following typing derivation for 
                $e'$:
                \begin{align*}
                    \inferrule*[right=t-copy]
                        { \eqref{eq:copy21}
                            \qquad
                            \eqref{eq:copy22}
                        }
                        {K \mid \Gamma \mid \Sigma' \vdash \textbf{copy}\ l_\rho\ \textbf{into}\ e'_2 : (\tau, \rho') \mid \varphi_1 \times \varphi'_2 \times \{\textbf{alloc}\ 1\ \rho' \} }
                \end{align*}
                From this we have that $\varphi' = \varphi_1 \times \varphi'_2 \times \{\textbf{alloc}\ 1\ \rho' \} $ and from assumption
                ~\ref{ite:t-copy4} that $\varphi = \varphi_1 \times \varphi_2 \times \{\textbf{alloc}\ 1\ \rho' \} $. By the typability of
                assumption~\ref{ite:t-copy2} we have that $K \vdash \varphi_1 : \text{Effect}$ \hypertarget{copy21}{(*)}, from the typability of assumption~\ref{ite:t-copy3}
                we have that $K \vdash \varphi_2 : \text{Effect}$ \hypertarget{copy22}{(**)}, and by induction we have that $\varphi'_2 \sqsubseteq \varphi_2$ \hypertarget{copy23}{(***)}. From this we can construct the following derivation for $\varphi_1 \times \varphi'_2 \times \{\textbf{alloc}\ 1\ \rho' \} \sqsubseteq \varphi_1 \times \varphi_2 \times \{\textbf{alloc}\ 1\ \rho' \}$: 

                \begin{equation}
                \inferrule*[right=$\kappa$-$\times$]
                                        {
                                            \texorpdfstring{\protect\hyperlink{copy22}{\text{(**)}}}{}
                                            \qquad 
                                            \inferrule*[right=$\kappa$-alloc]
                                                {
                                                    \inferrule*[right=$\kappa$-reg]
                                                        {\quad}
                                                        {K \vdash \rho' : \text{Region}}
                                                    \qquad
                                                    \inferrule*[right=$\kappa$-size]
                                                        {\quad}
                                                        {K \vdash 1 : \text{Size}}
                                                }
                                                {K \vdash \{\textbf{alloc}\ 1\ \rho' \} : \text{Effect} }
                                        }
                                        {K \vdash \varphi_2 \times \{\textbf{alloc}\ 1\ \rho' \} : \text{Effect} }
                    \label{eq:copy26}
                \end{equation}

                \begin{equation}
                \inferrule*[right=sb-$\times_{below}^{1}$]
                                {
                                    \texorpdfstring{\protect\hyperlink{copy21}{\text{(*)}}}{}
                                    \qquad                                     
                                    \eqref{eq:copy26}
                                } 
                                {K \vdash \varphi_1 \sqsubseteq \varphi_1 \times \varphi_2 \times \{\textbf{alloc}\ 1\ \rho' \} : \text{Effect}}
                    \label{eq:copy23}
                \end{equation}

                \begin{equation}
                    \inferrule*[right=sb-$\times_{below}^{1}$]  
                                                {
                                                    \texorpdfstring{\protect\hyperlink{copy23}{\text{(***)}}}{}
                                                    \qquad
                                                    \inferrule*[right=$\kappa$-alloc]
                                                        {
                                                            \inferrule*[right=$\kappa$-reg]
                                                                {\quad}
                                                                {K \vdash \rho' : \text{Region}}
                                                            \qquad
                                                            \inferrule*[right=$\kappa$-size]
                                                                {\quad}
                                                                {K \vdash 1 : \text{Size}}
                                                        }                                            
                                                        {K \vdash \{\textbf{alloc}\ 1\ \rho' \} : \text{Effect} }
                                                }
                                                {K \vdash \varphi'_2 \sqsubseteq \varphi_2 \times \{\textbf{alloc}\ 1\ \rho' \} : \text{Effect} }
                    \label{eq:copy27}
                \end{equation}

                \begin{equation}
                        \inferrule*[right=sb-$\times_{below}^{1}$]
                                        {
                                            \texorpdfstring{\protect\hyperlink{copy22}{\text{(**)}}}{}
                                            \qquad 
                                           \eqref{eq:copy27} 
                                        }
                                        {K \vdash \varphi'_2 \sqsubseteq  \varphi_1 \times \varphi_2 \times \{\textbf{alloc}\ 1\ \rho' \} : \text{Effect} }
                    \label{eq:copy29}
                \end{equation}

                \begin{equation}
                \inferrule*[right=sb-$\equiv$]
                                                        {\inferrule*[right=eq-refl]
                                                            {\inferrule*[right=$\kappa$-alloc]
                                                                {
                                                                    \inferrule*[right=$\kappa$-reg]
                                                                        {\quad}
                                                                        {K \vdash \rho' : \text{Region}}
                                                                    \qquad
                                                                    \inferrule*[right=$\kappa$-size]
                                                                        {\quad}
                                                                        {K \vdash 1 : \text{Size}}
                                                                }
                                                                {K \vdash \{\textbf{alloc}\ 1\ \rho' \} : \text{Effect} }}
                                                            {K \vdash \{\textbf{alloc}\ 1\ \rho' \} \equiv \{\textbf{alloc}\ 1\ \rho' \} : \text{Effect}}}
                                                        {K \vdash \{\textbf{alloc}\ 1\ \rho' \} \sqsubseteq \{\textbf{alloc}\ 1\ \rho' \} : \text{Effect} }
                    \label{eq:copy28}
                \end{equation}

                \begin{equation}
                    \inferrule*[right=sb-$\times_{below}^{2}$]
                                        {
                                            \texorpdfstring{\protect\hyperlink{copy21}{\text{(*)}}}{}
                                            \qquad
                                            \inferrule*[right=sb-$\times_{below}^{2}$]
                                                {
                                                    \texorpdfstring{\protect\hyperlink{copy22}{\text{(**)}}}{}
                                                    \qquad
                                                    \eqref{eq:copy28} 
                                                }
                                                {K \vdash \{\textbf{alloc}\ 1\ \rho' \} \sqsubseteq \varphi_2 \times \{\textbf{alloc}\ 1\ \rho' \} : \text{Effect}}
                                        }
                                        { K \vdash \{\textbf{alloc}\ 1\ \rho' \} \sqsubseteq \varphi_1 \times \varphi_2 \times \{\textbf{alloc}\ 1\ \rho' \} : \text{Effect} }
                    \label{eq:copy25}
                \end{equation}

                \begin{equation}
                    \inferrule*[right=sb-$\times_{below}^{1}$]
                                {
                                   \eqref{eq:copy29} 
                                    \qquad
                                    \eqref{eq:copy25}                              
                                }
                                {K \vdash \varphi'_2 \times \{\textbf{alloc}\ 1\ \rho' \} \sqsubseteq \varphi_1 \times \varphi_2 \times \{\textbf{alloc}\ 1\ \rho' \} : \text{Effect} }                
                    \label{eq:copy24}
                \end{equation}

                \begin{align*}
                    \inferrule*[right=sb-$\times_{above}$]
                        {   
                            \eqref{eq:copy23}
                            \qquad 
                            \eqref{eq:copy24} 
                        }
                        {K \vdash \varphi_1 \times \varphi'_2 \times \{\textbf{alloc}\ 1\ \rho' \} \sqsubseteq \varphi_1 \times \varphi_2 \times \{\textbf{alloc}\ 1\ \rho' \} : \text{Effect}}
                \end{align*}

                Finally, we can type the output store $\storeOuter'$ by induction;
                \begin{align*}
                    K \mid \Gamma \mid \Sigma' \vdash \storeOuter' 
                \end{align*}

            \item Subcase (\textsc{e-copyL}): \\
                From (\textsc{e-copyL}) we have the following assumptions:
                \begin{enumerate}[a.]
                    \itemsep0em 
                    \item\label{ite:e-copyL1} $(\storeInner_{\rho}^{in.}, s_a) = \storeOuter(\rho)$
                    \item\label{ite:e-copyL2} $(\storeInner_{\rho'}^{in.}, s'_a) = \storeOuter(\rho')$
                    \item\label{ite:e-copyL3} $v = \storeInner_{\rho'}^{in.} (l_\rho)$
                    \item\label{ite:e-copyL4} $l'_{\rho'}  = \text{freshLoc}(\rho')$
                    \item\label{ite:e-copyL5} $e' = l'_{\rho'}$
                    \item\label{ite:e-copyL6} $\storeOuter' = [\rho' \mapsto ((\storeInner_{\rho'}^{in.}, l'_{\rho'} \mapsto l_{\rho}), s'_a)] \storeOuter$
                \end{enumerate}

                From assumptions~\ref{ite:e-copyL4} and~\ref{ite:e-copyL6} we have that $l'_{\rho'}$ is a fresh location in
                $\rho'$ that points to $v$. Since this evaluation rule takes no inductive steps, let $\Sigma'$ equal $\Sigma$ with the 
                additional location $l'_{\rho'}$ such that $\tau = \Sigma(l'_{\rho'})$ \hypertarget{copyL1}{(*)}. 
                Furthermore, from the typability of assumption~\ref{ite:t-copy2} we know that $K \vdash \tau : \text{Type}$ \hypertarget{copyL2}{(**)}.
                We can then construct the following typing derivation for $e'$:

                \begin{align*}
                        \inferrule*[right=t-use-val]
                                { \inferrule*[]
                                    { \texorpdfstring{\protect\hyperlink{copyL1}{\text{(*)}}}{} \qquad \texorpdfstring{\protect\hyperlink{copyL2}{\text{(**)}}}{} }
                                    {K \mid \Gamma \mid \Sigma' \vdash l'_{\rho'} : \tau } }
                                {K \mid \Gamma \mid \Sigma' \vdash l'_{\rho'} : (\tau, \rho') \mid \{\bot\}}
                        \end{align*}
                        For which we have that $\varphi' = \{\bot\} $, and from assumption~\ref{ite:t-copy4} that 
                        $\varphi = \varphi_1 \times \varphi_2 \times \{\textbf{alloc}\ 1\ \rho' \}$. 
                        By the typability of
                assumption~\ref{ite:t-copy2} we have that $K \vdash \varphi_1 : \text{Effect}$ \hypertarget{copyL3}{(***)}, and from the typability of assumption~\ref{ite:t-copy3}
                we have that $K \vdash \varphi_2 : \text{Effect}$ \hypertarget{copyL4}{(****)}. From this we can construct the following derivation for $\{\bot \} \sqsubseteq \varphi_1 \times \varphi_2 \times \{\textbf{alloc}\ 1\ \rho' \}$:

                \begin{equation}
                \inferrule*[right=$\kappa$-$\times$]
                                        {
                                            \texorpdfstring{\protect\hyperlink{copyL3}{\text{(***)}}}{}
                                            \qquad
                                            \texorpdfstring{\protect\hyperlink{copyL4}{\text{(****)}}}{}
                                        }
                                        {K \vdash \varphi_1 \times \varphi_2 : \text{Effect}}
                    \label{eq:copyL2}
                \end{equation}

                \begin{equation}
                    \inferrule*[right=$\kappa$-alloc]
                                        {
                                            \inferrule*[right=$\kappa$-reg]
                                                {\quad}
                                                {K \vdash \rho' : \text{Region}}
                                            \qquad
                                            \inferrule*[right=$\kappa$-size]
                                                {\quad}
                                                {K \vdash 1 : \text{Size}}
                                        }
                                        {K \vdash \{\textbf{alloc}\ 1\ \rho' \} : \text{Effect} }
                    \label{eq:copyL1}
                \end{equation}

                        \begin{align*}
                        \inferrule*[right=sb-$\bot$]  
                            {\inferrule*[right=$\kappa$-$\times$]
                                {
                                   \eqref{eq:copyL2} 
                                    \qquad
                                   \eqref{eq:copyL1} 
                                }
                                {K \vdash \varphi_1 \times \varphi_2 \times \{\textbf{alloc}\ 1\ \rho' \} : \text{Effect} }}
                            {K \vdash \{\bot\} \sqsubseteq  \varphi_1 \times \varphi_2 \times \{\textbf{alloc}\ 1\ \rho' \} : \text{Effect}  } 
                        \end{align*}

                From our inductive hypothesis we have that: 
                \begin{align*}
                    K \mid \Gamma \mid \Sigma \vdash \storeOuter
                \end{align*}
                i.e., that the store prior to evaluation is well-typed~(\hypertarget{copyLStore1}{\text{\textdagger}}). 
                Therefore we have that currentSize($\rho$) $\sqsubseteq s_a$ prior to 
                evaluation~(\hypertarget{copyLStore2}{\text{\textdagger\textdagger}}). From this we have:
                \begin{align*}   
                   K \mid \Gamma \mid \Sigma \vdash \storeInner_{\rho}^{in.}
                \end{align*} 
                i.e., that the store for region $\rho$ is well-typed prior to the evaluation of $e$~(\hypertarget{copyLStore2}{\text{\textdagger\textdagger}})
                As per assumption~\ref{ite:e-copyL6}, (\textsc{e-copyL}) extends the region $\rho$ in the store with an additional location 
                $l'_{\rho}$ which points to location $l_{\rho}$.

                From the premise of our theorem statement, we know that $e$ will eventually be a part of a derivation 
                for which the typing of this allocation holds. Thus, using inversion (Lemma~\ref{lemma:typing-inversion}) on 
                the typing relation, along with the fact that sizeOf($l_{\rho}$) $= 1$ (as per Definition~\ref{def:sizeOf}) we have that $\exists s' . \text{currentSize}(\rho) + 1 + s' \sqsubseteq s_a$.

                From this and $\texorpdfstring{\protect\hyperlink{copyLStore3}{(\text{\textdagger\textdagger\textdagger})}}{}$ we know that 
                currentSize($\rho$) + $s \sqsubseteq s_a$~(\hypertarget{copyLStore4}{\text{\textdagger\textdagger\textdagger\textdagger}})
                , i.e., that the size of the region after the evaluation of $e$ does not exceed the maximum allocation bound for $\rho$.

                From the above we can construct the following typing derivation for the output store $\storeOuter'$: 

                \begin{equation}
                    \inferrule*[right=loc]
                        {
                            \texorpdfstring{\protect\hyperlink{copyL1}{(*)}}{}
                            \qquad 
                            (\texorpdfstring{\protect\hyperlink{copyLStore2}{\text{\textdagger\textdagger}}}{})
                            \qquad
                            \ref{ite:t-copy3}
                        }
                        {K \mid \Gamma \mid \Sigma \vdash ((\storeInner_{\rho}^{in.}, l'_{\rho} \mapsto l_{\rho}), s_a) }
                        \label{eq:copyLStore1}
                \end{equation}

                \begin{align*}
                    \inferrule*[right=subst]
                        { 
                            \texorpdfstring{\protect\hyperlink{copyLStore1}{(\text{\textdagger})}}{}
                            \qquad
                            \eqref{eq:copyLStore1}
                            \qquad
                            \texorpdfstring{\protect\hyperlink{copyLStore3}{(\text{\textdagger\textdagger\textdagger})}}{}
                        }
                        {K \mid \Gamma \mid \Sigma \vdash [\rho \mapsto ((\storeInner_{\rho}^{in.}, l'_{\rho} \mapsto l_{\rho}), s_a)] \storeOuter }
                \end{align*}

        \end{itemize}

    \item Case (\textsc{t-val}):\\
        From (\textsc{t-val}) we have the following assumptions:
        \begin{enumerate}[i]
            \itemsep0em 
            \item\label{ite:t-val1} $e = v\ [s]\ \textbf{at}\ e_1$
            \item\label{ite:t-val2} $K \mid \Gamma \mid \Sigma \vdash e_1 : (\tau', \rho) \mid \varphi_1$
            \item\label{ite:t-val3} $K \mid \Gamma \mid \Sigma \vdash v : \tau$
            \item\label{ite:t-val4} $K \vdash s : \text{Size}$
            \item\label{ite:t-val5} $\varphi = \varphi_1 \times \{ \textbf{alloc}\ s\ \rho\} $
        \end{enumerate}
        There are two evaluation rules by which $\langle e \mid \storeOuter
        \rangle \longrightarrow \langle e' \mid \storeOuter' \rangle $ can be
        derived: (\textsc{e-val}) and (\textsc{e-valL}).
        \begin{itemize}

            \item Subcase (\textsc{e-val}):\\
                From (\textsc{e-val}) we have the following assumptions:
                \begin{enumerate}[a.]
                    \itemsep0em 
                    \item\label{ite:e-val1} $\eformat{e'_1}{\storeOuter}{e'_1}{\storeOuter'}$
                    \item\label{ite:e-val2} $e' = v\ [s]\ \textbf{at}\ e'_1$
                \end{enumerate}
                By induction on assumption~\ref{ite:t-val2} we have: 
                \begin{equation}
                    K \mid \Gamma \mid \Sigma' \vdash e'_1 : (\tau', \rho) \mid \varphi_1'
                    \label{eq:e-val1}
                \end{equation}
                Let $\Sigma'$ equal any store typing such that $\Sigma' \supseteq \Sigma$ and the above typing holds. We can thus use $\Sigma'$ with assumption~\ref{ite:t-val3} to obtain: 
                \begin{equation}
                    K \mid \Gamma  \mid \Sigma' \vdash v : \tau
                    \label{eq:e-val2}
                \end{equation} 
                Using~\ref{eq:e-val1} and~\ref{eq:e-val2}, along with assumption~\ref{ite:t-val4}, we can construct the following typing
                derivation for $e'$:
                \begin{align*}
                    \inferrule*[right=t-val]
                        {
                            \eqref{eq:e-val1}
                            \qquad
                            \eqref{eq:e-val2}
                            \qquad 
                        \ref{ite:t-val4}
                        }             
                        {K \mid \Gamma \mid \Sigma' \vdash v \ [s]\ \textbf{at}\ e'_1 : (\tau, \rho) \mid \varphi'_1 \times \{ \textbf{alloc}\ s\ \rho\} }
                \end{align*}
                From which we have $\varphi' = \varphi'_1 \times \{ \textbf{alloc}\ s\ \rho\}$ and by 
                assumption~\ref{ite:t-val5} we have $\varphi = \varphi_1 \times \{ \textbf{alloc}\ s\ \rho \} $. 
                From the typability of assumption~\ref{ite:t-val2} we have that $K \vdash \varphi_1 : \text{Effect}$ \hypertarget{val1}{(*)}, and by induction 
                we have that $\varphi_1 \sqsubseteq \varphi'_1$  \hypertarget{val2}{(**)} which we can use to construct the following 
                derivation for $\varphi'_1 \times \{ \textbf{alloc}\ s\ \rho\} \sqsubseteq \varphi_1 \times \{ \textbf{alloc}\ s\ \rho\} $:

                \begin{equation}
                                \inferrule*[right=$\kappa$-alloc]
                                    {\inferrule*[right=$\kappa$-size]
                                        {\quad}
                                        {K \vdash s : \text{Size}}
                                        \qquad 
                                        \inferrule*[right=$\kappa$-reg]
                                            {\quad}
                                            {K \vdash \rho : \text{Region}}}
                                    {K \vdash \{\textbf{alloc}\ s\ \rho \} : \text{Effect}}
                \label{eq:e-val-phi2}
                \end{equation}

                \begin{equation}
                    \inferrule*[right=$\times_{below}^1$]
                            {
                                \texorpdfstring{\protect\hyperlink{val2}{\text{(**)}}}{}
                                \qquad 
                                \eqref{eq:e-val-phi2}
                            }
                            {K \vdash \varphi'_1 \sqsubseteq \{\textbf{alloc}\ s\ \rho \} \times \varphi_1 : \text{Effect}} 
                    \label{eq:e-val-phi1}
                \end{equation}

                \begin{equation}
                    \inferrule*[right=sb-$\equiv$]
                                {
                                    \inferrule*[right=eq-refl]
                                        {
                                            \inferrule*[right=$\kappa$-alloc]  
                                                {\inferrule*[right=$\kappa$-reg]{\quad}{K \vdash \rho : \text{Region}} \qquad \inferrule*[right=$\kappa$-size]{\quad}{K \vdash s : \text{Size}}}
                                                {K \vdash : \{\textbf{alloc}\ s\ \rho \} : \text{Effect}}
                                        }
                                        {K \vdash \{\textbf{alloc}\ s\ \rho\} \equiv \{\textbf{alloc}\ s\ \rho\} : \text{Effect}}
                                }
                                {K \vdash \{ \textbf{alloc}\ s\ \rho \} \sqsubseteq \{ \textbf{alloc}\ s\ \rho \} : \text{Effect} }
                            \label{eq:e-val-phi4}
                \end{equation}

                \begin{equation}
                        \inferrule*[right=sb-$\times_{below}^{1}$]
                            { 
                                \eqref{eq:e-val-phi4}
                                \qquad 
                                \texorpdfstring{\protect\hyperlink{val1}{\text{(*)}}}{}
                                }
                            {K \vdash  \{ \textbf{alloc}\ s\ \rho \}  \sqsubseteq \varphi_1 \times \{ \textbf{alloc}\ s\ \rho \} : \text{Effect} }
                    \label{eq:e-val-phi3}
                \end{equation}

                \begin{align*}
                    \inferrule*[right=sb-$\times_{above}$]
                        { \eqref{eq:e-val-phi1}
                        \qquad 
                        \eqref{eq:e-val-phi3}
                        }
                        {K \vdash \varphi'_1 \times \{ \textbf{alloc}\ s\ \rho\} \sqsubseteq \varphi_1 \times \{ \textbf{alloc}\ s\ \rho\} : \text{Effect} }
                \end{align*}

                Finally, we can type the output store $\storeOuter'$ by induction:
                \begin{align*}
                    K \mid \Gamma \mid \Sigma' \vdash \storeOuter'
                \end{align*}
                
            \item Subcase (\textsc{e-valL}):\\
                From (\textsc{e-valL}) we have the following assumptions:
                \begin{enumerate}[a.]
                    \itemsep0em 
                    \item\label{ite:e-valL1} $(\storeInner_{\rho}^{in.}, s_a) = \storeOuter(\rho)$
                    \item\label{ite:e-valL2} $s_v = \text{sizeOf}(v)$
                    \item\label{ite:e-valL3} $s_v \sqsubseteq s$ 
                    \item\label{ite:e-valL4} $l'_{\rho} = \text{freshLoc}(\rho)$
                    \item\label{ite:e-valL5} $e' = l'_{\rho}$
                    \item\label{ite:e-valL6} $[\rho \mapsto ((\storeInner_{\rho}^{in.}, l'_{\rho} \mapsto v), s_a)] \storeOuter$
                \end{enumerate}

                From assumptions~\ref{ite:e-valL4} and~\ref{ite:e-valL6} we have that $l'_{\rho}$ is a fresh location in
                $\rho$ that points to $v$. Since this evaluation rule takes no inductive steps, let $\Sigma'$ equal $\Sigma$ with the 
                additional location $l'_{\rho}$ such that $\tau = \Sigma(l'_\{\rho\})$ \hypertarget{valL1}{(*)}. 
                Furthermore, from the typability of assumption~\ref{ite:t-val1} we know that $K \vdash \tau : \text{Type}$ \hypertarget{valL2}{(**)}.
                We can then construct the following typing derivation for $e'$:
                
                \begin{align*}
                    \inferrule*[right=t-use-val]
                        { \inferrule*[right=t-loc]
                            { \texorpdfstring{\protect\hyperlink{valL1}{\text{(*)}}}{} \qquad  \texorpdfstring{\protect\hyperlink{valL2}{\text{(**)}}}{} }
                            {K \mid \Gamma \mid \Sigma' \vdash l'_{\rho} : \tau} }
                        {K \mid \Gamma \mid \Sigma' \vdash l'_{\rho} : (\tau, \rho) \mid \{ \bot \} }
                \end{align*}
                From which we have that $\varphi' = \{\bot\}$ and by assumption~\ref{ite:t-val5} we have that  
                $\varphi = \varphi_1 \times \{\textbf{alloc}\ s\ \rho\}$. By the typability of assumption~\ref{ite:t-val1}
                we have that $\varphi_1 : \text{Effect}$ \hypertarget{valL3}{(***)}, from which we can construct the following 
                derivation for $\{\bot \} \sqsubseteq \varphi_1 \times \{\textbf{alloc}\ s\ \rho\} $:

                \begin{equation}
                    \inferrule*[right=$\kappa$-alloc]{
                                \inferrule*[right=$\kappa$-size]
                                    {\quad}
                                    { K \vdash s : \text{Size}}
                            \qquad 
                            \inferrule*[right=$\kappa$-reg]
                                {\quad}
                                {K \vdash \rho : \text{Region}}
                            }{ K \vdash \{\textbf{alloc}\ s\ \rho\} : \text{Effect}}
                    \label{eq:valL-phi-1}
                \end{equation}

                \begin{align*}
                    \inferrule*[right=sb-$\bot$]
                        { \inferrule*[right=$\kappa$-$\times$]
                            { \texorpdfstring{\protect\hyperlink{valL3}{\text{(***)}}}{} \qquad \eqref{eq:valL-phi-1} }
                            {K \vdash \varphi_1 \times \{\textbf{alloc}\ s\ \rho\} : \text{Effect} }}
                        {K \vdash \{ \bot \} \sqsubseteq \varphi_1 \times \{\textbf{alloc}\ s\ \rho\} : \text{Effect} }
                \end{align*}

                From our inductive hypothesis we have that: 
                \begin{align*}
                    K \mid \Gamma \mid \Sigma \vdash \storeOuter
                \end{align*}
                i.e., that the store prior to evaluation is well-typed~(\hypertarget{valLStore1}{\text{\textdagger}}). 
                Therefore we have that currentSize($\rho$) $\sqsubseteq s_a$ prior to 
                evaluation~(\hypertarget{valLStore2}{\text{\textdagger\textdagger}}). From this we have:
                \begin{align*}   
                   K \mid \Gamma \mid \Sigma \vdash \storeInner_{\rho}^{in.}
                \end{align*} 
                i.e., that the store for region $\rho$ is well-typed prior to the evaluation of $e$~(\hypertarget{valLStore2}{\text{\textdagger\textdagger}})
                As per assumption~\ref{ite:e-valL6}, (\textsc{e-valL}) extends the region $\rho$ in the store with an additional location 
                $l'_{\rho}$ which points to value $v$.

                From the premise of our theorem statement, we know that $e$ will eventually be a part of a derivation 
                for which the typing of this allocation holds. Thus, using inversion (Lemma~\ref{lemma:typing-inversion}), on 
                this typing relation, we have that $\exists s' . \text{currentSize}(\rho) + s + s' \sqsubseteq s_a$.

                Furthermore, from assumption~\ref{ite:e-valL2} we have that $s_v$ is the actual size of $v$ (the value being allocated), 
                and from~\ref{ite:e-valL3} we have that $s_v \sqsubseteq s$. Therefore, we can substitute $s_v$ for $s$ in 
                the above constraint to yield: $\exists s' . \text{currentSize}(\rho) + s_v + s' \sqsubseteq s_a$, which holds 
                by transitivity.

                From this and $\texorpdfstring{\protect\hyperlink{valLStore2}{(\text{\textdagger\textdagger})}}{}$ we know that 
                currentSize($\rho$) + $s_v \sqsubseteq s_a$~(\hypertarget{valLStore3}{\text{\textdagger\textdagger\textdagger}})
                , i.e., that the size of the region after the evaluation of $e$ does not exceed the maximum allocation bound for $\rho$.

                From the above we can construct the following typing derivation for the output store $\storeOuter'$: 

                \begin{equation}
                    \inferrule*[right=loc]
                        {
                            \texorpdfstring{\protect\hyperlink{valL1}{(*)}}{}
                            \qquad 
                            (\texorpdfstring{\protect\hyperlink{valLStore2}{\text{\textdagger\textdagger}}}{})
                            \qquad
                            \ref{ite:t-val3}
                        }
                        {K \mid \Gamma \mid \Sigma \vdash ((\storeInner_{\rho}^{in.}, l'_{\rho} \mapsto v), s_a) }
                        \label{eq:valStore1}
                \end{equation}

                \begin{align*}
                    \inferrule*[right=subst]
                        { 
                            \texorpdfstring{\protect\hyperlink{valLStore1}{(\text{\textdagger})}}{}
                            \qquad
                            \eqref{eq:valStore1}
                            \qquad
                            \texorpdfstring{\protect\hyperlink{valLStore3}{(\text{\textdagger\textdagger\textdagger})}}{}
                        }
                        {K \mid \Gamma \mid \Sigma \vdash [\rho \mapsto ((\storeInner_{\rho}^{in.}, l'_{\rho} \mapsto v), s_a)] \storeOuter }
                \end{align*}

        \end{itemize}

    \item Case (\textsc{t-app}):\\
        From (\textsc{t-app}) we have the following assumptions:
        \begin{enumerate}[i]
            \itemsep0em 
            \item\label{ite:t-app1} $e = e_1\ e_2$
            \item\label{ite:t-app2} $K \mid \Gamma \mid \Sigma \vdash e_1 : (\mu_1 \xrightarrow{\varphi} \mu_2, \rho) \mid \varphi_1$
            \item\label{ite:t-app3} $K \mid \Gamma \mid \Sigma \vdash e_2 : \mu_1 \mid \varphi_2$
            \item\label{ite:t-app4} $\varphi = \varphi_1 \times \varphi_2 \times \varphi$
        \end{enumerate}
        There are three evaluation rules by which $\langle e \mid \storeOuter
        \rangle \longrightarrow \langle e' \mid \storeOuter' \rangle$ can be derived:
        (\textsc{e-app1}), (\textsc{e-app2}) and (\textsc{e-appL}).
        \begin{itemize}

            \item Subcase (\textsc{e-app1}): \\ 
                From (\textsc{e-app1}) we have the following assumptions:
                \begin{enumerate}[a.]
                    \itemsep0em 
                    \item\label{ite:e-app11} $\eformat{e_1}{\storeOuter}{e'_1}{\storeOuter'}$
                    \item\label{ite:e-app12} $e' = e'_1\ e_2$
                \end{enumerate}
                By induction on assumption~\ref{ite:t-app2} we have: 
                \begin{equation}
                    K \mid \Gamma \mid \Sigma' \vdash e'_1 : (\mu_1 \xrightarrow{\varphi} \mu_2, \rho) \mid \varphi'_1  
                    \label{eq:app11}
                \end{equation}
                which we can combine with assumption~\ref{ite:t-app3} to type: 
                \begin{equation}
                    K \mid \Gamma \mid \Sigma' \vdash e_2 : \mu_1 \mid \varphi_2 
                    \label{eq:app12}
                \end{equation} 
                Let $\Sigma'$ equal any context such that $\Sigma' \supseteq
                \Sigma$ and the above typings hold.

                From~\ref{eq:app11} and~\ref{eq:app12}, we can construct the following typing derivation for $e'$:
                \begin{align*}
                    \inferrule*[right=t-app]
                        { \eqref{eq:app11}  \qquad \eqref{eq:app12} }
                        {K \mid \Gamma \mid \Sigma' \vdash e'_1\ e_2 : \mu_2 \mid \varphi'_1 \times \varphi_2 \times \varphi }
                \end{align*}
                From which we have that $\varphi' = \varphi'_1 \times \varphi_2 \times \varphi$
                and from assumption~\ref{ite:t-app4}, that $\varphi = \varphi_1 \times \varphi_2 \times \varphi$. 
                By the typability of~\ref{ite:t-app1}, we have that $K \vdash \varphi_1 : \text{Effect}$  \hypertarget{app11}{(*)},  
                that $K \vdash \varphi_2 : \text{Effect}$ \hypertarget{app12}{(**)}, $K \vdash \varphi : \text{Effect}$ \hypertarget{app13}{(***)}, 
                and by induction we have that $\varphi'_1 \sqsubseteq \varphi_1$ \hypertarget{app14}{(****)}. 
                From this we can construct the following derivation for $\varphi'_1 \times \varphi_2 \times \varphi 
                \sqsubseteq \varphi_1 \times \varphi_2 \times \varphi$:

                \begin{equation}
                    \inferrule*[right=sb-$\times_{below}^{1}$]
                                {
                                     \texorpdfstring{\protect\hyperlink{app14}{\text{(****)}}}{} 
                                    \qquad 
                                    \inferrule*[right=$\kappa$-$\times$]
                                        {\texorpdfstring{\protect\hyperlink{app12}{\text{(**)}}}{} 
                                        \qquad 
                                        \texorpdfstring{\protect\hyperlink{app13}{\text{(***)}}}{} }
                                        {K \vdash \varphi_2 \times \varphi : \text{Effect}}
                                }
                                { K \vdash \varphi'_1 \sqsubseteq \varphi_1 \times \varphi_2 \times \varphi : \text{Effect} }
                        \label{eq:app13} 
                \end{equation}

                \begin{equation}
                            \inferrule*[right=sb-$\equiv$]
                                                {\inferrule*[right=eq-refl]
                                                    { \texorpdfstring{\protect\hyperlink{app12}{\text{(**)}}}{}  }
                                                    { K \vdash \varphi_2 \equiv \varphi_2 : \text{Effect}}}
                                                {K \vdash \varphi_2 \sqsubseteq \varphi_2 : \text{Effect}}
                    \label{eq:app18}
                \end{equation}
                
                \begin{equation}
                    \inferrule*[right=sb-$\times_{below}^{1}$]
                                { 
                                    \inferrule*[right=sb-$\times_{below}^{2}$]
                                        {  
                                            \texorpdfstring{\protect\hyperlink{app11}{\text{(*)}}}{}
                                            \qquad 
                                            \eqref{eq:app18} 
                                            }
                                        {K \vdash \varphi_2 \sqsubseteq \varphi_1 \times \varphi_2 : \text{Effect} }
                                    \qquad  
                                        \texorpdfstring{\protect\hyperlink{app13}{\text{(***)}}}{}                               
                                }
                                { K \vdash \varphi_2 \sqsubseteq \varphi_1 \times \varphi_2 \times \varphi : \text{Effect}}                
                                \label{eq:app15}
                \end{equation}

                \begin{equation}
                    \inferrule*[right=sb-$\times_{above}$]
                            {
                                \eqref{eq:app13}
                                \qquad 
                                \eqref{eq:app15}                            
                                }
                            { K \vdash \varphi'_1 \times \varphi_2 \sqsubseteq \varphi_1 \times \varphi_2 \times \varphi : \text{Effect} }
                    \label{eq:app14}             
                \end{equation}
                
                \begin{equation}
                        \inferrule*[right=$\kappa$-$\times$]
                                    { 
                                       \texorpdfstring{\protect\hyperlink{app11}{\text{(*)}}}{}
                                        \qquad 
                                        \texorpdfstring{\protect\hyperlink{app12}{\text{(**)}}}{}
                                     }
                                    {K \vdash \varphi_1 \times \varphi_2 : \text{Effect}}                
                    \label{eq:app17}
                \end{equation}

                \begin{equation}
                    \inferrule*[right=sb-$\times_{below}^{2}$]
                            {
                                \eqref{eq:app17} 
                                 \qquad
                                 \inferrule*[right=sb-$\equiv$]  
                                    {\inferrule*[right=eq-refl]
                                        { \texorpdfstring{\protect\hyperlink{app13}{\text{(***)}}}{}}
                                        { K \vdash \varphi \equiv \varphi : \text{Effect}}}
                                    {K \vdash \varphi \sqsubseteq \varphi : \text{Effect} }                                 
                                 }
                            {K \vdash \varphi \sqsubseteq  \varphi_1 \times \varphi_2 \times \varphi : \text{Effect} }
                \label{eq:app16}
                \end{equation}

                \begin{align*}
                    \inferrule*[right=sb-$\times_{above}$]
                        { \eqref{eq:app14}
                            \qquad 
                            \eqref{eq:app16}
                        }
                        {K \vdash \varphi'_1 \times \varphi_2 \times \varphi \sqsubseteq \varphi_1 \times \varphi_2 \times \varphi : \text{Effect} }
                \end{align*}

                Finally, we can type the output store $\storeOuter'$ by induction:
                \begin{align*}
                    \Gamma \mid \Delta'_i \mid \Sigma' \vdash \storeOuter'
                \end{align*}

            \item Subcase (\textsc{e-app2}): \\
                From (\textsc{e-app2}) we have the following assumptions:
                \begin{enumerate}[a.]
                    \itemsep0em 
                    \item\label{ite:e-app21} $\eformat{e_2}{\storeOuter}{e'_2}{\storeOuter'}$
                    \item\label{ite:e-app22} $e' = l_\rho\ e'_2$
                \end{enumerate}
                By induction on assumption~\ref{ite:t-app3} we have: 
                \begin{equation}
                    K \mid \Gamma \mid \Sigma' \vdash e'_2 : \mu_1 \mid \varphi'_2 
                    \label{eq:app21}
                \end{equation} 
                which we can combine with assumption~\ref{ite:t-app2} to type: 
                \begin{equation}
                    K \mid \Gamma \mid \Sigma' \vdash l_\rho : (\mu_1 \xrightarrow{\varphi} \mu_2, \rho) \mid \varphi_1  
                    \label{eq:app22}
                \end{equation}
                Let $\Sigma'$ equal any context such that $\Sigma' \supseteq
                \Sigma$ and the above typings hold.

                From~\ref{eq:app21} and~\ref{eq:app22}, we can construct the following typing derivation for $e'$:
                \begin{align*}
                    \inferrule*[right=t-app]
                        { \eqref{eq:app21}  \qquad \eqref{eq:app22} }
                        {K \mid \Gamma \mid \Sigma' \vdash l_\rho\ e'_2 : \mu_2 \mid \varphi_1 \times \varphi'_2 \times \varphi }
                \end{align*}
                From which we have that $\varphi' = \varphi_1 \times \varphi'_2 \times \varphi$
                and from assumption~\ref{ite:t-app4}, that $\varphi = \varphi_1 \times \varphi_2 \times \varphi$. 
                By the typability of~\ref{ite:t-app1}, we have that $K \vdash \varphi_1 : \text{Effect}$  \hypertarget{app21}{(*)},  
                that $K \vdash \varphi_2 : \text{Effect}$ \hypertarget{app22}{(**)}, $K \vdash \varphi : \text{Effect}$ \hypertarget{app23}{(***)}, 
                and by induction we have that $\varphi'_2 \sqsubseteq \varphi_2$ \hypertarget{app24}{(****)}. 
                From this we can construct the following derivation for $\varphi_1 \times \varphi'_2 \times \varphi 
                \sqsubseteq \varphi_1 \times \varphi_2 \times \varphi$:

                \begin{equation}
                        \inferrule*[right=sb-$\times_{below}^{1}$]
                                    { 
                                        \inferrule*[right=sb-$\equiv$]
                                            {\inferrule*[right=eq-refl]
                                                { 
                                                    \texorpdfstring{\protect\hyperlink{app21}{\text{(*)}}}{}
                                                }
                                                {K \vdash \varphi_1 \equiv \varphi_1 : \text{Effect}}}
                                            {K \vdash \varphi_1 \sqsubseteq  \varphi_1 : \text{Effect}}
                                        \qquad
                                                \texorpdfstring{\protect\hyperlink{app23}{\text{(***)}}}{}     
                                    }
                                    {K \vdash \varphi_1 \sqsubseteq \varphi_1 \times \varphi_2 \times \varphi : \text{Effect}}                
                    \label{eq:app27}
                \end{equation}

                \begin{equation}
                    \inferrule*[right=sb-$\times_{below}^{1}$]
                                    { 
                                        \inferrule*[right=sb-$\times_{below}^{2}$]
                                            {
                                                \texorpdfstring{\protect\hyperlink{app21}{\text{(*)}}}{}     
                                                \qquad
                                                \texorpdfstring{\protect\hyperlink{app24}{\text{(****)}}}{}     
                                            }
                                            {K \vdash \varphi'_2 \sqsubseteq \varphi_1 \times \varphi_2 : \text{Effect}}
                                        \qquad
                                            \texorpdfstring{\protect\hyperlink{app23}{\text{(***)}}}{}
                                    }
                                    {K \vdash \varphi'_2 \sqsubseteq \varphi_1 \times \varphi_2 \times \varphi : \text{Effect}}    
                    \label{eq:app28}
                \end{equation}

                \begin{equation}
                    \inferrule*[right=sb-$\times_{above}$]
                        { 
                               \eqref{eq:app27} 
                                \qquad  
                                \eqref{eq:app28} 
                            }
                            {K \vdash \varphi_1 \times \varphi'_2 \sqsubseteq \varphi_1 \times \varphi_2 \times \varphi : \text{Effect}}
                    \label{eq:app23}
                \end{equation}

                \begin{equation}                
                    \inferrule*[right=$\kappa$-$\times$]
                                    {   \texorpdfstring{\protect\hyperlink{app21}{\text{(*)}}}{} 
                                        \qquad 
                                        \texorpdfstring{\protect\hyperlink{app22}{\text{(**)}}}{} }
                                    {K \vdash \varphi_1 \times \varphi_2 : \text{Effect}}
                    \label{eq:app24}
                \end{equation}

                \begin{equation}
                    \inferrule*[right=sb-$\equiv$]
                                    { \inferrule*[right=eq-refl]
                                        { \texorpdfstring{\protect\hyperlink{app23}{\text{(***)}}}{} }
                                        {K \vdash \varphi \equiv \varphi : \text{Effect}}}
                                    {K \vdash \varphi \sqsubseteq \varphi : \text{Effect}}
                \label{eq:app25}
                \end{equation}

                \begin{equation}
                    \inferrule*[right=sb-$\times_{below}^{2}$]
                            {                                 
                            \eqref{eq:app24}
                                \qquad 
                            \eqref{eq:app25}                             
                            }
                            {K \vdash \varphi \sqsubseteq \varphi_1 \times \varphi_2 \times \varphi : \text{Effect}}
                    \label{eq:app26}
                \end{equation}

                \begin{align*}
                    \inferrule*[right=sb-$\times_{above}$]
                        { 
                         \eqref{eq:app23} 
                         \qquad
                         \eqref{eq:app26}                       
                          }
                        {K \vdash \varphi_1 \times \varphi'_2 \times \varphi \sqsubseteq \varphi_1 \times \varphi_2 \times \varphi : \text{Effect}}
                \end{align*}

                Finally, we can type the output store $\storeOuter'$ by induction:
                \begin{align*}
                    \Gamma \mid \Delta'_i \mid \Sigma' \vdash \storeOuter'
                \end{align*}

            \item Subcase (\textsc{e-appL}): \\
                From (\textsc{e-appL}) we have the following assumptions:
                \begin{enumerate}[a.]
                    \itemsep0em 
                    \item\label{ite:e-appL1} $(\storeInner, s_a) = \storeOuter(\rho)$
                    \item\label{ite:e-appL2} $(\lambda x . e_3) = \storeInner(l_\rho)$
                    \item\label{ite:e-appL3} $e' = [x \mapsto l'_{\rho'}]e_3$ 
                    \item\label{ite:e-appL4} $\storeOuter' = \storeOuter$
                \end{enumerate}
        
        Since the evaluation rule does not take additional steps $\Sigma'$, $\varphi'_1$, and
        $\varphi'_2$ are all equal to their original counterparts.

        Thus, using inversion (Lemma~\ref{lemma:typing-inversion}), we can deconstruct 
        the typing derivation for $\lambda x . e_3$, (the value at location 
        $l_\rho$, from assumptions~\ref{ite:e-appL1} and~\ref{ite:e-appL2}), yielding: 
        \begin{equation}
            K \mid \Gamma, x : \mu_1 \mid \Sigma \vdash e_3 : \mu_2 \mid  \varphi_1 \times \varphi_2 \times \varphi 
            \label{eq:appL1}
        \end{equation}
        Using this with substitution (Lemma~\ref{lemma:subst}), we can construct the following typing derivation for 
        $e'$: 
        \begin{align*}
            K \mid \Gamma \mid \Sigma \vdash [x \mapsto {l'}_{\rho'}] e_3 : \mu_2 \mid \varphi_1 \times \varphi_2 \times \varphi
        \end{align*}
        Since $\varphi_1$, $\varphi_2$ are unchanged, we trivially know that $\varphi' \sqsubseteq \varphi$.

        By assumption~\ref{ite:e-appL4}, we have that the output store is unchanged, thus we can type $\storeOuter'$ by appealing 
        to the typability of assumption~\ref{ite:t-app1}.

        \end{itemize}

    \item Case (\textsc{t-ref}):\\
        From (\textsc{t-ref}) we have the following assumptions:
        \begin{enumerate}[i]
            \itemsep0em 
            \item\label{ite:t-ref1} $e = \textbf{ref}\ e_1$
            \item\label{ite:t-ref2} $K \vdash \Gamma \mid \Sigma \vdash e_1 : (\tau, \rho) \mid \varphi_1 $
            \item\label{ite:t-ref3} $ \varphi = \varphi_1 \times \{\textbf{alloc}\ 1\ \rho \} $
        \end{enumerate}
        There are two evaluation rules by which $\langle e \mid \storeOuter
        \rangle \longrightarrow \langle e' \mid \storeOuter' \rangle$ can be derived:
        (\textsc{e-ref}) and (\textsc{e-refL}).
        \begin{itemize}

            \item Subcase (\textsc{e-ref}): \\
                From (\textsc{e-ref}) we have the following assumptions:
                \begin{enumerate}[a.]
                    \itemsep0em 
                    \item\label{ite:e-ref1} $\eformat{e_1}{\storeOuter}{e'_1}{\storeOuter'}$
                    \item\label{ite:e-ref2} $e' = \textbf{ref}\ e'_1$
                \end{enumerate}

                By induction on assumption~\ref{ite:t-ref2} we have: 
                \begin{equation}
                    K \mid \Gamma \mid \Sigma' \vdash e'_1 : (\tau, \rho) \mid \varphi'_1
                    \label{eq:ref1}
                \end{equation} 
                Let $\Sigma'$ equal any context such that $\Sigma' \supseteq
                \Sigma$ and the above typings hold.
                From this, we can construct the following typing derivation for $e'$:
                \begin{align*}
                    \inferrule*[right=t-ref]
                        {\eqref{eq:ref1}}
                        {K \mid \Gamma \mid \Sigma' \vdash \textbf{ref}\ e'_1 : (\text{Ref}\ \tau, \rho) \mid \varphi_1' \times \{ \textbf{alloc}\ 1\ \rho \} }
                \end{align*}

                From which we have that $\varphi' = \varphi'_1 \times \{\textbf{alloc}\ 1\ \rho \}$
                and from assumption~\ref{ite:t-ref3}, that $\varphi = \varphi'_1 \times \{ \textbf{alloc}\ 1\ \rho \} $. 
                By the typability of~\ref{ite:t-ref1}, we have that $K \vdash \varphi_1 : \text{Effect}$  \hypertarget{ref1}{(*)}, 
                and by induction we have that $\varphi'_1 \sqsubseteq \varphi_1$ \hypertarget{ref2}{(**)}. 
                From this we can construct the following derivation for $ \varphi'_1 \times \{\textbf{alloc}\ 1\ \rho\}\sqsubseteq \varphi_1 \times \{\textbf{alloc}\ 1\ \rho\}$:

                \begin{equation}
                        \inferrule*[right=$\kappa$-alloc]
                                        {\inferrule*[right=$\kappa$-reg]
                                            {\quad}
                                            {K \vdash \rho : \text{Region}}
                                        \qquad 
                                        \inferrule*[right=$\kappa$-size]
                                            {\quad}
                                            {K \vdash 1 : \text{Size} }}
                                        {K \vdash \{\textbf{alloc}\ 1\ \rho\} : \text{Effect} }        
                    \label{eq:ref4}
                \end{equation}

                \begin{equation}
                \inferrule*[right=sb-$\times_{below}^{1}$]
                                { 
                                    \texorpdfstring{\protect\hyperlink{ref2}{\text{(**)}}}{}
                                    \qquad 
                                    \eqref{eq:ref4} 
                                }
                                {K \vdash \varphi'_1 \sqsubseteq \varphi_1 \times \{\textbf{alloc}\ 1\ \rho\} : \text{Effect} }
                    \label{eq:ref2}
                \end{equation}

                \begin{equation}
                    \inferrule*[right=sb-$\equiv$] 
                                        {\inferrule*[right=eq-refl]
                                            { \inferrule*[right=$\kappa$-alloc]
                                                { \inferrule*[right=$\kappa$-reg]{\quad}{K \vdash \rho : \text{Region}} \qquad \inferrule*[right=$\kappa$-size]{\quad}{K \vdash 1 : \text{Size}}}
                                                { K \vdash \{\textbf{alloc}\ 1\ \rho \} : \text{Effect} } }
                                            {K \vdash \{ \textbf{alloc}\ 1\ \rho \} \equiv \{ \textbf{alloc}\ 1\ rho \} : \text{Effect} }}
                                        {K \vdash \{\textbf{alloc}\ 1\ \rho \} \sqsubseteq \{\textbf{alloc}\ 1\ \rho \} : \text{Effect}}
                    \label{eq:ref5}
                \end{equation}
                
                \begin{equation}
                    \inferrule*[right=sb-$\times_{below}^{2}$]
                                { 
                                    \texorpdfstring{\protect\hyperlink{ref1}{\text{(*)}}}{}
                                    \qquad 
                                    \eqref{eq:ref5}                                    
                                }
                                {K \vdash \{\textbf{alloc}\ 1\ \rho\}\sqsubseteq \varphi_1 \times \{\textbf{alloc}\ 1\ \rho\} : \text{Effect} }                
                    \label{eq:ref3}
                \end{equation}

                \begin{align*}
                    \inferrule*[right=sb-$\times_{above}$]
                        {
                            \eqref{eq:ref2} 
                            \qquad
                            \eqref{eq:ref3} 
                        }
                        {K \vdash  \varphi'_1 \times \{\textbf{alloc}\ 1\ \rho\} \sqsubseteq \varphi_1 \times \{\textbf{alloc}\ 1\ \rho\} : \text{Effect} }
                \end{align*}

                Finally, we can type the output store $\storeOuter'$ by induction:
                \begin{align*}
                    K \mid \Gamma \mid \Sigma' \vdash \storeOuter'
                \end{align*}

            \item Subcase (\textsc{e-refL}): \\
                From (\textsc{e-refL}) we have the following assumptions:
                \begin{enumerate}[a.]
                    \itemsep0em 
                    \item\label{ite:e-refL1} $(\storeInner_{\rho}^{in.}, s_a) = \storeOuter(\rho)$
                    \item\label{ite:e-refL2} $l'_{\rho} = \text{freshLoc}(\rho)$
                    \item\label{ite:e-refL3} $e' = l'_\rho$
                    \item\label{ite:e-refL4} $\storeOuter' = [\rho \mapsto ((\storeInner_{\rho}^{in.}, l'_{\rho} \mapsto l_\rho), s_a)] \storeOuter$
                \end{enumerate}

                From assumptions~\ref{ite:e-refL2} and~\ref{ite:e-refL4}, we have that $l'_{\rho}$ is a fresh location in
                $\rho$ that contains a reference to $l_\rho$. Therefore, let $\Sigma'$ equal $\Sigma$ with the
                additional location $l'_{\rho}$ such that $\text{Ref}\ \tau = \Sigma'(l'_{\rho})$ \hypertarget{refL1}{(*)} and 
                from the typability of assumption~\ref{ite:t-ref1} we have that $K \vdash \text{Ref}\ \tau : \text{Type}$ \hypertarget{refL2}{(**)} 
                
                From this we can construct the following typing derivation for $e'$:                
                \begin{align*}
                    \inferrule*[right=t-use-val]
                        {\inferrule*[right=t-loc]
                            { \texorpdfstring{\protect\hyperlink{refL1}{\text{(*)}}}{} \qquad \texorpdfstring{\protect\hyperlink{refL2}{\text{(**)}}}{}  }
                            {K \mid \Gamma \mid \Sigma' \vdash l'_{\rho} : \text{Ref}\ \tau }}
                        {K \mid \Gamma \mid \mid \Sigma' \vdash l'_{\rho} : (\text{Ref}\ \tau, \rho) \mid \{ \bot \} }
                \end{align*}

                From which we have that $\varphi' = \{ \bot \} $ and from assumption~\ref{ite:t-ref3} that $\varphi = \varphi_1 \times \{\textbf{alloc}\ 1\ \rho \}$. From the typability of~\ref{ite:t-ref1} we have that $K \vdash \varphi_1 : \text{Effect}$ \hypertarget{refL3}{(***)}
                which we can use to construct the following derivation for $\{ \bot \} \sqsubseteq \varphi_1 \times \{\textbf{alloc}\ 1\ \rho \}$: 

                \begin{equation} 
                        \inferrule*[right=$\kappa$-alloc]
                                    { \inferrule*[right=$\kappa$-reg]{\quad}{K \vdash \rho : \text{Region}} \qquad \inferrule*[right=$\kappa$-size]{\quad}{K \vdash 1 : \text{Size}} }
                                    {K \vdash \{\textbf{alloc}\ 1\ \rho \} : \text{Effect} }
                    \label{eq:refL1}
                \end{equation}

                \begin{align*}
                    \inferrule*[right=sb-$\bot$]
                        {\inferrule*[right=$\kappa$-$\times$]
                            {  \texorpdfstring{\protect\hyperlink{refL3}{\text{(***)}}}{} \qquad \eqref{eq:refL1}
                                }
                            {K \vdash \varphi_1 \times \{\textbf{alloc}\ 1\ \rho \} : \text{Effect}  }}
                        {K \vdash \{ \bot \} \sqsubseteq \varphi_1 \times \{\textbf{alloc}\ 1\ \rho \} }
                \end{align*}

                From our inductive hypothesis we have that: 
                \begin{align*}
                    K \mid \Gamma \mid \Sigma \vdash \storeOuter
                \end{align*}
                i.e., that the store prior to evaluation is well-typed~(\hypertarget{reflLStore1}{\text{\textdagger}}). 
                Therefore we have that currentSize($\rho$) $\sqsubseteq s_a$ prior to 
                evaluation~(\hypertarget{refLStore2}{\text{\textdagger\textdagger}}). From this we have:
                \begin{align*}   
                   K \mid \Gamma \mid \Sigma \vdash \storeInner_{\rho}^{in.}
                \end{align*} 
                i.e., that the store for region $\rho$ is well-typed prior to the evaluation of $e$~(\hypertarget{refLStore2}{\text{\textdagger\textdagger}})
                As per assumption~\ref{ite:e-refL4}, (\textsc{e-refL}) extends the region $\rho$ in the store with an additional location 
                $l'_{\rho}$ which points to location $l_{\rho}$.

                From the premise of our theorem statement, we know that $e$ will eventually be a part of a derivation 
                for which the typing of this allocation holds. Thus, using inversion (Lemma~\ref{lemma:typing-inversion}) on 
                the typing relation, along with the fact that sizeOf($l_{\rho}$) $= 1$ (as per Definition~\ref{def:sizeOf}) we have that $\exists s' . \text{currentSize}(\rho) + 1 + s' \sqsubseteq s_a$.

                From this and $\texorpdfstring{\protect\hyperlink{valLStore3}{(\text{\textdagger\textdagger\textdagger})}}{}$ we know that 
                currentSize($\rho$) + $s \sqsubseteq s_a$~(\hypertarget{refLStore4}{\text{\textdagger\textdagger\textdagger\textdagger}})
                , i.e., that the size of the region after the evaluation of $e$ does not exceed the maximum allocation bound for $\rho$.

                From the above we can construct the following typing derivation for the output store $\storeOuter'$: 

                \begin{equation}
                    \inferrule*[right=loc]
                        {
                            \texorpdfstring{\protect\hyperlink{refL1}{(*)}}{}
                            \qquad 
                            (\texorpdfstring{\protect\hyperlink{refLStore2}{\text{\textdagger\textdagger}}}{})
                            \qquad
                            \ref{ite:t-ref2}
                        }
                        {K \mid \Gamma \mid \Sigma \vdash ((\storeInner_{\rho}^{in.}, l'_{\rho} \mapsto l_{\rho}), s_a) }
                        \label{eq:refLStore1}
                \end{equation}

                \begin{align*}
                    \inferrule*[right=subst]
                        { 
                            \texorpdfstring{\protect\hyperlink{refLStore1}{(\text{\textdagger})}}{}
                            \qquad
                            \eqref{eq:refLStore1}
                            \qquad
                            \texorpdfstring{\protect\hyperlink{refLStore3}{(\text{\textdagger\textdagger\textdagger})}}{}
                        }
                        {K \mid \Gamma \mid \Sigma \vdash [\rho \mapsto ((\storeInner_{\rho}^{in.}, l'_{\rho} \mapsto l_{\rho}), s_a)] \storeOuter }
                \end{align*}

        \end{itemize}

    \item Case (\textsc{t-deref}):\\
        From (\textsc{t-deref}) we have the following assumptions:
        \begin{enumerate}[i]
            \itemsep0em 
            \item\label{ite:t-deref1} $e =\ !e_1$
            \item\label{ite:t-deref2} $K \mid \Gamma \mid \Sigma \vdash e : (\text{Ref}\ \tau, \rho) \mid \varphi_1 $
            \item\label{ite:t-deref3} $\varphi = \varphi_1 $
        \end{enumerate}
        There are two evaluation rules by which $\langle e \mid \storeOuter
        \rangle \longrightarrow \langle e' \mid \storeOuter' \rangle$ can be derived:
        (\textsc{e-deref}) and (\textsc{e-derefL}).
        \begin{itemize}

            \item Subcase (\textsc{e-deref}): \\
                From (\textsc{e-deref}) we have the following assumptions:
                \begin{enumerate}[a.]
                    \itemsep0em 
                    \item\label{ite:e-deref1} $\eformat{e_1}{\storeOuter}{e'_1}{\storeOuter'}$
                    \item\label{ite:e-deref2} $e' =\ !e'_1$
                \end{enumerate}
                By induction on assumption~\ref{ite:t-deref2} we have: 
                \begin{equation}
                    K \mid \Gamma \mid \Sigma' \vdash e'_1 : (\text{Ref}\ \tau, \rho) \mid \varphi'_1 
                    \label{eq:deref1}
                \end{equation}
                Let $\Sigma'$ equal any context such that $\Sigma' \supseteq
                \Sigma$ and the above typing holds.
                From this we can construct the following typing derivation for $e'$: 
                \begin{align*}
                    \inferrule*[right=t-deref]
                        {\eqref{eq:deref1}}
                        {K \mid \Gamma \mid  \Sigma' \vdash\ !e'_1 : (\tau, \rho) \mid \varphi_1'}
                \end{align*}
                For this we have that $\varphi' = \varphi'_1$ and from assumption~\ref{ite:t-deref3} 
                that $\varphi = \varphi$. Thus, by induction we have that $\varphi'_1 \sqsubseteq \varphi_1$. 

                Finally, we can type the output store $\storeOuter'$ by induction:
                \begin{align*}
                    \Gamma \mid \Delta'_i \mid \Sigma' \vdash \storeOuter'
                \end{align*}

            \item Subcase (\textsc{e-derefL}): \\
                From (\textsc{e-derefL}) we have the following assumptions:
                \begin{enumerate}[a.]
                    \itemsep0em 
                        \item\label{ite:e-derefL1} $(\storeInner_{\rho}^{in.}, s_a) = \storeOuter(\rho)$
                        \item\label{ite:e-derefL2} $l'_{\rho} = \storeInner_{\rho}^{in.}(l_\rho)$
                        \item\label{ite:e-derefL3} $e' = l'_{\rho}$
                        \item\label{ite:e-derefL4} $\storeOuter' = \storeOuter$
                \end{enumerate}

                From assumptions~\ref{ite:e-derefL1} and~\ref{ite:e-derefL2}, we have that $l'_{\rho}$ is a location in $\rho$ that
                contains a reference to $l_\rho$. Therefore, let $\Sigma'$ equal $\Sigma$ such that
                $\text{Ref}\ \tau  = \Sigma'(l'_{\rho})$ \hypertarget{derefL1}{(*)} and from the typability of 
                assumption~\ref{ite:t-deref2} that $K \vdash \text{Ref}\ \tau : \text{Type}$ \hypertarget{derefL2}{(**)}

                From this we can construct the following typing derivation for $e'$:

                \begin{align*}
                    \inferrule*[right=t-use-val]
                        { \inferrule*[right=t-loc]
                            { \texorpdfstring{\protect\hyperlink{derefL1}{\text{(*)}}}{}  \qquad  \texorpdfstring{\protect\hyperlink{derefL2}{\text{(**)}}}{} }
                            {K \mid \Gamma \mid \Sigma' \vdash l'_{\rho} : \tau } }
                        {K \mid \Gamma \mid \Sigma' \vdash l'_{\rho} : (\tau, \rho) \mid \{ \bot \} }
                \end{align*}
                From this we have that $\varphi' = \{ \bot \} $ and from assumption~\ref{ite:t-deref3} that $\varphi = \varphi_1$. Furthermore,
                from the typability of assumption~\ref{ite:t-deref2} that $K \vdash \varphi_1 : \text{Effect}$ \hypertarget{derefL3}{(***)}. Thus we can construct the 
                following derivation for $\{ \bot \} \sqsubseteq \varphi_1$:
                \begin{align*}
                    \inferrule*[right=sb-$\bot$]
                        {\texorpdfstring{\protect\hyperlink{derefL3}{\text{(***)}}}{}}
                        {K \vdash \{ \bot \} \sqsubseteq \varphi_1 : \text{Effect} } 
                \end{align*}

                Finally, since we take no inductive step the store remains unchanged, which we can 
                therefore type:
                \begin{align*}
                    K \mid \Gamma \mid \Sigma' \vdash \storeOuter'
                \end{align*}

        \end{itemize}

    \item Case (\textsc{t-assign}):\\
        From (\textsc{t-assign}) we have the following assumptions:
        \begin{enumerate}[i]
            \itemsep0em 
            \item\label{ite:t-assign1} $e = e_1 := e_2$
            \item\label{ite:t-assign2} $K \mid \Gamma \mid \Sigma \vdash e_1 : (\text{Ref}\ \tau, \rho') \mid \varphi_1 $
            \item\label{ite:t-assign3} $K \mid \Gamma \mid \Sigma \vdash e_2 : \mu \mid \varphi_2 $
            \item\label{ite:t-assign4} $\varphi = \varphi_1 \times \varphi_2$
        \end{enumerate}
        There are three evaluation rules by which $\langle e \mid \storeOuter
        \rangle \longrightarrow \langle e' \mid \storeOuter' \rangle$ can be derived:
        (\textsc{e-assign1}), (\textsc{e-assign2}) and (\textsc{e-assignL}).
        \begin{itemize}

            \item Subcase (\textsc{e-assign1}): \\
                From (\textsc{e-assign1}) we have the following assumptions:
                \begin{enumerate}[a.]
                    \itemsep0em 
                    \item\label{ite:e-assign11} $\eformat{e_1}{\storeOuter}{e'_1}{\storeOuter'}$
                    \item\label{ite:e-assign12} $e' = e'_1 := e_2$
                \end{enumerate}
                By induction on assumption~\ref{ite:t-assign2} we have: 
                \begin{equation}
                    K \mid \Gamma \mid \Sigma' \vdash e'_1 : (\text{Ref}\ \tau, \rho') \mid \varphi'_1 
                    \label{eq:assign11}
                \end{equation}
                which we can use with assumption~\ref{ite:t-assign3} to obtain: 
                \begin{equation}
                    K \mid \Gamma \mid \Sigma' \vdash e_2 : \mu \mid \varphi_2 
                    \label{eq:assign12}
                \end{equation}
                Let $\Sigma'$ equal any context such that $\Sigma' \supseteq
                \Sigma$ and the above typings hold.
                From this we can construct the following typing derivation for $e'$:
                \begin{align*}
                    \inferrule*[right=t-assign]
                        {\eqref{eq:assign11} \qquad \eqref{eq:assign12}}
                        {K \mid \Gamma \mid \Sigma' \vdash e'_1 := e_2 : \text{Unit} \mid \varphi'_1 \times \varphi_2}
                \end{align*}
                From which we have that $\varphi' = \varphi'_1 \times \varphi_2$ and from assumption~\ref{ite:t-assign4} 
                that $\varphi = \varphi_1 \times \varphi_2$. From the typability of assumption~\ref{ite:t-assign2} we have 
                that $K \vdash \varphi_1 : \text{Effect}$ \hypertarget{assign11}{(*)}, from the typability of 
                assumption~\ref{ite:t-assign3} we have that $K \vdash \varphi_2 : \text{Effect}$ \hypertarget{assign12}{(**)}, 
                and by induction that $\varphi'_1 \sqsubseteq \varphi_1$ \hypertarget{assign13}{(***)}. From this 
                we can construct the following derivation for $\varphi'_1 \times \varphi_2 \sqsubseteq \varphi_1 \times \varphi_2$:

                \begin{equation}
                        \inferrule*[right=sb-$\times_{below}^{2}$]
                                { 
                                    \texorpdfstring{\protect\hyperlink{assign11}{\text{(*)}}}{}
                                    \qquad 
                                    \inferrule*[right=sb-$\equiv$]
                                        {\inferrule*[right=eq-refl]
                                            {\texorpdfstring{\protect\hyperlink{assign12}{\text{(**)}}}{}}
                                            {K \vdash \varphi_2 \equiv \varphi_2 : \text{Effect}}}
                                        {K \vdash \varphi_2 \sqsubseteq \varphi_2 : \text{Effect}}
                                }
                                {K \vdash \varphi_2 \sqsubseteq \varphi_1 \times \varphi_2 : \text{Effect}}
                    \label{eq:assign13}
                \end{equation}

                \begin{align*}
                    \inferrule*[right=sb-$\times_{above}$]
                        { 
                            \inferrule*[right=$\times_{below}^{1}$]   
                                { 
                                    \texorpdfstring{\protect\hyperlink{assign13}{\text{(***)}}}{}
                                    \qquad
                                    \texorpdfstring{\protect\hyperlink{assign12}{\text{(**)}}}{}
                                }
                                {K \vdash \varphi'_1 \sqsubseteq \varphi_1 \times \varphi_2 : \text{Effect}}
                            \qquad
                            \eqref{eq:assign13} 
                        }
                        {K \vdash \varphi'_1 \times \varphi_2 \sqsubseteq \varphi_1 \times \varphi_2 : \text{Effect}} 
                \end{align*}

                Finally, we can type the output store $\storeOuter'$ by induction:
                \begin{align*}
                    K \mid \Gamma \mid \Sigma' \vdash \storeOuter' 
                \end{align*}

            \item Subcase (\textsc{e-assign2}): \\
                From (\textsc{e-assign2}) we have the following assumptions:
                \begin{enumerate}[a.]
                    \itemsep0em 
                    \item\label{ite:e-assign21} $\eformat{e_2}{\storeOuter}{e'_2}{\storeOuter'}$
                    \item\label{ite:e-assign22} $e' = l_\rho := e'_2$
                \end{enumerate}
                By induction on assumption~\ref{ite:t-assign3} we have: 
                \begin{equation}
                    K \mid \Gamma \mid \Sigma' \vdash e'_2 : \mu \mid \varphi'_2 
                    \label{eq:assign21}
                \end{equation}
                which we can use with assumption~\ref{ite:t-assign3} to obtain: 
                \begin{equation}
                    K \mid \Gamma \mid \Sigma' \vdash l_\rho : (\text{Ref}\ \tau, \rho') \mid \varphi_1 
                    \label{eq:assign22}
                \end{equation}
                Let $\Sigma'$ equal any context such that $\Sigma' \supseteq
                \Sigma$ and the above typings hold.
                From this we can construct the following typing derivation for $e'$:
                \begin{align*}
                    \inferrule*[right=t-assign]
                        {\eqref{eq:assign21} \qquad \eqref{eq:assign22}}
                        {K \mid \Gamma \mid \Sigma' \vdash l_\rho := e'_2 : \text{Unit} \mid \varphi_1 \times \varphi'_2}
                \end{align*}
                From which we have that $\varphi' = \varphi_1 \times \varphi'_2$ and from assumption~\ref{ite:t-assign4} 
                that $\varphi = \varphi_1 \times \varphi_2$. From the typability of assumption~\ref{ite:t-assign2} we have 
                that $K \vdash \varphi_1 : \text{Effect}$ \hypertarget{assign21}{(*)}, from the typability of 
                assumption~\ref{ite:t-assign3} we have that $K \vdash \varphi_2 : \text{Effect}$ \hypertarget{assign22}{(**)}, 
                and by induction that $\varphi'_2 \sqsubseteq \varphi_2$ \hypertarget{assign23}{(***)}. From this 
                we can construct the following derivation for $\varphi_1 \times \varphi'_2 \sqsubseteq \varphi_1 \times \varphi_2$:

                \begin{equation}
                        \inferrule*[right=sb-$\times_{below}^{1}$]
                                { 
                                    \inferrule*[right=sb-$\equiv$]
                                        {\inferrule*[right=eq-refl]
                                            {\texorpdfstring{\protect\hyperlink{assign21}{\text{(*)}}}{}}
                                            {K \vdash \varphi_1 \equiv \varphi_1 : \text{Effect}}}
                                        {K \vdash \varphi_1 \sqsubseteq \varphi_1 : \text{Effect}}
                                    \qquad 
                                    \texorpdfstring{\protect\hyperlink{assign22}{\text{(**)}}}{}
                                }
                                {K \vdash \varphi_1 \sqsubseteq \varphi_1 \times \varphi_2 : \text{Effect}}
                    \label{eq:assign23}
                \end{equation}

                \begin{align*}
                    \inferrule*[right=sb-$\times_{above}$]
                        { 
                            \eqref{eq:assign13} 
                            \qquad
                            \inferrule*[right=$\times_{below}^{2}$]   
                                { 
                                    \texorpdfstring{\protect\hyperlink{assign21}{\text{(*)}}}{}
                                    \qquad
                                    \texorpdfstring{\protect\hyperlink{assign23}{\text{(***)}}}{}
                                }
                                {K \vdash \varphi'_2 \sqsubseteq \varphi_1 \times \varphi_2 : \text{Effect}}
                        }
                        {K \vdash \varphi_1 \times \varphi'_2 \sqsubseteq \varphi_1 \times \varphi_2 : \text{Effect}} 
                \end{align*}

                Finally, we can type the output store $\storeOuter'$ by induction:
                \begin{align*}
                    K \mid \Gamma \mid \Sigma' \vdash \storeOuter' 
                \end{align*}

            \item Subcase (\textsc{e-assignL}): \\
                From (\textsc{e-assignL}) we have the following assumptions:
                \begin{enumerate}[a.]
                    \itemsep0em 
                    \item\label{ite:e-assignL1} $(\storeInner_{\rho}^{in.}, s_a) = \storeOuter(\rho)$
                    \item\label{ite:e-assignL2} $e' = l^{1}_{\globalRegion}$
                    \item\label{ite:e-assignL3} $\storeOuter' = [\rho \mapsto ([l_\rho \mapsto l'_{\rho'}] \storeInner_{\rho}^{in.}, s_a)] \storeOuter$
                \end{enumerate}
                From the premises of our theorem, we have that $\text{Unit} = \Sigma'(l^{1}_{\globalRegion})$ \hypertarget{assignL1}{(*)}. From this we can construct the following typing derivation for $e'$:
                \begin{align*}
                    \inferrule*[right=t-use-val]
                        { \inferrule*[right=t-loc]{ \texorpdfstring{\protect\hyperlink{assignL1}{\text{(*)}}}{} \qquad \inferrule*[right=$\kappa$-unit]{\quad}{K \vdash \text{Unit} : \text{Type}} }{K \mid \Gamma \mid \Sigma \vdash l^{1}_{\globalRegion} : \text{Unit} }}
                        {K \mid \Gamma \mid \Sigma' \vdash l^{1}_{\globalRegion} : (\text{Unit}, \globalRegion) \mid \{ \bot \} }
                \end{align*}

                From this we have that $\varphi' = \{ \bot \} $ and from assumption~\ref{ite:t-assign4} that $\varphi = \varphi_1 \times \varphi_2$. Furthermore,
                from the typability of assumption~\ref{ite:t-assign2} that $K \vdash \varphi_1 : \text{Effect}$ \hypertarget{assignL2}{(**)}, and 
                from the typability of assumption~\ref{ite:t-assign3} that $K \vdash \varphi_2 : \text{Effect}$ \hypertarget{assignL3}{(***)}. Thus we can construct the following derivation for $\{ \bot \} \sqsubseteq \varphi_1 \times \varphi_2$:
                \begin{align*}
                    \inferrule*[right=sb-$\bot$]
                        {\inferrule*[right=$\kappa$-$\times$]
                            { \texorpdfstring{\protect\hyperlink{assignL2}{\text{(**)}}}{} \qquad \texorpdfstring{\protect\hyperlink{assignL3}{\text{(***)}}}{}} {K \vdash \varphi_1 \times \varphi_2 : \text{Effect}} }
                        {K \vdash \{ \bot \} \sqsubseteq \varphi_1 : \text{Effect} } 
                \end{align*}

                Finally, we can use the store update lemma (Lemma~\ref{lemma:store-update}) to type $\storeOuter'$:
                \begin{align*}                    
                    K \mid \Gamma \mid \Sigma' \vdash \storeOuter, \rho \mapsto ([l_\rho \mapsto l'_{\rho'}] \storeInner_{\rho}^{in.}, s_a)
                \end{align*}

        \end{itemize}

    \item Case (\textsc{t-seq}):\\
        From (\textsc{t-seq}) we have the following assumptions:
        \begin{enumerate}[i]
            \itemsep0em 
            \item\label{ite:t-seq1} $e = e_1 ; e_2$
            \item\label{ite:t-seq2} $K \mid \Gamma \mid \Sigma \vdash e_1 : (\text{Unit}, \rho) \mid \varphi_1 $
            \item\label{ite:t-seq3} $K \mid \Gamma \mid \Sigma \vdash e_2 : \mu \mid \varphi_2 $
            \item\label{ite:t-seq4} $\varphi = \varphi_1 \times \varphi_2$
        \end{enumerate}
        There are two evaluation rules by which $\langle e \mid \storeOuter
        \rangle \longrightarrow \langle e' \mid \storeOuter' \rangle$ can be derived:
        (\textsc{e-seq}) and (\textsc{e-seqNext}).
        \begin{itemize}

            \item Subcase (\textsc{e-seq}): \\
                From (\textsc{e-seq}) we have the following assumptions:
                \begin{enumerate}[a.]
                    \itemsep0em 
                    \item\label{ite:e-seq1} $\eformat{e_1}{\storeOuter}{e'_1}{\storeOuter'}$
                    \item\label{ite:e-seq2} $e' = e'_1 ; e_2$
                \end{enumerate}
                By induction on assumption~\ref{ite:t-seq2} we have:
                \begin{equation}
                    K \mid \Gamma \mid \Sigma' \vdash e'_1 : (\text{Unit}, \rho) \mid \varphi'_1 
                    \label{eq:seq-1}
                \end{equation}
                which we can use with assumption~\ref{ite:t-seq3} to obtain:
                \begin{equation}
                    K \mid \Gamma \mid \Sigma' \vdash e_2 : \mu \mid \varphi_2 
                    \label{eq:seq-2}
                \end{equation}
                Let $\Sigma'$ equal any context such that $\Sigma' \supseteq
                \Sigma$ and the above typings hold.
                From the above we can construct the following typing derivation for $e'$:   
                \begin{align*}
                    \inferrule*[right=seq]
                        { \eqref{eq:seq-1} \qquad \eqref{eq:seq-2} }
                        {K \mid \Gamma \mid \Sigma' \vdash e'_1 ; e_2 : \mu \mid \varphi'_1 \times \varphi_2 }
                \end{align*}
                From which we have that $\varphi' = \varphi'_1 \times \varphi_2$ and from assumption~\ref{ite:t-seq4}
                we have that $\varphi = \varphi_1 \times \varphi_2$. From the typability of assumption~\ref{ite:t-seq2} we have that 
                $K \vdash \varphi_1 : \text{Effect}$, from the typability of assumption~\ref{ite:t-seq3} we have that 
                $K \vdash \varphi_2 : \text{Effect}$, and by induction we have that $\varphi'_1 \sqsubseteq \varphi_1$. From this 
                we can construct the following derivation for $\varphi'_1 \times \varphi_2 \sqsubseteq \varphi_1 \times \varphi_2$:

                \begin{equation}
                    \inferrule*[right=sb-$\times_{below}^{}$]
                                { 
                                    \texorpdfstring{\protect\hyperlink{seq1}{\text{(*)}}}{} 
                                    \qquad
                                    \inferrule*[right=sb-$\equiv$]
                                        {\inferrule*[right=eq-refl]
                                            { \texorpdfstring{\protect\hyperlink{seq2}{\text{(**)}}}{}  }
                                            { K \vdash \varphi_2 \equiv \varphi_2 : \text{Effect} }}
                                        {K \vdash \varphi_2 \sqsubseteq \varphi_2 : \text{Effect}}
                                }
                                {K \vdash \varphi_2 \sqsubseteq \varphi_1 \times \varphi_2 : \text{Effect}}
                    \label{eq:seq-3}
                \end{equation}

                \begin{align*}
                    \inferrule*[right=sb-$\times_{above}$]
                        { 
                            \inferrule*[right=sb-$\times_{below}^{1}$]
                                { 
                                    \texorpdfstring{\protect\hyperlink{seq3}{\text{(***)}}}{}
                                    \qquad
                                    \texorpdfstring{\protect\hyperlink{seq2}{\text{(**)}}}{}
                                }
                                {K \vdash \varphi'_1 \sqsubseteq \varphi_1 \times \varphi_2 : \text{Effect}}
                            \qquad
                            \eqref{eq:seq-3} 
                        }
                        {K \vdash \varphi'_1 \times \varphi_2 \sqsubseteq \varphi_1 \times \varphi_2 : \text{Effect} }
                \end{align*}

                Finally, we can type the output store $\storeOuter'$ by induction:
                \begin{align*}
                    K \mid \Gamma \mid \Sigma' \vdash \storeOuter'
                \end{align*}

            \item Subcase (\textsc{e-seqNext}): \\
                From (\textsc{e-seqNext}) we have the following assumptions:
                \begin{enumerate}[a.]
                    \itemsep0em 
                    \item\label{ite:e-seqN1} $e' = e_2$
                    \item\label{ite:e-seqN2} $\storeOuter' = \storeOuter$
                \end{enumerate}
                Since the evaluation rule does not take inductive evaluation steps, 
                $\Sigma'$, $\varphi'_1$, and $\varphi'_2$ are all equal to their original counterparts.
                From this and the typability of assumption~\ref{ite:t-seq3} we can construct the following type for $e'$: 
                \begin{align*} 
                    K \mid \Gamma \mid \Sigma' \vdash e_2 : \mu \mid \varphi'_2
                \end{align*}
                From which we have that $\varphi' = \varphi'_2$ and from assumption~\ref{ite:t-seq4} that 
                $\varphi = \varphi_1 \times \varphi_2$. Since $\varphi_1$, $\varphi_2$ are unchanged, we trivially know that $\varphi' \sqsubseteq \varphi$.

                We can also type the output store $\storeOuter'$ by our assumption of store typing 
                consistency (as the store and contexts are unchanged):
                \begin{align*}
                    K \mid \Gamma \mid \Sigma' \vdash \storeOuter 
                \end{align*}
        \end{itemize}

    \item Case (\textsc{t-tyApp}):\\
        From (\textsc{t-tyApp}) we have the following assumptions:
        \begin{enumerate}[i]
            \itemsep0em 
            \item\label{ite:t-bigLam1} $e = e_1\ @\ (\tau, \rho')$
            \item\label{ite:t-bigLam2} $K \mid \Gamma \mid \Sigma \vdash e_1 : \forall\{\alpha : \text{Type}, \rho : \text{Region} \}  \mid \{ \bot \}  $
            \item\label{ite:t-bigLam3} $K \vdash (\tau, \rho') : \text{Type}$
            \item\label{ite:t-bigLam4} $\varphi = \{ \bot \}$
        \end{enumerate}
        There is one evaluation rule by which $\langle e \mid \storeOuter
        \rangle \longrightarrow \langle e' \mid \storeOuter' \rangle$ can be derived:
        (\textsc{e-tyApp}).
        \begin{itemize}
            \item Subcase (\textsc{e-tyApp}): \\
                From (\textsc{e-tyApp}) we have the following assumptions:
                \begin{enumerate}[a.]
                    \itemsep0em 
                    \item\label{ite:e-tyApp1} $e' = e_1$
                    \item\label{ite:e-tyApp2} $\storeOuter' = \storeOuter$
                \end{enumerate}

                Since the evaluation rule does not take any additional steps, $\Sigma'$, and $\varphi'$ are 
                all equal to their original counterparts. Thus, the evaluation simply steps to $e_1$, which is
                well-typed by assumption~\ref{ite:t-bigLam2}:
                 \begin{align*}
                        K \mid \Gamma \mid \Sigma \vdash e_1 : \forall\{\alpha : \text{Type}, \rho : \text{Region} \}  \mid \{ \bot \}
                \end{align*}
                Since $\Sigma' = \Sigma, \varphi' = \varphi$, we have trivially that
                $\{\bot\} \sqsubseteq \{\bot\}$:
                \begin{align*}
                    \inferrule*[right=sb-$\bot$]
                        {\inferrule*[right=$\kappa$-$\bot$]
                            {\quad}
                            {K \vdash \{ \bot \} : \text{Effect}}}
                        {K \vdash \{\bot\} \sqsubseteq \{ \bot \} : \text{Effect}}
                \end{align*}

                By assumption~\ref{ite:e-tyApp2}, we have that the output store is unchanged. Thus 
                we can type $\storeOuter'$ by appealing to the typability of assumption~\ref{ite:t-bigLam1}.

        \end{itemize}

    \item Case (\textsc{t-let})\\
        From (\textsc{t-let}) we have the following assumptions:
        \begin{enumerate}[i]
            \itemsep0em 
            \item\label{ite:t-let1} $e = \textbf{let}\ x : \mu_1 = e_1\ \textbf{in}\ e_2$
            \item\label{ite:t-let2} $K \mid \Gamma \mid \Sigma \vdash e_1 : \mu_1 \mid \varphi_1$
            \item\label{ite:t-let3} $K \mid \Gamma, x : \mu_1 \mid \Sigma \vdash e_2 : \mu_2 \mid \varphi_2$
            \item\label{ite:t-let4} $\varphi = \varphi_1 \times \varphi_2$
        \end{enumerate}
        There are two evaluation rules by which $\langle e \mid \storeOuter
        \rangle \longrightarrow \langle e' \mid \storeOuter' \rangle$ can be derived:
        (\textsc{e-let}) and (\textsc{e-letL}).
        \begin{itemize}
            \item Subcase (\textsc{e-let}): \\
                From (\textsc{e-let}) we have the following assumptions:
                \begin{enumerate}[a.]
                    \itemsep0em 
                    \item\label{ite:e-let1} $\eformat{e_1}{\storeOuter}{e'_1}{\storeOuter'}$
                    \item\label{ite:e-let2} $e' = \textbf{let}\ x : \mu_1 = e'_1\ \textbf{in}\ e_2$
                \end{enumerate}

                By induction on assumption~\ref{ite:t-let2} we have: 
                \begin{equation}
                    K \mid \Gamma \mid \Sigma \vdash e'_1 : \mu_1 \mid \varphi'_1
                    \label{eq:let1}
                \end{equation} 
                which we can use with assumption~\ref{ite:t-let3} to obtain: 
                \begin{equation}             
                   K \mid \Gamma, x : \mu_1 \mid \Sigma' \vdash e_2 : \mu_2 \mid \varphi_2
                   \label{eq:let2}
                \end{equation}
                Let $\Sigma'$ equal any context such that $\Sigma' \supseteq
                \Sigma$ and the above typings hold.
                From this, we can construct the following typing derivation for $e'$:
                \begin{align*}
                    \inferrule*[right=t-let]
                        {\eqref{eq:let1} \qquad \eqref{eq:let2}}
                        {K \mid \Gamma \mid \Sigma' \vdash \textbf{let}\ x : \mu_1 = e'_1\ \textbf{in}\ e_2 \mid \varphi_1' \times \varphi_2 }
                \end{align*}

                From which we have that $\varphi' = \varphi'_1 \times \varphi_2$
                and from assumption~\ref{ite:t-let4}, that $\varphi = \varphi_1 \times \varphi_2$. 
                By the typability of~\ref{ite:t-let2}, we have that $K \vdash \varphi_1 : \text{Effect}$~\hypertarget{let1}{(*)}, 
                and by induction we have that $\varphi'_1 \sqsubseteq \varphi_1$ \hypertarget{let2}{(**)}. 
                Furthermore, from the typability of~\ref{ite:t-let3}, we have that $K \vdash \varphi_2 : \text{Effect}$~\hypertarget{let3}{(***)}.
                From this we can construct the following derivation for $ \varphi'_1 \times \varphi_2 \sqsubseteq \varphi_1 \times \varphi_2$:

                \begin{equation}
                    \inferrule*[right=sb-$\times_{below}^{1}$]
                        {
                            \texorpdfstring{\protect\hyperlink{let2}{\text{(**)}}}{}
                            \qquad
                            \texorpdfstring{\protect\hyperlink{let3}{\text{(***)}}}{}
                        }
                        {K \vdash \varphi'_1 \sqsubseteq \varphi_1 \times \varphi_2 : \text{Effect}}
                        \label{eq:let3}
                \end{equation}

                \begin{equation}
                    \inferrule*[right=sb-$\times_{below}^{2}$]
                        {
                            \texorpdfstring{\protect\hyperlink{let1}{\text{(*)}}}{}
                            \qquad
                            \inferrule*[right=sb-$\equiv$]
                                {\inferrule*[right=eq-refl]
                                    {
                                        \texorpdfstring{\protect\hyperlink{let3}{\text{(***)}}}{}
                                    }
                                    {K \vdash \varphi_2 \equiv \varphi_2 : \text{Effect}}}
                                {K \vdash \varphi_2 \sqsubseteq \varphi_2 : \text{Effect}}
                        }
                        {K \vdash \varphi_2 \sqsubseteq \varphi_1 \times \varphi_2 : \text{Effect}}
                        \label{eq:let4}
                \end{equation}

                \begin{align*}
                    \inferrule*[right=sb-$\times_{above}$]
                        {
                            \eqref{eq:let3} 
                            \qquad
                            \eqref{eq:let4} 
                        }
                        {K \vdash  \varphi'_1 \times \varphi_2 \sqsubseteq \varphi_1 \times \varphi_2 : \text{Effect} }
                \end{align*}

                Finally, we can type the output store $\storeOuter'$ by induction:
                \begin{align*}
                    K \mid \Gamma \mid \Sigma' \vdash \storeOuter'
                \end{align*}

            \item Subcase (\textsc{e-letL}): \\
                From (\textsc{e-letL}) we have the following assumptions:
                \begin{enumerate}[a.]
                    \itemsep0em 
                    \item\label{ite:e-letL1} $e' = [x \mapsto l_\rho] e_2$
                    \item\label{ite:e-letL2} $\storeOuter' = \storeOuter$
                \end{enumerate}

                Since the evaluation rule does not take any additional steps, $\Sigma', \varphi'_1$, and $\varphi'_2$ are 
                all equal to their original counterparts. Thus, using substitution (Lemma~\ref{lemma:subst}) with 
                assumption~\ref{ite:t-let3}, we obtain:
                \begin{align*}
                    K \mid \Gamma \mid \Sigma' \vdash [x \mapsto l_\rho] e_2 : \mu_2 \mid \varphi'_1 \times \varphi'_2
                \end{align*}
                Since $\Sigma' = \Sigma, \varphi'_1 = \varphi_1$, and $\varphi'_2 = \varphi_2$, we have trivially that
                $\varphi'_1 \times \varphi'_2 \sqsubseteq \varphi_1 \times \varphi_2$.

                By assumption~\ref{ite:e-let2}, we have that the output store is unchanged. Thus 
                we can type $\storeOuter'$ by appealing to the typability of assumption~\ref{ite:t-let1}.

            \end{itemize}

    \item Case (\textsc{t-fix})\\
        From (\textsc{t-fix}) we have the following assumptions: 
        \begin{enumerate}[i]
            \itemsep0em 
            \item\label{ite:t-fix1} $e = \textbf{let}\ f = \textbf{fix}(f, (\Lambda \{ \alpha , \rho , \epsilon \} . \lambda x . e_1)\ [s]\ \textbf{at}\ e_2)\ \textbf{in}\ e_3$
            \item\label{ite:t-fix2} $K \mid \Gamma, f : (\forall \{ \alpha, \rho, \epsilon  \} . (\alpha, \rho) \xrightarrow{\{ \epsilon \}} \mu_1, \rho_f) \mid \Sigma \vdash (\Lambda\{ \alpha, \rho, \epsilon \}. \lambda x . e_1)\ [s]\ \textbf{at}\ e_2 : (\forall \{ \alpha , \rho, \epsilon  \} . (\alpha, \rho) \xrightarrow{\varphi} \mu_1, \rho_f)  \mid \varphi_1$
            \item\label{ite:t-fix3} $\varphi' = [\{ \epsilon \} \mapsto \{ \textbf{rec}\ \epsilon\ \varphi \}] \varphi$ 
            \item\label{ite:t-fix4} $K \mid \Gamma, f : (\forall \{ \alpha , \rho , \epsilon  \} . (\alpha, \rho) \xrightarrow{\varphi'} \mu_1, \rho_f) \mid \Sigma \vdash e_3 : \mu_2 \mid \varphi_2$ 
            \item\label{ite:t-fix5} $\varphi = \varphi_1 \times \varphi_2$
        \end{enumerate}
        There are two evaluation rules by which $\langle e \mid \storeOuter
        \rangle \longrightarrow \langle e' \mid \storeOuter' \rangle$ can be derived:
        (\textsc{e-fix}) and (\textsc{e-fixL}).
        \begin{itemize}
            \item Subcase (\textsc{e-fix}): \\
                From (\textsc{e-fix}) we have the following assumptions:
                \begin{enumerate}[a.]
                    \itemsep0em 
                    \item\label{ite:e-fix1} $\eformat{\Lambda \{\alpha, \rho, \epsilon\} . \lambda x . e_1\ [s]\ \textbf{at}\ e_2}{\storeOuter}{e'}{\storeOuter'}$
                    \item\label{ite:e-fix2} $e' = \textbf{let}\ f = (f, e')\ \textbf{in}\ e_3$
                \end{enumerate}

                By induction on assumption~\ref{ite:t-fix1} we have: 
                \begin{equation}
                    K \mid \Gamma, f : (\forall \{ \alpha, \rho, \epsilon  \} . (\alpha, \rho) \xrightarrow{\{ \epsilon \}} \mu_1, \rho_f) \mid \Sigma \vdash (\Lambda\{ \alpha, \rho, \epsilon \}. \lambda x . e_1)\ [s]\ \textbf{at}\ e'_2 : (\forall \{ \alpha , \rho, \epsilon  \} . (\alpha, \rho) \xrightarrow{\varphi} \mu_1, \rho_f)  \mid \varphi'_1
                    \label{eq:fix1}
                \end{equation} 
                Let $\Sigma'$ equal any context such that $\Sigma' \supseteq
                \Sigma$ and the above typing holds.
                
                From this along with assumptions~\ref{ite:t-fix3} and \ref{ite:t-fix4}, we can
                construct the following typing derivation for $e'$:
                \begin{align*}
                    \inferrule*[right=t-fix]
                        {\eqref{eq:fix1}
                        \qquad 
                        \eqref{ite:t-fix3}
                        \qquad
                        \eqref{ite:t-fix4}
                        }
                        {K \mid \Gamma \mid \Sigma' \vdash \textbf{let}\ f = \textbf{fix}(f, (\Lambda \{ \alpha , \rho , \epsilon \} . \lambda x . e_1)\ [s]\ \textbf{at}\ e'_2)\ \textbf{in}\ e_3 : \mu_2 \mid \varphi'_1 \times \varphi_2 }
                \end{align*}

                From which we have that $\varphi' = \varphi'_1 \times \varphi_2$
                and from assumption~\ref{ite:t-fix5}, that $\varphi = \varphi_1 \varphi_2$. 
                Then, by induction we have that $\varphi'_1 \times \varphi_2 \sqsubseteq \varphi_1 \times \varphi_2$. 
                
                Finally, we can type the output store $\storeOuter'$ by induction:
                \begin{align*}
                    K \mid \Gamma \mid \Sigma' \vdash \storeOuter'
                \end{align*}

            \item Subcase (\textsc{e-fixL}): \\
                From (\textsc{e-fixL}) we have the following assumptions:
                \begin{enumerate}[a.]
                    \itemsep0em 
                    \item\label{ite:e-fixL1} $\storeInner_{\rho}^{in.} = \storeOuter(\rho)$
                    \item\label{ite:e-fixL2} $\Lambda\{\alpha, \rho, \epsilon\} . \lambda x . e_1 = \storeInner_{\rho}^{in.}(l_\rho)$
                    \item\label{ite:e-fixL3} $e' = [x \mapsto l_\rho] e_3$
                    \item\label{ite:e-fixL4} $\storeOuter' = \storeOuter$
                \end{enumerate}

                Since the evaluation rule does not take any additional steps, $\Sigma'$, $\varphi'_1$, and $\varphi'_2$ are 
                all equal to their original counterparts. Thus, using substitution (Lemma~\ref{lemma:subst}) with 
                assumptions~\ref{ite:t-fix3} and~\ref{ite:t-fix4}, we obtain:
                \begin{align*}
                    K \mid \Gamma, f : (\forall \{ \alpha , \rho , \epsilon  \} . (\alpha, \rho) \xrightarrow{\varphi'} \mu_1, \rho_f) \mid \Sigma' \vdash [x \mapsto l_\rho] e_3 : \mu_2 \mid \varphi'_1 \times \varphi'_2
                \end{align*}
                Since $\Sigma' = \Sigma, \varphi'_1 = \varphi_1$, and $\varphi'_2 = \varphi_2$, we have trivially that
                $\varphi'_1 \times \varphi'_2 \sqsubseteq \varphi_1 \times \varphi_2$.

                By assumption~\ref{ite:e-fixL4}, we have that the output store is unchanged. Thus 
                we can type $\storeOuter'$ by appealing to the typability of assumption~\ref{ite:t-fix1}.

        \end{itemize}

    \item Case (\textsc{t-if})\\
        From (\textsc{t-if}) we have the following assumptions: 
        \begin{enumerate}[i]
            \itemsep0em 
            \item\label{ite:t-if1} $e = \textbf{if}\ e_1\ \textbf{then}\ e_2\ \textbf{else}\ e_3$
            \item\label{ite:t-if2} $K \mid \Gamma \mid \Sigma \vdash e_1 : (\text{Bool}, \rho) \mid \varphi_1$
            \item\label{ite:t-if3} $K \mid \Gamma \mid \Sigma \vdash e_2 : \mu_2 \mid \varphi_2 $
            \item\label{ite:t-if4} $K \mid \Gamma \mid \Sigma \vdash e_3 : \mu_2 \mid \varphi_3$
            \item\label{ite:t-if5} $\varphi = \varphi_1 \times (\varphi_2 \sqcup \varphi_3)$
        \end{enumerate}
        There are three evaluation rules by which $\langle e \mid \storeOuter
        \rangle \longrightarrow \langle e' \mid \storeOuter' \rangle$ can be derived:
        (\textsc{e-if}), (\textsc{e-ifTrue}), and (\textsc{e-ifFalse}).
        \begin{itemize}
            \item Subcase (\textsc{e-if}): \\
                From (\textsc{e-if}) we have the following assumptions:
                \begin{enumerate}[a.]
                    \itemsep0em 
                    \item\label{ite:e-if1} $\eformat{e_1}{\storeOuter}{e'_1}{\storeOuter'}$
                    \item\label{ite:e-if2} $e' = \textbf{if}\ e'_1\ \textbf{then}\ e_2\ \textbf{else}\ e_3$
                \end{enumerate}

                By induction on assumption~\ref{ite:t-if2} we have: 
                \begin{equation}
                    K \mid \Gamma \mid \Sigma \vdash e'_1 : (\text{Bool}, \rho) \mid \varphi'_1
                    \label{eq:if1}
                \end{equation} 
                which we can combine with assumption~\ref{ite:t-if3}:
                \begin{equation}
                    K \mid \Gamma \mid \Sigma' \vdash e_2 : \mu_2 \mid \varphi_2
                    \label{eq:if2}
                \end{equation}
                and assumption~\ref{ite:t-if4}:
                \begin{equation}
                    K \mid \Gamma \mid \Sigma' \vdash e_3 : \mu_2 \mid \varphi_3
                    \label{eq:if3}
                \end{equation}
                Let $\Sigma'$ equal any context such that $\Sigma' \supseteq
                \Sigma$ and the above typings hold.
                From this, we can construct the following typing derivation for $e'$:
                \begin{align*}
                    \inferrule*[right=t-if]
                        {\eqref{eq:if1} \qquad \eqref{eq:if2} \qquad \eqref{eq:if3}}
                        {K \mid \Gamma \mid \Sigma' \vdash \textbf{if}\ e'_1\ \textbf{then}\ e_2\ \textbf{else}\ e_3 : (\tau, \rho) \mid \varphi_1' \times (\varphi_2 \sqcup \varphi_3)}
                \end{align*}

                From which we have that $\varphi' = \varphi'_1 \times (\varphi_2 \sqcup \varphi_3)$
                and from assumption~\ref{ite:t-if5}, that $\varphi = \varphi_1 \times (\varphi_2 \sqcup \varphi_3)$. 
                From the typability of assumption~\ref{ite:t-if2} we have that 
                $K \vdash \varphi_1 : \text{Effect}$~\hypertarget{if1}{(*)}, from the typability of assumption~\ref{ite:t-if3} we have that 
                $K \vdash \varphi_2 : \text{Effect}$~\hypertarget{if2}{(**)}, from the typability of assumption~\ref{ite:t-if4} we have that 
                $K \vdash \varphi_3 : \text{Effect}$~\hypertarget{if3}{(***)}, and by induction we have that $\varphi'_1 \sqsubseteq \varphi_1$~\hypertarget{if4}{(****)}.
                From this we can construct the following derivation for $\varphi'_1 \times (\varphi_2 \sqcup \varphi_3) \sqsubseteq \varphi_1 \times (\varphi_2 \sqcup \varphi_3)$:

                \begin{equation}
                            \inferrule*[right=sb-$\times_{below}^{1}$]{
                                    \texorpdfstring{\protect\hyperlink{if4}{\text{(****)}}}{}
                                    \qquad 
                                    \inferrule*[right=$\kappa$-meet]
                                        {
                                            \texorpdfstring{\protect\hyperlink{if2}{\text{(**)}}}{}
                                            \qquad 
                                            \texorpdfstring{\protect\hyperlink{if3}{\text{(***)}}}{}
                                        }
                                        {K \vdash \varphi_2 \sqcup \varphi_3 : \text{Effect}}
                                }
                                {K \vdash \varphi'_1 \sqsubseteq \varphi_1 \times (\varphi_2 \sqcup \varphi_3) : \text{Effect} }
                            \label{eq:if-phi-1}
                \end{equation}

                \begin{equation}
                    \inferrule*[right=sb-$\sqcup_{below}^{1}$]
                                                {
                                                    \inferrule*[right=sb-$\equiv$]
                                                        {\inferrule*[right=eq-refl]
                                                            {
                                                                \texorpdfstring{\protect\hyperlink{if2}{\text{(**)}}}{}
                                                            }
                                                            {K \vdash \varphi_2 \equiv \varphi_2 : \text{Effect}}}
                                                        {K \vdash \varphi_2 \sqsubseteq \varphi_2 : \text{Effect} }
                                                    \qquad
                                                        \texorpdfstring{\protect\hyperlink{if3}{\text{(***)}}}{}
                                                }
                                                {K \vdash \varphi_2 \sqsubseteq \varphi_2 \sqcup \varphi_3 : \text{Effect}} 
                    \label{eq:if-phi-4}
                \end{equation}

                \begin{equation} 
                                    \inferrule*[right=sb-$\times_{below}^{2}$]{
                                            \texorpdfstring{\protect\hyperlink{if1}{\text{(*)}}}{}
                                            \qquad 
                                            \eqref{eq:if-phi-4}
                                        }
                                        {K \vdash \varphi_2 \sqsubseteq \varphi_1 \times (\varphi_2 \sqcup \varphi_3) : \text{Effect}}
                    \label{eq:if-phi-3}
                \end{equation}

                \begin{equation}
                    \inferrule*[right=sb-$\times_{below}^{2}$]
                                        {
                                            \texorpdfstring{\protect\hyperlink{if1}{\text{(*)}}}{}
                                            \qquad 
                                            \inferrule*[right=sb-$\sqcup_{below}^{2}$]
                                                {
                                                    \texorpdfstring{\protect\hyperlink{if2}{\text{(**)}}}{}
                                                \qquad
                                                    \inferrule*[right=sb-$\equiv$]
                                                        {\inferrule*[right=eq-refl]
                                                            {
                                                                \texorpdfstring{\protect\hyperlink{if3}{\text{(***)}}}{}
                                                            }
                                                            {K \vdash \varphi_3 \equiv \varphi_3 : \text{Effect}}}
                                                        {K \vdash \varphi_3 \sqsubseteq \varphi_3 : \text{Effect}}
                                                }
                                                {K \vdash \varphi_3 \sqsubseteq \varphi_2 \sqcup \varphi_3 : \text{Effect}}
                                        }
                                        {K \vdash \varphi_3 \sqsubseteq \varphi_1 \times (\varphi_2 \sqcup \varphi_3) : \text{Effect}}               
                    \label{eq:if-phi-5}
                \end{equation}
                
                \begin{equation}
                            \inferrule*[right=sb-$\sqcup_{above}$]
                                    {
                                    \eqref{eq:if-phi-3}    
                                    \qquad 
                                   \eqref{eq:if-phi-5} 
                                }
                                {K \vdash \varphi_2 \sqcup \varphi_3 \sqsubseteq \varphi_1 \times (\varphi_2 \sqcup \varphi_3) : \text{Effect}}
                \label{eq:if-phi-2}
                \end{equation}

                \begin{align*}
                    \inferrule*[right=sb-$\times_{above}$]
                        {
                            \eqref{eq:if-phi-1} 
                            \qquad
                            \eqref{eq:if-phi-2}  
                        }
                        {K \vdash \varphi'_1 \times (\varphi_2 \sqcup \varphi_3) \sqsubseteq \varphi_1 \times (\varphi_2 \sqcup \varphi_3) : \text{Effect} }
                \end{align*}
                
                Finally, we can type the output store $\storeOuter'$ by induction:
                \begin{align*}
                    K \mid \Gamma \mid \Sigma' \vdash \storeOuter'
                \end{align*}

            \item Subcase (\textsc{e-ifTrue}): \\
                From (\textsc{e-ifTrue}) we have the following assumptions:
                \begin{enumerate}[a.]
                    \itemsep0em 
                    \item\label{ite:e-ifTrue1} $\storeInner_{\rho}^{in.} = \storeOuter(\rho)$
                    \item\label{ite:e-ifTrue2} $\text{true} = \storeInner_{\rho}^{in.}(l_\rho)$
                    \item\label{ite:e-ifTrue3} $e' = e_2$
                    \item\label{ite:e-ifTrue4} $\storeOuter' = \storeOuter$
                \end{enumerate}

                Since the evaluation rule does not take any additional steps, $\Sigma', \varphi'_1, \varphi'_2$, and $\varphi'_3$ are 
                all equal to their original counterparts. Thus, the evaluation simply steps to $e_2$, which is
                well-typed by assumption~\ref{ite:t-if3}:             
                \begin{align*}
                    K \mid \Gamma \mid \Sigma \vdash e_2 : \mu_2 \mid \varphi_2 
                \end{align*}
                Since $\Sigma' = \Sigma, \varphi'_1 = \varphi_1, \varphi'_2 = \varphi_2$, and $\varphi'_3 = \varphi_3$, we have trivially that
                $\varphi'_2 \sqsubseteq \varphi_2$.

                By assumption~\ref{ite:e-ifTrue4}, we have that the output store is unchanged. Thus 
                we can type $\storeOuter'$ by appealing to the typability of assumption~\ref{ite:t-if2}.

            \item Subcase (\textsc{e-ifFalse}): \\
                From (\textsc{e-ifFalse}) we have the following assumptions:
                \begin{enumerate}[a.]
                    \itemsep0em 
                    \item\label{ite:e-ifFalse1} $\storeInner_{\rho}^{in.} = \storeOuter(\rho)$
                    \item\label{ite:e-ifFalse2} $\text{false} = \storeInner_{\rho}^{in.}(l_\rho)$
                    \item\label{ite:e-ifFalse3} $e' = e_3$
                    \item\label{ite:e-ifFalse4} $\storeOuter' = \storeOuter$
                \end{enumerate}

                Since the evaluation rule does not take any additional steps, $\Sigma', \varphi'_1, \varphi'_2$, and $\varphi'_3$ are 
                all equal to their original counterparts. Thus, the evaluation simply steps to $e_3$, which is
                well-typed by assumption~\ref{ite:t-if4}:             
                \begin{align*}
                    K \mid \Gamma \mid \Sigma \vdash e_3 : \mu_2 \mid \varphi_3
                \end{align*}
                Since $\Sigma' = \Sigma, \varphi'_1 = \varphi_1, \varphi'_2 = \varphi_2$, and $\varphi'_3 = \varphi_3$, we have trivially that
                $\varphi'_3 \sqsubseteq \varphi_3$.

                By assumption~\ref{ite:e-ifFalse4}, we have that the output store is unchanged. Thus 
                we can type $\storeOuter'$ by appealing to the typability of assumption~\ref{ite:t-if2}.
        \end{itemize}

\end{itemize}
\end{proof}


\end{document}